\documentclass[12pt, a4paper]{article}

\usepackage{math tools,math,mathrsfs,mathtools}
\usepackage[T1]{fontenc}
\usepackage[latin9]{inputenc}
\usepackage{fullpage}
\usepackage{amsmath,amssymb,amsfonts,amsthm}
\usepackage{natbib,float,enumerate,tabularx,booktabs,url,enumitem}
\usepackage[flushleft]{threeparttable}
\usepackage{bbm}
\usepackage{bigints}
\usepackage[onehalfspacing]{setspace}
\usepackage{bm}

\usepackage{pdfpages}
\usepackage{lscape}
\usepackage{microtype}
\pdfpagewidth=\paperwidth
\pdfpageheight=\paperheight
\usepackage{ragged2e}
\usepackage{geometry}
\usepackage{comment}
\usepackage[multiple]{footmisc}
\usepackage{tabularx}
\usepackage{dcolumn}
\usepackage{multirow}
\newcommand\Tstrut{\rule{0pt}{0.5cm}}   

\allowdisplaybreaks

\makeatletter
\newcommand*{\addFileDependency}[1]{
  \typeout{(#1)}
  \@addtofilelist{#1}
  \IfFileExists{#1}{}{\typeout{No file #1.}}
}
\makeatother

\usepackage{multicol}
\usepackage{tikz}
\usetikzlibrary{fit,shapes.geometric,arrows.meta, arrows, positioning}
\usepackage{graphicx}
\usepackage{caption}
\usepackage{subcaption}

\DeclareMathOperator*{\argmin}{arg\,min}

\newtheoremstyle{newthmstyle}%
  {\topsep}
{\topsep}
{\itshape}
{0pt}
{\bfseries\color{blue}}
{. }
{0pt}
{}
\theoremstyle{newthmstyle}

\newtheorem{lemma}{{\color{blue}Lemma}}
\newtheorem{definition}{Definition}

\newtheorem{remark}{{\color{blue}Remark}}
\newtheorem{prop}{{\color{blue}Proposition}}
\newtheorem{thm}{{\color{blue}Theorem}}

\newtheorem{appprop}{{\color{blue}Proposition}}
\newtheorem{applemma}{{\color{blue}Lemma}}

\theoremstyle{definition}

\newif\ifold
\oldfalse 

\newif\ifcalc
\calctrue 

\begin{document}
\title{\vspace{-0.0cm}\Large Specialization, Complexity, and Resilience in Supply Chains\footnote{We particularly thank Pol Antr{\`a}s, Vasco Carvalho, Matteo Escud\'e, Mathieu Parenti, and Alireza Tahbaz-Salehi for many discussions that have been instrumental to the inception of this paper. We also thank  David Berger, Jan Eeckhout, Matt Elliott, Gene Grossman, David H\'emous, Cosmin Ilut, Gregor Jarosch, Yasutaka Koike-Mori, Andrea Lanteri, Jennifer La'O, Fabrizio Leone (discussant), Hugo Lhuillier, Salvatore Lo Bello, Rocco Macchiavello, Isabelle Mejean, Paolo Martellini, Guido Menzio, Marek Pycia, Stan Rabinovich, Alejandro Sabal, Edouard Schaal, Armin Schmutzler, Mathieu Taschereau-Dumouchel, Farid Toubal, Venky Venkateswaran, and Daniel Xu for their feedback. This paper also benefited from seminars at the University of Zurich, PSE and Sciences-Po, Duke, UNC, Penn State, NYU, Columbia, Cambridge and UC Louvain as well as presentations at the ``Workshop on International Economic Networks'' in Vienna, ``Trade, Value Chains and Financial Linkages in the Global Economy'' conference at the Bank of Italy, the CEPR - Naples Trade \& Development Workshop, the Barcelona Summer Forum Workshop on ``Networks and Firm Dynamics in Macro and Trade'',  SED and the GEN Workshop. Alessandro Ferrari gratefully acknowledges financial support from the Swiss National Science Foundation (grant number 100018-215543) and from the Severo Ochoa Programme for Centres of Excellence in R\&D (Barcelona School of Economics CEX2024-001476-S). The views expressed in this paper are those of the authors and do not necessarily reflect the views of the Bank of Italy.}}
\author{\color{blue}\large{Alessandro Ferrari}\color{black}
\\ \small UPF, CREi, BSE \& CEPR \and \color{blue}\large{Lorenzo Pesaresi}\color{black}
\\ \small  Bank of Italy}
\date{\today}


\maketitle
\vspace{-30pt}
\begin{abstract} 
    We study how product specialization choices affect supply chain resilience. We propose a theory of supply chain formation in which only compatible inputs can be used in final production. Intermediate producers choose how much to specialize their goods, trading off higher value added against a smaller pool of compatible final producers. Final producers operate complex supply chains, requiring multiple complementary inputs. Specialization choices determine how quickly final producers can replace suppliers after disruptions, and thus supply chain resilience. In equilibrium, production inputs are over-specialized due to a novel \textit{network externality}. Intermediate producers fail to internalize how their specialization choices affect the likelihood that final producers source all required inputs, and therefore the lost value added from complementary inputs if production halts. As a result, supply chains are more productive in normal times but less resilient than socially desirable. We characterize the optimal transfer that restores the efficient allocation and show that non-fiscal interventions, such as compatibility standards, are generally welfare-enhancing.


\end{abstract}

\begin{raggedright} Keywords: Supply chains, specialization, product design, resilience 
\end{raggedright}
\newpage

\section{Introduction}
In modern economies, production is organized around long and complex supply chains. Advances in information and transportation technologies enable firms to source specialized inputs from a large number of suppliers. While supply chains bring about substantial productivity gains relative to in-house production, concerns are mounting about the frequency and economic costs of their disruptions. Often, a disruption to a single supplier creates shortages of key inputs, temporarily halting production altogether. These halts impose consumption losses, which have sparked policy interest in the notion of supply chain \textit{resilience}: the speed at which final production recovers following a disruption.\footnote{In June 2021, the Biden-Harris Administration instituted the \textit{Supply Chain Disruption Task Force} ``to provide a whole-of-government response to address near-term supply chain challenges to the economic recovery'' (see \cite{whitehouse}).}  Despite growing policy interest, the determinants of supply chain resilience are still not well understood.

In this paper, we develop a new theory of supply chain formation centered on three key elements: (i) complementarities across inputs, allowing for production halts following single-supplier disruptions, (ii) compatibility frictions, reflecting difficulties in replacing specialized inputs, and (iii) search frictions, capturing the time required for input procurement. Our framework is built around a technological constraint: inputs must be \textit{compatible} with the production processes in which they are used. Intermediate producers choose the degree of specialization of their goods. Higher specialization increases the value added of the good but reduces the pool of compatible final producers. Final producers operate \textit{complex} supply chains, defined as  production processes requiring multiple complementary inputs. Specialization
choices determine how quickly final producers can replace suppliers after disruptions, and
thus supply chain resilience. 
This framework allows us to analyze the sources of inefficiency affecting supply chain resilience and address normative questions related to the optimal organization of supply chains. 
In our economy, higher input specialization increases the value of production in normal times, but results in a higher cost of disruptions, as replacing specialized inputs takes longer.
Endogenous specialization choices thus shape the equilibrium balance between productivity and resilience.

We start by characterizing the equilibrium in a static economy where heterogeneous intermediate producers choose the degree of specialization of their goods. The specialization decision is modeled as in  \cite{Menzio2023_ProductDesign}, where
optimal specialization solves a trade-off between the value added of the good and its compatibility with prospective buyers.
The notion of compatibility is stochastic: higher specialization reduces the probability that an intermediate good is compatible with a final producer's technology. Ex-ante identical final producers must source multiple complementary inputs in the market to produce.
The markets for intermediate goods display realistic sourcing frictions. Each final producer contacts a finite number of intermediate producers, governed by the degree of search frictions. Within this set, the final producer selects the suppliers which are compatible with its production process -- if any exist -- and requests a quote from each of them.\footnote{According to a McKinsey report, "\textit{At most organizations [\dots], it takes about three months to complete a single supplier search, with a sourcing professional logging more than 40 hours of work -- and yet able to consider only a few dozen suppliers from a total population of thousands.}" \citep{McKinsey2021suppliers}. We interpret this evidence as indicative of search and compatibility frictions in intermediate markets.} Intermediate producers respond by submitting sealed bids characterized by the surplus (profits) offered to the final producer. The final producer trades with the intermediate firm offering the highest surplus, provided that it exceeds a reservation surplus.
Whenever the procurement search is unsuccessful for any of the complementary inputs, the entire production process halts. 

In this static framework, we highlight the key inefficiency arising from endogenous specialization when production is complex, i.e., multiple inputs are required to produce a final good. Intermediate producers choose their optimal specialization by trading off the private surplus they obtain from each transaction against the probability of trading. Crucially, they fail to internalize how their specialization choices affect the likelihood that final producers source all the required inputs, and therefore the unrealized value added from complementary
inputs if production halts.  

For example, consider electric vehicle (EV) production. All EVs require a battery, but not all battery types are equally compatible with the production of each model. High-end EV models are more compatible with type-A batteries than with type-B batteries, whereas type-C batteries are generally incompatible.\footnote{Concretely, type-A batteries are lithium nickel cobalt aluminium (LNCA), type-B batteries are lithium iron phosphate (LFP), and type-C batteries are nickel-metal hydride (NiMH).} Hybrid EVs are incompatible with type-A  batteries and are more compatible with type-C batteries than with type-B batteries. The higher the compatibility of the battery used, the higher the value of the EV model for final consumers. Despite their critical role, batteries are typically sourced from external suppliers. Similarly, each EV model requires other externally sourced components, such as tires, cameras, and carbon fiber.
Due to compatibility frictions, an EV producer may, in a given period, only locate battery suppliers offering an incompatible type, in which case production halts entirely. This mismatch arises from suppliers' specialization decisions: type-B batteries, though of lower value added, are compatible with most models, whereas type-A and type-C batteries are restricted to high-end applications and hybrid EVs, respectively. Our theory predicts that battery manufacturers, when choosing their degree of specialization, do not account for the lost value added from complementary inputs when car makers cannot source compatible batteries. If EV production halts, suppliers of components such as tires and cameras also fail to generate value added.
This example illustrates a \textit{network externality}. We show that this negative externality, arising from firms' specialization decisions, induces equilibrium \textit{over-specialization} relative to the constrained-efficient allocation. 
Importantly, specialization choices interact with production complexity: if EVs were composed only of batteries -- that is, if the production process were not complex -- the network externality would vanish, and the equilibrium would be constrained-efficient.

The network externality arises from the combination of three elements: (i) search and compatibility frictions, which prevent final producers from finding required inputs with certainty, (ii) endogenous specialization decisions, which determine the equilibrium level of compatibility, and (iii) complementarity among inputs, so that failure to source even a single input halts production. To the best of our knowledge, our model is the first to jointly incorporate these three elements.

Next, we consider a dynamic version of our model to account for the long-term nature of supplier relations and study supply chain resilience. Once formed, supplier relations last until the final producer -- or a supplier link -- experiences a disruption, which happens stochastically.
Long-term relations introduce an additional motive for specialization: tailoring products to the needs of existing customers (\textit{customization}). 

This dynamic framework delivers a novel insight. Different from the existing literature on supply chain fragility, the \textit{weakest link} in our model is not the input with the highest disruption probability. Rather, it is the input for which the joint probability of supplier disruption and failed replacement is greatest. Since higher disruption probability reduces the return from specialization, the weakest link may even be an input with relatively low disruption probability, but that is very difficult to replace when missing, due to high equilibrium specialization.

From a normative perspective, in addition to the network externality, endogenous specialization choices in a dynamic context introduce a \textit{search externality}  \citep{pissarides2000}: intermediate producers do not internalize the effect of their specialization on their probability of meeting a buyer, as more specialized goods raise the share of searching final producers.
This externality pushes the equilibrium towards inefficiently low specialization. 
However, when production is complex, the network externality dominates the search externality, so the equilibrium continues to exhibit over-specialization relative to the planner's allocation.

Importantly, the dynamic framework provides a precise and welfare-relevant definition of supply chain resilience: the probability that final production is restored in a given period following a disruption. This theoretical definition of resilience closely aligns with the notion of ``ability to return to normal operations over an
acceptable period of time, post-disruption'' commonly used among risk management practitioners \citep{miroudot2020resilience}. 
The model also allows us to characterize the determinants of supply chain resilience, namely the degree of search frictions, average specialization, disruption frequency, and production complexity. It follows that equilibrium over-specialization implies that supply chains are \textit{less resilient than efficient}. We view this result as evidence that policy concern for supply chain resilience is warranted, given that the latter is both welfare-relevant and inefficient in our economy.

In our model, the resilience and the average productivity of supply chains jointly determine output and welfare. As a result, endogenous specialization choices shape the equilibrium balance between productivity and resilience.
Supply chains relying on more specialized inputs are more productive in normal times, when they function smoothly.  However, when a disruption occurs, higher specialization makes it harder to restore production in a timely fashion. Due to the network externality,  the equilibrium allocation tilts excessively toward productivity. Consequently, a social planner would optimally forgo some surplus in normal times in exchange for more resilient supply chains. We show that our conclusions are robust to allowing final producers to bargain over the surplus, direct their search within each input market, hold inventories, source the same input from multiple suppliers, produce some of the inputs in-house, and write contracts in which suppliers' remuneration is not tied to production. 


We conclude by studying the normative implications of our framework. We show that a planner can decentralize the constrained-efficient allocation through a lump-sum transfer to \textit{active} final producers in both the static and dynamic economies. This transfer reduces the reservation surplus -- or, equivalently, raises the reservation price -- of final producers. Under price posting, a higher reservation price leads intermediate producers to optimally post higher prices. However, the increase in posted prices is heterogeneous across intermediate producers, depending on their position in the productivity ladder: lower productivity firms respond more strongly to changes in the reservation price. Hence, a lump-sum transfer to final producers induces a heterogeneous increase in profits per transaction for intermediate producers, as if the policy were targeted towards low-productivity firms that are farther from their efficient specialization.
As a result, intermediate producers have more \textit{skin in the (final production) game} and reduce their privately optimal specialization to the efficient level.
Since the optimal policy requires taxing households to subsidize profit-making firms, it might be politically infeasible. Therefore, we also consider a non-fiscal, regulatory policy instrument: the introduction of a \textit{standard} or minimum interoperability requirement.\footnote{This kind of policy is common in international regulation, as exemplified by the recent EU Directive on USB-C chargers.} We show that standards are welfare-improving and can be a powerful and implementable tool to address over-specialization. 

\paragraph{Related Literature}
This paper is related to several strands of literature. First, our interest in supply chains speaks to the growing literature on production network formation and its role in the economy's shock-absorption capacity. \cite{levine2012production}, \cite{baqaee2018cascading}, \cite{Elliott2022_Fragility}, \cite{acemoglu2024macroeconomics}, \cite{carvalho_bottlenecks}, and \cite{capponi2024resilience} study the \textit{fragility} of economies in which supply chains feature coordination problems due to complementarities, public goods problems, or market incompleteness. \cite{carvalho2021supply}, \cite{alessandria2023aggregate}, \cite{ferrari2022inventories}, \cite{carreras2021increasing}, and \cite{carreras_ferrari} study the dynamic response of supply chains to disruptions. \cite{arkolakis2010penetration}, \cite{eaton2022twosided},  \cite{mejean2022search}, \cite{fontaine2023frictions}, \cite{mejean2023RS}, \cite{esala2024search}, \cite{huang2022heterogeneity}, \cite{bai2024causal}, \cite{demir2024ring}, and \cite{arkolakis2025production} study the implications of search and matching frictions in firm-to-firm trade.
Closest to our work, \cite{grossman2023resilience} explores market failures in vertical supply chains where firms can invest to reduce their risk of disruptions, and the terms of trade in firm-to-firm transactions are bargained over. Existing studies typically use static models where the fragility of supply chains is only shaped by their \textit{robustness} to shocks -- the ability to maintain operations in the face of a shock.
In contrast, we propose a dynamic framework that allows us to disentangle robustness from \textit{resilience} -- the ability to restore production after a shock.

Our work is also closely related to the literature on the efficient number of varieties, whose foundations date back to \cite{spence1976product}, \cite{dixit1977monopolistic}, and \cite{mankiw1986free}. Recent advances of these traditional models include \cite{zhelobodko2012monopolistic}, \cite{parenti2017toward}, \cite{dhingra2019monopolistic}, and \cite{bajoetal}. \cite{grossman2023supply} proposes an application to supply chains. We show that many of the insights from this literature have tight analogs in our search model with endogenous specialization choices and price posting, chief among them the offsetting effects of \textit{appropriability} and \textit{business-stealing} externalities.

Both methodologically and conceptually, our work is close to the strands of literature on product design \citep{kiyotakiwright1993monetary,bar2012search,bar2023targeted,Menzio2023_ProductDesign,albrecht2023vertical}, supply chain formation \citep{oberfield2018IO,antras2020geography,boehm2020misallocation,kopytov2021endogenous,kim2023supplychain,antras2023austrian}, and specialization \citep{rauch1999networks,nunn2007relationship, barrot2016input, mejean2023RS, Ekerdt,Chen2025compatibility}. Relative to these contributions, we study the effects of specialization choices on the resilience of supply chains. 
Methodologically, the paper most closely related to ours is \cite{Menzio2023_ProductDesign}, which studies the endogenous response of product specialization to declining search frictions in the presence of heterogeneous consumer preferences, and derives conditions under which price dispersion and competition remain constant as search frictions vanish.
We apply the specialization problem of \cite{Menzio2023_ProductDesign} to a production network context, in which buyers are not consumers with heterogeneous preferences but final producers with idiosyncratic technological requirements and a complex production function.
Our key novelty is to study product design choices when goods are complements in production. This allows us to highlight inefficiencies associated with supply chain formation that are absent in final-consumer markets. Our model nests that of \cite{Menzio2023_ProductDesign} when final production is not complex, that is, when only a single input is required to produce a final good.

Finally, our model is close in spirit to the large literature on the hold-up problem and investment incentives in firm-to-firm transactions \citep[see][]{klein1978vertical,williamson1979transaction,rogerson1992,hartmoore,tirole,cantillon2004procurement}, as well as on the organization of production \citep{antras2003firms,antras2006contractual,antras2013organizing, boehm2022impact,boehm2024network}.
We study a context in which firms make sunk investments in specialization and cannot appropriate the full returns from their investment, yet our price-posting mechanism generates bilaterally efficient outcomes in frictional markets with simultaneous search. 

\section{Static Model}\label{static_sec}
In this section, we consider the sourcing problem of a final producer that needs to source its production inputs anew from the market. Focusing on the sourcing problem of a final producer allows us to single out the specific stage at which the specialization of intermediate inputs plays a role in driving the supply chain dynamics. As highlighted in the business literature, the more specialized intermediate inputs are, the more difficult it is for final producers to source a compatible input \citep{miroudot2020resilience,Sanz2023resilience}. We are interested in understanding how intermediate producers choose the specialization of their inputs and whether such specialization decisions are efficient. To this end, we propose a static model of sourcing from frictional markets when intermediate producers make endogenous product design decisions.

\subsection{Environment}
The economy is populated by a representative household, a measure $1$ of perfectly competitive final producers and a measure $m$ of intermediate producers. Final producers differ in technological requirements, whereas intermediate producers are heterogeneous in their productivity $z$. The market for the consumption good is perfectly competitive, whereas the markets for
intermediate goods are frictional.
The consumption good is the numeraire. The model is static.

\subparagraph{Representative household.}
The representative household has linear preferences over an aggregate consumption good and convex disutility from hours worked. Aggregate consumption is a linear aggregator of quality-weighted consumption goods supplied by final producers, while disutility from labor is logarithmic, scaled by a parameter $\psi>0$. 
Members of the representative household choose whether to purchase the indivisible consumption good offered by each final producer $i$ ($C_i$), with associated quality $\mathcal{Q}_i$ at price $P_i$, and how many hours $\ell$ to supply at each wage rate $w$. Formally, they solve the following utility maximization problem:
\begin{align} \label{ump}
        \max_{\substack{
C_i \in \{0,1\},\; i \in [0,1] \\
 \ell \in [0,1]}} \ \mathcal{U} &=  C + \psi \log\left(1-\ell\right)
        &\quad \text{s.t.} \ \  & \int_0^1 P_i C_i \ di = w \ell + \overline{\Pi}, \\ \nonumber
        && &C = \int_0^1\mathcal{Q}_i C_i \ di,    
\end{align}   
where $\overline{\Pi}$ represents profit rebates.

\subparagraph{Final producers.}
    Each final producer $i$ operates a Leontief production function combining $N>1$ complementary and indivisible inputs into one unit of consumption good:
\begin{align} \label{def:prod_fnct}
Y_i = \mathbbm{1}\{\min\{y_1,\dots,y_N\}=1\}.
\end{align}
Following \cite{Elliott2022_Fragility}, we refer to the number of inputs required for final production, $N$, as the \textit{complexity} of the production process. 
  We think of the $N$ inputs as defining a recipe for final production in the spirit of \cite{oberfield2018IO}. We assume that these inputs cannot be produced in-house and must be sourced from a frictional market.\footnote{We extend the model to accommodate in-house production in Section \ref{sec:discussion}. We consider the optimal choice of $N$ in an extension in Appendix \ref{app:end_N}.}
  
  Final producers are ex ante identical but have idiosyncratic technological requirements: not all intermediate input varieties are compatible with their production process. Formally, each variety $v$ of input $j$ is compatible with the production process of a final producer $i$ with probability $\phi_{v(j)} \in (0,1)$. If compatible, the input variety can contribute positive match surplus to final producer $i$, $A_{v(j)} > 0$. Otherwise, the variety cannot be used in its production process, as in \cite{kiyotakiwright1993monetary}. 
  Because $\phi_{v(j)}$ is independent of $i$, final producers are ex ante identical. However, the realized technological compatibility is stochastic: a given variety $v$ of input $j$ may end up being compatible with final producer $i$, but not with final producer $i^\prime$. 
  We interpret heterogeneity in realized compatibilities as reflecting idiosyncratic technological requirements.

Output quality $\mathcal{Q}_i$ is an additively separable function of the match surpluses associated with each input $j$ that the final producer is able to source:\footnote{In Appendix \ref{app:gen_quality} we develop the model with a general quality function.}
\begin{align} \label{def:quality_fnct}
    \mathcal{Q}_i = \sum_{j=1}^N A_j,
\end{align}
where the match surplus $A_j$ equals the value added of input $j$ in final production.
In turn, the higher the match surplus with the inputs sourced, the higher the quality of the consumption good. Notice that higher complexity of the production process increases the quality of the consumption good. This captures in reduced form the gains from organizing production around supply chains, resulting from technological complementarities across different inputs.

For example, suppose a car maker needs to source a battery and a tire to assemble an electric SUV. Yet, not all types of batteries and tires it can find in the market fit its production process equally well.  
A type-A battery would be a better fit than a type-B battery, whereas a type-C battery would be incompatible. Similarly, a large tire would fit the production process better than a small tire. 
 The quality function (\ref{def:quality_fnct}) assumes that the match surplus with different (compatible) inputs is perfectly substitutable, even though the inputs themselves are perfect complements. Keeping with our example, a producer could make a $\$50,000$-worth SUV with a type-A battery and a small tire, or with a type-B battery and a large tire. 
 
Each final producer $i$ obtains profits:
\begin{align} \label{final_profits}
   \pi_i = \left(P_i-\sum_{j=1}^N p_j\right)\mathbbm{1}\{Y_i=1\}.
\end{align}
Profits equal the difference between the price of the consumption good and the sum of the prices paid for the $N$ intermediate inputs sourced. If the final producer is unable to source all the $N$ complementary inputs, production halts, as implied by (\ref{def:prod_fnct}).

\subparagraph{Intermediate market structure.}
Intermediate good markets are characterized by search and compatibility frictions. Search frictions are represented by a finite number $n \sim \text{Poisson}(\lambda)$ of intermediate producers that each final producer is able to meet, where $n \in \mathbb{N}$ and $\mathbb{E}[n]=\lambda$. Hence, $\lambda$ represents the \textit{average number of meetings} per final producer.
Compatibility frictions are represented by an endogenous distribution of compatibility probability with mean $\bar{\phi}$. Hence, $\bar\phi$ represents the \textit{average probability that a potential supplier is compatible} with the final production process.
As a result, the number of compatible intermediate producers met by a final producer is Poisson-distributed with mean $\lambda \bar{\phi}$.
\footnote{For this result to hold exactly, we assume throughout that the variance of $\phi$ is small relative to its mean ($\bar\phi$) \citep{Menzio2023_ProductDesign}.}
Therefore, $\lambda\bar\phi$ represents the \textit{average number of compatible potential suppliers} per final producer.
It follows that the probability that a final producer finds a compatible input, or \textit{input finding probability}, is given by: 
\begin{align} \label{finding_prob}
    f=1-\exp\{-\lambda \bar{\phi}\}.
\end{align}
Once the compatible intermediate producers -- if any exist -- are selected, final producers request a quote from them. Intermediate producers submit the quote as a sealed bid characterized by a surplus $x$ offered to the final producer, as in a first-price auction with an unknown number of bidders. This represents the amount of surplus the final producer extracts from the transaction. We denote the endogenous distribution of surplus offered by $G(x)$. Final producers choose the intermediate producer offering the highest surplus.  In Appendix \ref{app:x_interpretation}, we show that this procurement protocol (sealed-bid auction) is equivalent to surplus posting by intermediate producers, as in \cite{Menzio2023_ProductDesign}.


\subparagraph{Intermediate producers.}  
Intermediate producers differ in productivity $z \sim \Gamma(z)$, where $\Gamma$ is a continuous distribution with support $[\underline{z},\bar{z}]$ and density $\gamma(z)$. 
Each producer can manufacture one variety of each input through distinct product lines.\footnote{Since intermediate producers obtain positive expected profits in equilibrium in each product line, they optimally choose to compete in all $N$ inputs.}
Intermediate producers make endogenous product design decisions by choosing the specialization $s$ of their input variety, as well as a surplus $x$ offered to final producers. Let the match surplus from the input variety with compatible final producers be an increasing and additively separable function of the degree of specialization and productivity: $A = A(s;z)$.
On the one hand, higher specialization increases the match surplus from the input variety with compatible final producers. Formally, $ A^\prime(s)>0, \ A^{\prime\prime}(s) <0$.  On the other hand, higher specialization reduces the share of compatible final producers. Formally,
$\phi=\phi(s)$, $\phi^\prime(s)<0$,  $\phi^{\prime\prime}(s)>0$. Optimal specialization trades off a higher match surplus from the input variety against a lower share of compatible final producers.

For each intermediate good, market tightness is defined as the ratio of the measure of final producers to the measure of intermediate producers, $\theta=\frac{1}{m}$. An intermediate producer meets an expected number $m_k$ of final producers who are in contact with $k=0,1,2,\dots$ other \textit{compatible} intermediate producers, where:
\begin{align*}
    %
      m_k =& \  \theta \sum_{n=k+1}^\infty n \ \underbrace{\vphantom{\frac{(n-1)!}{k! (n-k-1)!} (\bar{\phi})^k (1-\bar{\phi})^{n-k-1}}\frac{\lambda^n e^{-\lambda}}{n!}}_{\text{Prob. buyer meets n sellers}} \underbrace{\frac{(n-1)!}{k! (n-k-1)!} (\bar\phi)^k (1-\bar\phi)^{n-k-1}}_{\text{Prob. k other sellers are compatible}} = \ \theta \lambda  \ e^{-\lambda\bar{\phi}} \frac{(\lambda\bar{\phi})^k}{k!},
\end{align*}
This statistic summarizes the expected number of final customers \textit{binned} by how many other compatible intermediate competitors they are in contact with. The final equality follows from the Poisson approximation of the binomial distribution.
For example, an intermediate producer of good $j$ may expect to meet 3 final producers who each are in contact with 4 other compatible intermediate producers of good $j$ (on top of itself): this amounts to $m_4=3$. 

The expected operating profits of an intermediate producer with productivity $z$ choosing specialization $s$ and offering surplus $x$ to final producers equal:
\begin{align*}
\Pi(s,x;z,N) =& \ \left(\sum_{k=0}^\infty m_k \ \phi(s)  \ G(x)^k\right)f^{N-1}\left[A(s;z)-x\right] 
= \ \underbrace{\vphantom{\left[A(s;z)-x(z)\right]}\theta  \lambda}_{\substack{\text{\vphantom{\big|}exp \#} \\ \text{meetings}}} \underbrace{\mathcal{P}(s,x;N)}_{\substack{\text{trading prob.\big|} \\ \text{meeting}}} \underbrace{\vphantom{\mathcal{P}(s,x)}\left[A(s;z)-x\right]}_{\substack{\text{\vphantom{\big|} unit profit} \\ \equiv \  p(s,x;z)}}. 
\end{align*}
Expected operating profits are composed of three elements. First,
every intermediate producer meets an expected number $\theta\lambda$ of final producers, as governed by the extent of search frictions. Second, an intermediate producer choosing specialization $s$ and offering surplus $x$ trades with the final producers it meets with probability $\mathcal{P}(s,x;N)$. Finally, an intermediate producer with productivity $z$ and specialization $s$ offering surplus $x$ makes a unit profit of $p(s,x;z)$ on each completed transaction, where $p(s,x;z) \equiv A(s;z)-x$. The match surplus $A(s;z)$, which determines the quality of the consumption good, equals the value added created by the match between a final producer and an intermediate producer with productivity $z$ and specialization $s$. In turn, $x$ represents the surplus accruing to the final producer.

The endogenous compatibility frictions described in the previous paragraph are subsumed into the conditional trading probability $\mathcal{P}(s,x;N)$, which is itself a composite of three elements:
\begin{align*}
   \mathcal{P}(s,x;N) \equiv {\underbrace{\vphantom{\exp \left\{-\lambda \bar{\phi}\left[1-G(x(s,p))\right]\right\}} \phi(s)}_{\substack{\text{prob. compatible}}}} \ 
   {\underbrace{\exp \left\{-\lambda \bar{\phi}\left[1-G(x)\right]\right\}}_{\substack{\text{prob. best among compatible} \\ \text{contacted suppliers}}}} \ 
   { \underbrace{\vphantom{\exp \left\{-\lambda \bar{\phi}\left[1-G(x)\right]\right\}} f^{N-1}.}_{\substack{\text{prob. other complementary } \\ \text{inputs are sourced}}}}
\end{align*}
Suppose an intermediate producer with specialization $s$, offering surplus $x$ meets a final producer. First, its input variety is compatible with the production process of the final producer with probability $\phi(s)$. Second, since the number of compatible intermediate producers per final producer is Poisson-distributed, the probability that the surplus offered $x$ is the highest among the set of compatible intermediate producers contacted by the final producer is given by $\exp\{-\lambda\bar\phi(1-G(x))\}$.\footnote{This matching process is similar to the one proposed by \cite{huang2022heterogeneity}. The key difference is that, in our setting, intermediate producers can influence the share of potential buyers via endogenous product design decisions.} Finally, the last term states that trade only occurs if all other $N-1$ inputs are also successfully sourced. This happens with probability $f^{N-1}$. Hence, the trading probability of an intermediate producer of given production input depends on the input finding probabilities of all its complementary inputs.


According to this market structure, we can specify the product design problem for intermediate inputs. Let $q(s), \, q^\prime(s)>0, \, q^{\prime\prime}(s)>0$ denote the labor requirement for specialization. The intermediate producer's problem reads:
\begin{align} \label{static:pmp}
    \max_{s,x} V(s,x;z,N) =  \Pi(s,x;z,N)-wq(s).
\end{align}
The intermediate producer chooses specialization and surplus to maximize the value of the firm, given by the expected operating profits minus the cost of specialization. The key trade-off it faces is that more specialized input varieties generate more surplus conditional on trading, but a lower likelihood of trading due to higher compatibility frictions. Convex labor requirements ensure an interior solution to optimal specialization.

\paragraph{Discussion} We conclude this section with a brief discussion of some important assumptions in our model.

We model the final production function as a Leontief aggregator of indivisible inputs. This is supported by evidence on low elasticity of substitution across production inputs at the firm level \citep{barrot2016input,boehm2019propagation}. 
Moreover, we show in Appendix \ref{app:ces_prod_function} that our final production setting may arise from a more general class of models. In particular, our setting is observationally equivalent to models in which the consumption good is homogeneous and produced using a constant-elasticity-of-substitution (CES) aggregator with complementary inputs and integer constraints.\footnote{In that context, sourcing inputs of higher value added would raise output quantity for given price -- rather than output price for given quantity. We further relax the assumption of a linear quality function in Appendix \ref{app:gen_quality}, where we allow for arbitrary substitutability or complementarity in the quality function.} The key insight is that, whenever the production function features \textit{any} form of complementarity across inputs, the inability to source even a single input halts the entire production process.


 We model search as random and simultaneous within narrowly defined markets. In other words, final producers can screen only a random subset of firms producing each desired input. This is in line with evidence on supplier replacement \citep{miyauchi2024matching}. Unlike matching models of the labor market, trading probabilities are not governed by a reduced-form, market-level matching function, but by micro-founded compatibility frictions. This reflects the role of technological compatibility -- rather than imperfect information or geographical distance -- as the key friction in supply chain formation. 
 The search behavior of final producers mirrors common procurement practices, in which firms request quotes from multiple potential suppliers and select the best quality-adjusted bid. 
 In Section \ref{sec:discussion}, we show that our results are robust to allowing final producers to direct their search within input markets, and that our price-posting protocol corresponds to the efficient benchmark of a more general bargaining framework.

\subsection{Equilibrium characterization}
We now define our equilibrium concept and characterize the equilibrium allocation. 
\begin{definition}[Static equilibrium]
A static equilibrium consists of:
\begin{itemize}
    \item Purchasing decisions, $C_i=C_i(P_i) \in \{0,1\}$ for each available consumption good $i$, and labor supply decisions, $\ell=\ell(w)$, that solve the utility maximization problem of the representative household (\ref{ump}), taking the price vector $\boldsymbol{P}$ and the wage rate $w$ as given;
    \item A reservation surplus $x_{0,j}$ for each input $j=1,\dots,N$ which makes a final producer indifferent between buying the input or not, taking the expected surplus offered by intermediate producers of other inputs $\mathbb{\hat{E}}[\boldsymbol{x}_{-j}]$ as given;
    \item A selection rule for the input provider of each input $j=1,\dots,N$ which maximizes profits of the final producer (\ref{final_profits}), taking the surplus offered $x$ by each of the compatible intermediate producers contacted as given;
    \item Policy functions of specialization and surplus offered to final producers, $s(z)$ and $x(z)$ $\forall z$, which solve the profit maximization problem (\ref{static:pmp}) of intermediate producers, taking the wage rate $w$ and the distribution of surplus offered to final producers $G(x)$ as given;
    \item Distribution of surplus offered to final producers $G(x)$ which is consistent with the policy functions $s(z)$ and $x(z)$;
    \item Labor market clearing pinning down the wage rate such that labor supply equals labor requirement in specialization: $\ell(w) = Nm \bar{q}$, where $\bar{q} = \int q(s(z))\gamma(z) \ dz$;
    \item Consumption good market clearing pinning down the price vector $\boldsymbol{P}$ such that supply and demand coincide for each final producer: $C_i(P_i) = Y_i, \forall i$.
\end{itemize}
\end{definition}

We now characterize the decision of each agent in our economy. 

\paragraph{Representative household.} 
The representative household optimally purchases the indivisible consumption good offered by final producer $i$ provided that its quality weakly exceeds its price, that is, $C_i = \mathbbm{1}\{\mathcal{Q}_i \geq P_i\}$. The optimal labor supply function is $\ell(w)=1-\frac{\psi}{w}$.

\paragraph{Final producers.}
The reservation surplus $x_{0,j}$ for each input $j=1,\dots,N$ is the minimum surplus offered that a final producer is willing to accept. As such, it solves an expected break-even condition, which makes the final producer indifferent between sourcing the intermediate input in the market and not producing at all.

\begin{lemma}[Static Reservation Surplus]
\label{lem:res_surplus}
The equilibrium reservation surplus equals: 
\begin{align}
\label{eq:res_surplus}
 \pi(x_{0,j};\mathbb{\hat{E}}[\boldsymbol{x}_{-j}]) = 0 
    & \iff x_{0,j} = 
    -(N-1) \mathbb{\hat{E}}[x^\star(z)] \quad \forall j,
\end{align}
where $\mathbb{\hat{E}}[x_n]$ is the expected surplus offered by intermediate producers of input $n$ conditional on matching with final producers. 
\end{lemma}
\noindent
Since, in equilibrium, final producers make positive expected profits on each input line, the reservation surplus for each input is negative whenever production is complex, i.e., $N>1$.\footnote{The only combination of reservation surplus and expected surplus offered upon matching that is consistent with a symmetric equilibrium is such that the former is negative and the latter is positive.} If and only if production is not complex, the reservation surplus equals zero.


When multiple compatible intermediate producers are contacted, final producers optimally select the input provider offering the highest surplus $x_j$ for each input $j=1,\dots,N$. This is because the surplus directly reflects the profit margin the final producer can earn on that input, already incorporating the trade-off between its value added and its price.

\paragraph{Intermediate producers.}
To solve the intermediate producer's problem (\ref{static:pmp}), we proceed by guessing and verifying that the optimal surplus offered to final producers is increasing in productivity, i.e., $x^\prime(z)>0$. The policy functions $s^\star(z)$ and $x^\star(z)$ solve the following system of differential equations:
\begin{align} \label{eq_x}
    & {x^\star}^\prime(z)=\lambda \phi(s^\star(z)) \gamma(z) \left[A\left(s^\star(\tilde{z});\tilde{z}\right)-x^\star(z)\right] ,\\[.1cm]
    \label{eq_s}
    & \theta \lambda \Pc(z;N)\bigg[A^\prime(s^\star(z))+\frac{\phi^\prime(s^\star(z))}{\phi(s^\star(z))} \left(A (s^\star(z);z)-x^\star(z)\right)\bigg]=w q^\prime(s^\star(z)),
\end{align}
with boundary condition $x^\star(\underline{z})=x_0$. 
We assume $q(s)$ is such that equation (\ref{eq_s}) is necessary and sufficient for characterizing optimal specialization.\footnote{This amounts to making sure that the SOC always holds. The SOC reads $w q^{\prime\prime}(s^\star(z))>\theta \lambda e^{-\lambda \hat{\phi}(z,\bar{z})}f^{N-1} \left[A^{\prime\prime}(s^\star(z))\phi(s^\star(z))+\phi^{\prime\prime}(s^\star(z))[A\left(s^\star(z);z\right)-x^\star(z)]+2\phi^\prime(s^\star(z))A^\prime(s^\star(z))\right]$. With a slight abuse of notation, throughout we denote by $A^\prime(s)$ the partial derivative of $A(s;z)$ with respect to $s$. By additive separability, this derivative is independent of $z$.} 

For ease of notation, we denote as $f(z) \equiv 1-e^{-\lambda \hat{\phi}(\underline{z},z)}$ the probability that the final producer meets at least one compatible intermediate producer with productivity lower than $z$. This object is informative on the probability that an intermediate producer with productivity $z$ faces some competition from lower-productivity competitors.
 Naturally,  $f(\bar z)$ equals the input finding probability $f$. Given this definition, we can characterize the optimal surplus offered $x^\star(z)$ by solving equation \eqref{eq_x}. 
\begin{lemma}[Surplus Offered] \label{lem:offered_surplus}
The optimal surplus offered by an intermediate producer with productivity $z$ equals the expected outside option of final producers when trading with the intermediate producer:
\begin{align} \label{eq_x_solved}
x^\star(z) = \left(1-f(z)\right) x_0 + f(z)\mathbb{E}_{\max\{\tilde{z}\}|\tilde{z}\leq z}[A\left(s^\star(\tilde{z});\tilde{z}\right)],
\end{align}
where the expectation is taken with respect to the productivity distribution of the highest-surplus-offering compatible intermediate producers with lower productivity than $z$ contacted by a final producer, conditional on its presence.
\end{lemma}
\noindent
The intuition behind Lemma (\ref{lem:offered_surplus}) is instructive. Intermediate producers post their surplus offered before meeting the final producer. Hence, each intermediate producer with productivity $z$ forms expectations on the competing suppliers met by its potential customer. With probability $1-f(z)$, the final producer meets no intermediate producers with productivity lower than $z$. In this scenario, the optimal surplus offered would be the reservation surplus $x_0$. 
With complement probability, the final producer meets at least one potential supplier with productivity lower than $z$.
Since the intermediate producer has the power to set the surplus offered unilaterally, the optimal choice would be to make the final producer just indifferent between accepting and rejecting the surplus offered. This amounts to offering the final producer the expected match surplus from the best compatible competitor with productivity $\tilde{z}<z$.\footnote{An intermediate producer with productivity $z$ loses any auction against the other compatible competitors with higher productivity $z^\prime>z$ contacted by a final producer. 
Hence, an intermediate producer enjoys surplus- (or price-)setting power only in the states of the world when it has the highest productivity among the set of compatible intermediate producers contacted by a final producer.} 

Hence, intermediate producers with the lowest productivity $\underline{z}$ find it optimal to offer the reservation surplus to final producers, i.e., $x(\underline{z})=x_0$. Intuitively, these intermediate producers trade with a final producer only if the latter is not in contact with any other compatible competitor. 
On the contrary, intermediate producers with the highest productivity $\bar z$ offer the expected match surplus over the unconditional productivity distribution of the best supplier contacted, i.e., $x^\star(\bar z)=(1-f) x_0+ f \mathbb{E}_{\max\{\tilde{z}\}}[A(s(\tilde{z});\tilde{z})]$. Intermediate producers with productivity strictly in between offer a convex combination of these two bounds, i.e., $ x^\star(z) \in \left(x_0,(1-f) x_0+ f \mathbb{E}_{\max\{\tilde{z}\}}[A(s(\tilde{z});\tilde{z})]\right), \, \ z \in (\underline{z},\bar{z})$. This implies that the optimal surplus offered is increasing in productivity, ${x^\star}^\prime(z)>0$, which verifies our guess.

The optimal specialization function $s^\star(z)$ is implicitly defined by equation (\ref{eq_s}). We can single out three forces shaping optimal specialization. First, higher specialization increases the match surplus conditional on trading ($A^\prime(s)>0$). Second, higher specialization reduces the trading probability, by lowering the compatibility probability ($\phi^\prime(s)<0$). Third, higher specialization increases specialization costs, by raising the labor requirement for specialization ($q^\prime(s)>0$). 
Since higher complexity reduces the trading probability $\Pc(z;N)$, optimal specialization decreases with complexity $N$. If $\lambda \bar{\phi}<1$, higher search efficiency increases the average number of meetings per intermediate producer leading to trade, $\theta\lambda\Pc(z;N)$. Hence, optimal specialization increases with search efficiency $\lambda$. Finally, notice that the value function of an intermediate producer, $V(s,x;z,N)$, is supermodular in $s^\star(z)$ and $x^\star(z)$, i.e., the cross derivative at the optimal solution is positive. Hence, optimal specialization and optimal surplus offered to final producers are strategic complements.
In other words, if an intermediate producer were forced to offer a higher surplus to final producers ($x(z) \uparrow$), it would find it optimal to specialize more ($s(z) \uparrow$), and vice versa.\footnote{This result stands in contrast with standard models in which an investment is made before the parties bargain over the surplus split (\textit{hold-up problem}). We clarify the relation with these class of models in Section \ref{sec:discussion}.}
Intuitively, offering a higher surplus $x(z)$ reduces the cost of foregone trade due to lower compatibility probability. Unlike the optimal surplus offered, optimal specialization need not be monotone in productivity (see Appendix \ref{app:spec_function} for more details on the role of the shape of the productivity distribution).



\paragraph{Market clearing and consistency conditions.} Since ${x^\star}^\prime(z)>0$, the distribution of surplus offered to final producers equals $G(x^\star(z))=\hat{\phi}(\underline{z},z)/\bar{\phi}$, where $\hat{\phi}(\underline{z},z) \equiv \int_{\underline{z}}^z \phi(s^\star(\tilde{z}))\gamma(\tilde{z}) d\tilde{z}$ and $\bar{\phi}=\hat{\phi}(\underline{z},\bar{z})$. Intuitively, the probability that a final producer meets a compatible intermediate producer with productivity lower than $z$ equals the relative mass of intermediate producers with productivity lower than $z$, weighted by the respective compatibility probability. 

The equilibrium wage rate equalizing labor supply and labor requirement for specialization is given by $w=\frac{\psi}{1-Nm\bar{q}}$. The equilibrium price equalizing demand and supply for each consumption good reflects its quality, $P_i=\mathcal{Q}_i$. As a result, aggregate output equals:
\begin{align}
\label{output_static}
Y = \underbrace{f^N}_{\text{\vphantom{\Big|} prob. active}} \ \underbrace{N\mathbb{\hat{E}}[A\left(s^\star(\tilde{z});\tilde{z}\right)],}_{\quad \text{expected match surplus \big| active}}
\end{align} 
where $\mathbb{\hat{E}}[.]$ denotes the expectation with respect to the productivity distribution of active matches (defined in Appendix \ref{app:static}).
The expression for aggregate output in (\ref{output_static}) yields two insights. First, production complexity ($N$) has ambiguous effects on output, as in \cite{levine2012production}. Greater complexity raises output quality per active match but lowers the probability that all complementary inputs are successfully sourced. Second, a similar trade-off applies to endogenous specialization. Greater specialization increases the expected surplus per active match but reduces the likelihood that a final producer is active, thereby shrinking the measure of active final producers. Previewing our findings from the dynamic model in Section
\ref{dynamic_sec}, the probability that a final producer is able to source all required complementary inputs, $f^N$, is closely related to the notion of supply chain \textit{resilience}. Intuitively, if a final producer were to lose contact with all its suppliers, $f^N$ would equal the probability of successfully sourcing all the required inputs and restarting production, while $1/f^N$ would represent the expected time required to do so.

\subsection{Efficiency}
We study the efficiency properties of equilibrium specialization decisions by comparing the equilibrium allocation with that chosen by a benevolent social planner. The social planner is subject to the same product market frictions as market participants and solves:\footnote{Notice that the problem setup implicitly assumes that the social planner makes final producers trade with the compatible input provider contacted with the highest productivity, which can be easily proven.}
\begin{align*}
   \max_{ \substack{
s_j(z),\; j = 1,\dots,N \\
\quad z \in [\underline{z},\bar{z}]
}  } \  \mathcal{W}= \ C+\psi \log(1-\ell) \qquad 
   \text{s.t.} \ & \  \ell = m \sum_{n=1}^N \int q(s_n(z))\gamma(z) \ dz, \\
   & \ C= \sum_{n=1}^N \prod_{v \neq n}f(\boldsymbol{s}_v(z)) \  \mathbb{E}_{\max\{z\}}[A(s_n(z);z)].
\end{align*}
In words, the social planner chooses the specialization $s_j(z)$ of each intermediate producer with productivity $z$ producing input $j$ to maximize the utility of the representative household. The maximization problem is subject to the labor resource constraint and the output resource constraint.
Unlike equilibrium specialization decisions, the social planner internalizes that higher specialization of an intermediate input reduces the expected surplus from other complementary inputs, as final producers are less likely to be active.  

Since all input markets are symmetric, the first-order condition for the efficient specialization function $\Sc(z)$ is given by:
\begin{align} \nonumber
    & \theta \lambda \Pc(z;N) \bigg[A^\prime(\Sc(z))+\frac{\phi^\prime(\Sc(z))}{\phi(\Sc(z))}\bigg(A(\Sc(z);z)  - f(z) \mathbb{E}_{\max\{\tilde{z}\}|\tilde{z} \leq z}[A(\Sc(\tilde{z});\tilde{z})]  \\ \label{eff_s}
    & 
    + (N-1)\left(1-f(z)\right)\mathbb{\hat{E}}[A(\Sc(\tilde{z});\tilde{z})] \bigg)\bigg] 
    = \frac{\psi}{1-N m \bar{q}} \  q^\prime(\Sc(z)). 
\end{align}
The social marginal benefit of specialization on the left-hand side of (\ref{eff_s}) consists of four components, which can be grouped into direct and external effects. \\
\textit{Direct effects.}
(i) The first term, $\theta\lambda\Pc(z;N) A^\prime(\Sc(z))>0$, captures the increase in an intermediate producer's match surplus (conditional on matching) when specialization is marginally increased, holding the trading probability fixed.
(ii) The second term, $\theta\lambda\Pc(z;N)\phi^\prime(\Sc(z))/\phi(\Sc(z))A(\Sc(z);z)<0$, reflects the corresponding reduction in expected match surplus from a lower trading probability, holding constant both the conditional match surplus and the weighted productivity distribution, $\hat\phi(z,\bar z)$.
Together, these two terms represent the direct effects of specialization on the expected surplus from trade with a final producer. Accordingly, the same forces appear in the condition for privately optimal specialization (\ref{eq_s}). \\
\textit{External effects.}
(iii) The third term, $-\theta\lambda\Pc(z;N)\phi^\prime(\Sc(z))/\phi(\Sc(z))f(z) \mathbb{E}_{\max{\tilde{z}}|\tilde{z}\leq z}[A(\Sc(\tilde{z});\tilde{z})]$ $>0$, captures the positive spillover on the expected match surplus from other intermediate producers supplying the same input, holding constant both the conditional match surplus and the compatibility probability. Intuitively, by increasing specialization, a more productive firm raises the trading probability of less productive suppliers of that input.
(iv) By contrast, the fourth term, $\theta \lambda \Pc(z;N)(N-1)\phi^\prime(\Sc(z))/\phi(\Sc(z))\left(1-f(z)\right)\mathbb{\hat{E}}[A(\Sc(\tilde{z});\tilde{z})]<0$, captures the negative spillover on the expected match surplus from intermediate producers of complementary inputs. Here, greater specialization reduces their trading probability, because final producers are less likely to find all the inputs they need.
Taken together, terms (iii) and (iv) represent the external effects of specialization on other intermediate producers. Unlike the direct effects, these externalities are not internalized in the privately optimal specialization condition (\ref{eq_s}).

The social marginal cost on the right-hand side of (\ref{eff_s}) is given by the marginal increase in hours worked induced by higher specialization, weighted by the marginal disutility of labor. The social marginal cost exhibits the same structure as the private marginal cost.


We are interested in the efficiency properties of equilibrium specialization. 
To characterize them, we focus on the marginal welfare effect of specialization. By definition, efficient specialization equalizes the marginal welfare effect of specialization to zero, i.e., $\pd{\Wc}{s(z)}\big|_{s(z)=\mathcal{S}(z)}=0$.
Hence, evaluating the marginal welfare effect of specialization at the equilibrium solution allows us to single out the sources and direction of the inefficiency induced by equilibrium specialization.
\begin{prop}[Externalities from Specialization in the Static Economy]\label{prop:externalities}
The marginal welfare effect of equilibrium specialization can be decomposed as:
\begin{align} \nonumber
    \pd{\Wc}{s(z)}\bigg|_{s(z)=s^\star(z)} \hspace{-.5cm} \propto  \quad & \underbrace{\ f(z)\mathbb{E}_{\max\{\tilde{z}\}|\tilde{z}\leq z}\left[A(s^\star(\tilde{z});\tilde{z})\right]}_{\text{business-stealing externality}} \underbrace{- \vphantom{\mathbb{E}_{\max\{\tilde{z}\}|\tilde{z}\leq z}\left[A(s^\star(\tilde{z});\tilde{z})\right]} \ x^\star(z)}_{\substack{\text{appropriability} \\ \text{externality}}} \\ \label{dwds}
    & \underbrace{- \ \left(1-f(z)\right) (N-1)\mathbb{\hat{E}}[A(s^\star(\tilde{z});\tilde{z})]}_{\text{network externality}}.
    \end{align}
\end{prop}
Equilibrium specialization induces three externalities.
First, intermediate producers do not internalize that decreasing the specialization of their variety reduces the trading probability of lower-productivity competitors producing the same input (\textit{business-stealing externality}).
Second, intermediate producers do not internalize the total match surplus foregone due to lower trading probabilities, but only the portion they are able to appropriate (\textit{appropriability externality}).
Finally, intermediate producers do not internalize that increasing the specialization of their own variety reduces the trading probability of producers of complementary inputs (\textit{network externality}). Greater specialization of a given input variety reduces the probability that final producers source all required inputs and, as a consequence, purchase inputs from complementary input providers. As such, the network externality can be interpreted as a demand externality imposed on \textit{complementary} input providers.
 
Using our running example of the EV supply chain, the decision of a battery maker to specialize in type-A batteries induces a positive business-stealing externality on the trading probability of its lower-productivity competitors producing batteries with different chemical composition; a negative appropriability externality on EV manufacturers' surplus due to a lower probability of successfully sourcing batteries; and a negative network externality on the trading probability of tire and camera producers.

To understand the direction of the inefficiency induced by equilibrium specialization, we use equations \eqref{eq:res_surplus} and \eqref{eq_x_solved} to substitute the optimal surplus offered $x^\star(z)$ into equation \eqref{dwds}. This allows us to further decompose the appropriability externality into two components -- one operating \textit{within} an input market, and the other operating \textit{across} input markets (through the reservation surplus). The latter arises from intermediate producers anticipating that they would optimally offer the reservation surplus if a final producer were not in contact with any compatible lower-productivity competitors. Consequently, the intermediate producers appropriate part of the expected surplus offered by complementary input providers to final producers, in proportion to the probability $1-f(z)$ that the final producer does not find any lower-productivity compatible competitors. The next proposition formalizes these observations.
\begin{prop}[Offsetting Externalities Within Markets]\label{prop:offsetting}
Substituting equilibrium conditions into the marginal welfare effect of equilibrium specialization yields:
   \begin{align} \label{dwds_expanded}
     \pd{\Wc}{s(z)}\bigg|_{s(z)=s^\star(z)} \hspace{-.5cm} \propto \quad & \underbrace{\ f(z)\mathbb{E}_{\max\{\tilde{z}\}|\tilde{z}\leq z}\left[A(s^\star(\tilde{z});\tilde{z})\right]}_{\text{business-stealing externality}} \underbrace{- \ f(z)\mathbb{E}_{\max\{\tilde{z}\}|\tilde{z}\leq z}\left[A(s^\star(\tilde{z});\tilde{z})\right]}_{\substack{\text{appropriability externality} \\ \text{(within market)}}} \\ \nonumber
    & \underbrace{- \ \left(1-f(z)\right) (N-1)\mathbb{\hat{E}}[A(s^\star(\tilde{z});\tilde{z})]}_{\text{network externality}} \underbrace{+ \ \left(1-f(z)\right) (N-1)\mathbb{\hat{E}}[x^\star(\tilde{z})]}_{\substack{\text{appropriability externality} \\ \text{(across markets)}}} \leq 0.
\end{align}
\end{prop}
Equation (\ref{dwds_expanded}) sheds light on the direction of the inefficiency induced by equilibrium specialization. We highlight two key insights. 
First,  the appropriability externality within each input market perfectly offsets the business-stealing externality (first row).\footnote{Readers familiar with CES monopolistic competition models will note the same coincidence of offsetting externalities. We discuss this parallel in detail in Section \ref{sec:discussion}.} This result follows from our surplus/price-posting mechanism with simultaneous search: intermediate producers optimally offer final producers the expected \textit{match} surplus from the second-best compatible supplier, when one exists. As a result, the portion of match surplus that an intermediate producer does not appropriate exactly equals the expected match surplus that would have been generated by lower-productivity competitors had they traded.

Second, the direction of the overall inefficiency depends on the gap between the network externality and the appropriability externality across markets (second row). 
When production is not complex ($N=1$), both the network externality and the appropriability externality across markets vanish, so equilibrium specialization is efficient, as in \cite{Menzio2023_ProductDesign}. This result is immediate since (i) there are no complementary inputs that can be damaged by final-production halts, and (ii) there are no complementary inputs from which final producers can extract positive surplus, so the reservation surplus is zero. 
When production is complex ($N>1$), the two opposing forces emerge. On the one hand, the network externality pushes the privately optimal specialization above its socially efficient level, in proportion to the expected match surplus from complementary inputs that is foregone when production halts, $(N-1)\mathbb{\hat{E}}[A(s^\star(\tilde{z});\tilde{z})]$. On the other hand, the appropriability externality across markets pushes the privately optimal specialization below its socially efficient level, in proportion to the expected surplus offered from complementary inputs, $(N-1)\mathbb{\hat{E}}[x^\star(\tilde{z})]$, which intermediate producers appropriate in the form of negative reservation surplus.
Since the surplus offered is lower than the match surplus, the appropriability externality across input markets always falls short of the network externality. The lower (more negative) the reservation surplus, the larger the share of surplus from complementary inputs which intermediate producers appropriate and, consequently, the lower the inefficiency from the gap between network externality and appropriability externality across markets.

 We summarize the efficiency properties of equilibrium specialization in the following theorem. All proofs are relegated to Appendix \ref{app:proofs_resilience}.
\begin{thm}[Efficiency of the Static Economy]\label{thm:over-specialization}
    The equilibrium allocation is constrained-efficient if and only if production is not complex, i.e. $N=1$. If the production process is complex, i.e. $N>1$, the equilibrium features over-specialization.
\end{thm}
Overall, the marginal welfare effect of specialization at the equilibrium solution is negative. Since final producers do not appropriate the entire surplus from matching with complementary inputs, the reservation surplus is always too high. Hence, the planner would impose a lower level of specialization and, therefore, the equilibrium features \textit{over-specialization}. 

\paragraph{Network externality.}
The network externality arises from the interplay of three elements. First, search and compatibility frictions imply that final producers may not find an input they are looking for ($f<1$). Second, endogenous product design decisions imply that the input finding probability decreases with input specialization ($f^\prime(s)<0$). Third, complementarity among production inputs ($N>1$) implies that a final producer's failure to source one input prevents the surplus from successful matches with other intermediate producers to materialize. We clarify the economic forces behind the network externality in the following 2-firm example.

Assume there are two input providers supplying the only two complementary inputs required for final production ($N=2$). Suppose further that final producers meet exactly one intermediate producer per input, which implies that the latter faces no competition and can extract the entire match surplus.\footnote{Since in this example each intermediate producer is the only producer of its input, there is no business-stealing externality on other firms with lower productivity supplying the same input. As a consequence, the fact that the intermediate producer offers $x=0$ surplus is not an inefficiency per se.} Denote as $f(s_i)$ the probability that the input variety of provider $i$ is compatible with the final producer when its specialization is $s_i$. Therefore, the probability that a final producer source both inputs is $f(s_1)f(s_2)$. Higher specialization raises match surplus $A(s_i)$ $(A^\prime>0)$, but reduces the compatibility probability ($f^\prime<0$).
Social surplus reads $\mathcal{W}=f(s_1)f(s_2)[A(s_1)+A(s_2)]$. Efficient specialization of firm $1$ solves:
$$\underbrace{f(\Sc_1)A^\prime(\Sc_1)}_{\text{social MB}}=\underbrace{-f^\prime(\Sc_1)[A(\Sc_1)+A(\Sc_2)]}_{\text{social MC}}.$$
The left-hand side represents the additional surplus generated by the higher specialization, while the right-hand side equals the expected foregone social surplus in case of no trade. 

Expected private surplus of firm $1$ is $\pi_1=f(s_1)f(s_2)A(s_1)$. Equilibrium specialization solves: 
$$\underbrace{f(s^\star_1)A^\prime(s^\star_1)}_{\text{private MB}}=\underbrace{-f^\prime(s^\star_1)A(s^\star_1)}_{\text{private MC}}.$$

Evaluating the efficiency condition at the equilibrium specialization, we note that we are left with the term $f^\prime (s^\star_1)A(s^\star_2)<0$, which captures the negative externality of firm 1 on the trading probability of firm 2. It follows that the social marginal cost is higher than the private one, and firm 1 over-specializes. By symmetry, the same holds for firm $2$, as well.

Importantly, the network externality is \textit{pure}, in the sense that it arises only in no-trade scenarios. In other words, higher specialization of intermediate producers of input $j$ harms the intermediate producers of other inputs $i \neq j$ whenever the final producers they are in contact with are unable to source input $j$.
As a consequence, there is no price mechanism that allows intermediate producers of input $j$ to internalize this externality, since it materializes precisely in states in which no trade occurs.

\subsection{Discussion}
\label{sec:discussion}
Before extending the model to dynamics, we briefly discuss how our static framework relates to existing models and the robustness of our conclusions under alternative assumptions. In particular, we consider bargaining over surplus, the connection to monopolistic competition models with CES preferences, directed search, non-contingent and two-part contracts, vertical integration, and in-house emergency production of missing inputs.

\paragraph{Bargaining.}
Our baseline model of specialization posits that intermediate producers post the price of their input variety (or, interchangeably, the surplus offered) before meeting. We now show how our baseline model would change if intermediate and final producers bargained over the surplus.

Let the structure of our model be the same up to the surplus-sharing rule.
Assume that the surplus accruing to final producers is determined by generalized Nash bargaining upon meeting between a final producer and a compatible intermediate producer. Final producers can observe the match surplus $v$ from matching with any intermediate producer they meet. Let $\xi \in [0,1]$ denote the exogenous bargaining power weight of final producers. The surplus accruing to final producers in equilibrium solves 
 $  X^\star(v) = \argmax_{X} \left(X-\omega(v;\vec{v}_k)\right)^\xi \left(v-X\right)^{1-\xi}, $ 
where $\omega(v;\vec{v}_k)$ denotes the outside option of a final producer that is in contact with $k$ compatible intermediate producers whose match surplus is lower than $v$. While negotiating with the highest-surplus compatible intermediate producer, final producers keep the option of bargaining with the other less productive compatible intermediate producers they are in contact with. Formally, $\omega(v;\vec{v}_k)= \mathbbm{1}_{\{k=0\}}X_0+\mathbbm{1}_{\{k>0\}}\max\{v_1,\dots,v_k\}$, where $X_0$ is the reservation surplus. In Appendix \ref{app:bargaining} we show that the optimal specialization function $s_B^\star(z)$ with bargaining over surplus shares is implicitly defined by:
\begin{align} \label{eq_s_barg}
    &\theta \lambda \mathcal{P}(z;N) (1-\xi) \bigg[A^\prime(s_B^\star(z))+\frac{\phi^\prime(s^\star_B(z))}{\phi(s^\star_B(z))} \left(A(s_B^\star(z);z)-\mathbb{E}[\omega(z)]\right)\bigg]=w q^\prime(s_B^\star(z)),
\end{align}
where $\mathbb{E}[\omega(z)]$ is the expected outside option of a final producer trading with the intermediate producer, which equals the surplus offered to final producers in our baseline model.

The equilibrium specialization function $s^\star_B(z)$ is generally inefficient due to the gap between network externality and appropriability externality across markets \textit{and} a well-known hold-up externality \citep{tirole}.
The hold-up externality is related to the appropriability externality in our baseline model. However, there are important differences between the two. Hold-up arises when intermediate and final producers split ex-post the surplus generated by the investment in specialization. Increasing investment in specialization marginally generates $A^\prime$ units of additional match (or \textit{total}) surplus, which is split according to $\{\xi, \,1-\xi\}$ shares. Hence, part of the \textit{marginal} surplus generated by the intermediate producer's investment accrues to the final producer. Because the total and private marginal returns on investment are different, the investment level is inefficient. 
Differently, in our model with surplus posting, the intermediate producer commits to a \textit{level} $x$ of surplus offered to the final producer. Hence, any additional surplus generated by increasing investment in specialization marginally, $A^\prime$, only accrues to the intermediate producer itself. Since marginal total and private surpluses coincide, the investment level is bilaterally efficient. The two coincide only in the limiting case of the bargaining model in which the seller has all the bargaining power, in which the hold-up problem disappears.
In this sense, our model allows for bilaterally efficient investment even in a setting in which both parties in the transaction obtain a positive surplus. This result is driven by the competitive pressure exerted by the simultaneous search, which implies that buyers have a non-zero outside option. 
%
 %
Comparing equations (\ref{eq_s_barg}) and (\ref{eq_s}) allows us to make the following remark.
\begin{remark}[General Bargaining Model]
\label{remark:bargaining}
   The equilibrium allocation of the baseline static model is the same as that of a general bargaining model where intermediate producers hold all
   bargaining power, i.e., $\xi=0$. The allocation of the bargaining model is generally inefficient. If bargaining weights varies with the productivity of the intermediate producer involved in the transaction, there exists a unique vector $\xi^\star(z;N)$ such that the market and planner solutions coincide.
\end{remark}
\noindent
Overall, when production is not complex, our baseline model can be interpreted as the efficient benchmark of a general bargaining model.

\paragraph{Monopolistic competition with CES production function.} A common framework when studying firm-to-firm relations is that of \cite{dixit1977monopolistic}, as exemplified by \cite{grossman2023supply,grossman2023resilience}.
In this class of models, firm entry gives rise to a business-stealing externality and an appropriability externality. Still, the equilibrium allocation with free entry and inelastic labor supply is constrained-efficient. The reason is that the two externalities exactly offset each other, as a consequence of the elasticity of substitution -- which governs the business-stealing externality -- and the taste for varieties -- which governs the appropriability externality -- being the same parameter. In models where the two do not coincide, the economy is generally inefficient.\footnote{See \cite{dixit1975monopolistic,spence1976product,dixit1977monopolistic,ethier1982national,mankiw1986free,benassy1996taste,zhelobodko2012monopolistic,parenti2017toward,dhingra2019monopolistic,grossman2023supply,bajoetal} for several treatments of this problem, and \cite{brakman2001monopolistic} for a review.} 
In our model, endogenous product design decisions are akin to firm entry decisions. The average input specialization determines the mass of compatible varieties available to each final producer. Our price-posting protocol with simultaneous search implies that equilibrium specialization endogenously generates the same offsetting externalities as in \cite{dixit1977monopolistic} within each input market. The key difference is that, in our model, this outcome arises from the optimal choices of intermediate producers, who make final producers indifferent between trading with them and relying on their expected outside option, without hinging on any parametric restriction.
%
\begin{remark}[Monopolistic Competition \& CES]
\label{remark:MP_CES}
    The combination of endogenous specialization and price posting in frictional markets with simultaneous search provides an endogenous mechanism such that the appropriability and business-stealing externalities from designing new product varieties exactly offset each other within markets. Models of monopolistic competition with CES final production function obtain the same result by imposing the exogenous parametric restriction that the elasticity of substitution across varieties in the final production function equals the elasticity of firm-level inverse demand.
\end{remark}
%

\paragraph{Directed search.}
Our baseline model features random search within input markets. A natural question to ask is whether the efficiency properties of our baseline equilibrium hinge upon the search protocol. To address this question, in Appendix \ref{app:directed_search} we develop an alternative version of our model where final producers can direct their search within each input market.

We find that in a static competitive search equilibrium, the allocation is the same as in our baseline model with random search. 
 Hence, whether search is random or directed within input markets is irrelevant for our results.
Directed search would restore efficiency only if all input markets were integrated. We conclude that achieving constrained efficiency requires full (both vertical and horizontal) integration of the supply chain.
%

\begin{remark}[Directed Search]\label{rem:dir_search}
A static competitive search equilibrium yields the same privately optimal specialization choice as its random search counterpart. The resulting allocation is constrained-efficient only if all input markets are integrated. 
\end{remark}


\paragraph{Contract contingency and completeness.}
So far, we have considered an economy with a specific contract form between intermediate and final producers. First, purchases of intermediate goods are contingent on successfully sourcing \textit{all} inputs, so final producers face no risk of paying suppliers when production cannot occur.\footnote{\cite{Chor2025contracting_frictions} shows that contracting frictions generate substantial welfare losses in global sourcing.} Second, only a single payment is allowed, conditional on trade. Since our inefficiency arises in no-trade states, it is natural to ask whether a second, no-trade payment could eliminate it.

In Appendix \ref{app:non-cont}, we consider deviations from these assumptions. First, we examine an alternative, non-contingent contractual arrangement in which suppliers are paid independently of whether all inputs are successfully sourced. We find that intermediate producers have stronger incentives to over-specialize when contracts are not contingent.
\begin{remark}[Non-Contingent Contracts] \label{remark:noncontingent}
   Economies with non-contingent contracts exhibit greater equilibrium over-specialization than economies with contingent contracts.  
\end{remark}

Second, we consider an alternative extension to our baseline model, whereby intermediate producers can offer a surplus $x^\circ(z) \geq 0$ in case of no trade. Specifically, the intermediate producer offers $x^\circ(z)$ to the final producer upon meeting if the latter cannot find any compatible input suppliers. We view this contractual arrangement as a partial compensation for no-trade scenarios that would make incompatibility more costly for intermediate producers. However, final producers still find it optimal to choose the highest surplus-offering compatible supplier -- irrespective of no-trade payments. As a consequence, positive no-trade payments are not sustainable in equilibrium and do not solve the network externality.
\begin{remark}[No-Trade Payments]
\label{remark:notrade_payments}
    Suppose that final producers can offer a two-part contract contingent on the compatibility realizations. Then, no-trade payments are optimally set to zero in equilibrium. Therefore, the equilibrium is the same as in the baseline economy and features over-specialization.
\end{remark}

\paragraph{Boundaries of the firm, vertical integration, and in-house production.}
At the core of our model are frictions in sourcing inputs from outside the firm, which are amplified by the complexity of production, $N$. Greater complexity increases potential value added but reduces the probability that all inputs are successfully sourced. Firms can address this trade-off in multiple ways: they may choose the complexity of their production process or adjust the boundaries of the firm for a given production process. In Appendix \ref{app:end_N}, we study the optimal choice of $N$. Here, we focus on vertical integration and in-house input production.

Suppose first that final producers can vertically merge with $M<N$ input providers, so that these inputs are no longer subject to sourcing frictions. In this case, the effective complexity of production is $N-M$.

\begin{remark}[Vertical Integration]
\label{remark:vertical_integration}
Vertical integration mitigates inefficiencies from endogenous specialization by reducing the number of key inputs subject to search frictions. If all but one input are vertically integrated, the economy is constrained-efficient.
\end{remark}

\noindent
Hence, vertical integration is socially desirable whenever the merger cost is smaller than the social benefit from mitigating the network externality.

Next, suppose final producers can manufacture any intermediate good of value added $\underline{A}$ in-house at cost $F>\underline{A}$. This option may be exercised both before and after the resolution of sourcing uncertainty: firms that fail to source all the required inputs can ex-post manufacture the missing ones. In Appendix \ref{app:in-house}, we show that the decision to use in-house production depends on the sum of surpluses from successfully sourced inputs and the number of missing inputs.  

Importantly, because final producers do not appropriate the full match surplus from production, they underuse in-house production relative to the planner's optimum. As a result, the network externality persists -- analogous to the case of non-contingent contracts. If in-house production were costless on net ($F=\underline{A}$), both sources of externalities would vanish: final production would always occur, and the planner and firms would employ in-house production at the same rate. When $F>\underline{A}$, however, the market allocation remains inefficient.

\begin{remark}[In-House Production]
\label{remark:in-house}
Costly in-house production does not eliminate the network externality and introduces an additional inefficiency in the frequency of in-house production. If in-house production is costless on net, the economy is constrained-efficient.
\end{remark}

\section{Dynamic Model}\label{dynamic_sec}

In the previous section, we have studied how endogenous specialization of intermediate inputs affects the sourcing capacity of final producers. Our finding is that the equilibrium allocation features over-specialized inputs as intermediate producers do not internalize the cascading effect of halting final production on the value added produced by complementary inputs.
In this section, we show that the resilience of supply chains is closely related to the dynamic sourcing capacity of final producers after disruption shocks.
To this end, we extend the static model to a dynamic framework with long-term relationships between intermediate and final producers and stochastic disruptions affecting final producers or supplier links. This framework provides a precise and welfare-relevant notion of supply chain resilience: the probability that final production is restored in a given period following a disruption. This measure aligns well with practitioner usage.

\subsection{Environment}
Time is discrete and infinite. There is no aggregate risk. Agents discount the future at rate $\beta \in (0,1)$. Final producers face a disruption at the end of each period with probability $\delta$, where $\delta \in (0,1)$. Upon facing a disruption, final producers lose contact with all their input providers and need to replace them in the next period.
Final producers keep the same input providers until a disruption occurs. In other words, final producers do not search the market for replacing existing suppliers, e.g., due to switching costs.
 We refer to the average time a final producer does not face any disruption,  $1/\delta$, as \textit{robustness}.\footnote{In Appendix \ref{app:end_rob} we endogenize robustness by allowing final producers to make prudential investments to lower their disruption risk.} We assume that supplier relationships are governed by constant-surplus contracts.
 Intermediate producers never face disruptions and search for new customers in the market in each period. 
 The dynamic model nests the static model as final producers lose contact with all their input providers at the end of each period ($\delta \rightarrow 1$).

The timing convention is as follows: at the end of each period, a share $\delta$ of attached final producers faces a disruption, which makes them lose all their input providers. At the start of the next period, unattached final producers search for suppliers. The mass of searching final producers $\mu$ evolves according to the following law of motion:
\begin{align} \label{mu_dyn}
    \mu_{(+1)} = \delta \left[1-\left(1-\prod_{j=1}^Nf_j\right)\mu\right]+\left(1-\prod_{j=1}^Nf_j\right)\mu,
\end{align}
where $\left(1-\prod_{j=1}^Nf_j\right)\mu$ equals the mass of inactive final producers in the current period.
Market tightness equals the ratio between the measure of searching final producers and the measure of intermediate producers:
$
    \theta = \frac{\mu}{m}.
$\footnote{Throughout, we maintain the assumption that $m$ is fixed. In Appendix \ref{sec:free_entry}, we endogenize the mass of intermediate goods producers via a free entry condition.} 

The household problem is static and identical to the previous section. 
Intermediate producers choose the specialization of their input variety and the surplus offered to the flow of new customers by solving the following recursive profit maximization problem:
\begin{align} \nonumber
    V(\mathcal{D}_{(-1)},\bar{X}_{(-1)};z,N) = \max_{s,x} \ & \mathcal{D} A(s;z) - \bar{X}_{(-1)} - \theta\lambda\Pc(s,x;N)x -w q(s) \\
    & + \beta  V(\mathcal{D},\bar{X};z,N)
    \label{pmp_dyn}\\[.3cm] \label{demand_dyn}
    \text{s.t.} \ & \mathcal{D} = (1-\delta) \mathcal{D}_{(-1)} +
 \theta\lambda\Pc(s,x;N),\\   
 & \bar{X} = (1-\delta)\left[\bar{X}_{(-1)}+\theta\lambda\Pc(s,x;N)x \right], \label{eq:xlom}
\end{align}
%
where $\mathcal{D}$ denotes the number of final producers sourcing from the intermediate producer, or \textit{demand}, and $\bar{X}$ denotes the surplus offered to the stock of existing customers. $\mathcal D$ and $\bar X$ follow the laws of motion in equations \eqref{demand_dyn} and \eqref{eq:xlom}.

Since each intermediate producer serves a measure of final producers, a share $\delta$ of existing customers faces a disruption and breaks the supplier relationship in each period. At the same time, the intermediate producer acquires a measure $\theta\lambda\Pc(s,x;N)$ of new customers, depending on the market tightness $\theta$ and its own policies ($s$ and $x$).
While we only consider constant-surplus contracts within each buyer-supplier relation, we allow intermediate producers to re-optimize the specialization of their inputs in each period. Importantly, input producers are free to offer different pricing conditions to new potential customers compared to existing ones. 

\paragraph{Customization.} Our dynamic model allows intermediate producers to re-optimize their specialization choices in every period. This introduces a new \textit{customization motive} for specialization, that is, tailoring the product to the needs of the final producers \textit{within a match}. Long-term relationships make specialization more appealing for intermediate producers: within an existing match, increasing specialization comes at no cost in terms of lower trading probability.

The main prediction of our dynamic model is that intermediate producers gradually increase the specialization and narrow the scope of their variety during the life cycle as they progressively accumulate customers. The intuition is best understood by focusing on the first-order condition for specialization out of steady state:
\begin{align} \nonumber
 & \overbrace{(1-\delta)\mathcal{D}_{(-1)}}^{\text{existing customers}}A^\prime(s^\star(z))+\overbrace{\theta\lambda\mathcal{P}(z;N)}^{\text{new customers}}\left[A^\prime (s^\star(z))+\frac{\phi^\prime(s^\star(z))}{\phi(s^\star(z))}\bigg([A(s^\star(z);z)-x^\star(z)]+\beta(1-\delta) \right. \\ \label{eq_s_dyn_out}
 & \left. [A(s^\star_{(+1)}(z);z)-x^\star(z)] \bigg) \right] = w q^\prime(s^\star(z)).
\end{align}
Optimal specialization features both a backward- and a forward-looking component. The marginal benefit of specialization on the left-hand side can be divided into two components. First, for the $(1-\delta)\mathcal{D}_{(-1)}$ existing customers, higher specialization implies higher surplus $(A^\prime(s)>0)$. For the $\theta\lambda\mathcal{P}(z)$ new customers, this benefit is traded off against the lower trading probability, governed by $\phi^\prime(s)<0$. As a firm accumulates a customer base $\Dc_{(-1)}$, the first term becomes more important, increasing the marginal benefit of specialization. 
Since intermediate producers have already overcome compatibility frictions with many customers, they opt to increase match surplus with existing customers rather than attract new ones. Hence, for the share of attached customers, specialization corresponds to customization. It follows that both the value added and the price of the inputs traded between firms increase over the life cycle.  

\subsection{Stationary equilibrium}
We now characterize the stationary equilibrium with balanced flows of our dynamic economy. The policy functions of the representative household are the same as in the static model.

Unlike in the static model, the value of final producers embeds an option value of search, which reflects in the reservation surplus (see Appendix \ref{app:dyn} for further details).
\begin{lemma}[Dynamic Reservation Surplus]
\label{lem:res_surplus_dyn}
The equilibrium reservation surplus equals:
\begin{align} \label{res_surplus_dyn}
   x_0 = -\left(\alpha(N)N-1\right)\mathbb{\hat{E}}[x^\star(z)],
\end{align}
where $\alpha(N) \equiv  1-\frac{\beta(1-\delta) f^N}{1-\beta(1-\delta)(1-f^N)} \in (0,1)$ summarizes the influence of the option value of search on the reservation surplus.
\end{lemma}
\noindent
Final producers find it optimal to choose among compatible suppliers the one offering the highest net present value of surplus streams. Under constant-surplus contracts, this coincides with the supplier offering the highest current surplus.

The optimal surplus offered to final producers has the same structure as in the static model, given by (\ref{eq_x_solved}) (see Appendix \ref{app:dyn} for the derivation). 
The stationary demand of a seller choosing specialization $s$ and surplus $x$ is
%
$    \mathcal{D}(s,x;N) = \frac{\theta\lambda\Pc(s,x;N)}{\delta}.
$ 
The optimal specialization function is implicitly defined by the following first-order condition:
\begin{align} \label{eq_s_dyn}
 \mathcal{D}(z;N)\bigg[A^\prime (s^\star(z))+\delta[1+\beta(1-\delta)]\frac{\phi^\prime(s^\star(z))}{\phi(s^\star(z))}(A(s^\star(z);z)-x^\star(z))\bigg] = w q^\prime(s^\star(z)).
\end{align}
Increasing specialization affects existing and new customers differently. For all the $\mathcal{D}(z;N)$ customers, higher specialization implies a higher match surplus conditional on trading $(A^\prime(s)>0)$. For the existing $(1-\delta)\mathcal{D}(z;N)$ customers, higher specialization does not come at the cost of lower trading probability since the match has already been established in the past. For the new $\delta \mathcal{D}(z;N)$ customers, instead, higher specialization lowers the trading probability ($\phi^\prime(s)<0)$ -- as in the static model.\footnote{Since inflow and outflow of customers at the firm-level balance each other in stationary equilibrium, $\delta \mathcal{D}(z;N)$ equals both the mass of existing customers that separate and the mass of new customers.} Nonetheless, the forward-looking firm understands that lower trading probability today lowers demand tomorrow. Hence, the reduction in trading probability is weighted by the intertemporal value of a marginal customer, $1+\beta(1-\delta)$, which captures the possibility of re-optimizing specialization in every period.

\paragraph{Resilience.}
The dynamic model allows for a precise and welfare-relevant notion of supply chain resilience, namely, the probability that a final producer is able to source all the required inputs in a given period after a disruption:
\begin{align} \label{resilience}
    \mathcal{R}(f;N) \equiv f^N = [1-\exp\{-\lambda\bar{\phi}\}]^N.
\end{align}
Given this definition, it is apparent how the inverse of supply chain resilience equals the average time it takes for a final producer to restore production after a disruption.

Our model allows us to single out three determinants of supply chain resilience.
First, higher search efficiency ($\lambda \uparrow$), e.g., through ICT adoption and AI-based procurement, increases the resilience of supply chains.\footnote{See the practicioners' reports of \citet{McKinsey2021suppliers,McKinsey2024procurement}.} Second, higher resilience can be attained by relying on less specialized inputs ($\bar{\phi} \uparrow$). This corresponds to the strategy of building resilience "through products" highlighted by \cite{miroudot2020resilience}, such that firms choose to source more standardized inputs in the market because they are easier to replace. Finally, higher complexity of the production process ($N \uparrow$) reduces supply chain resilience. This echoes the main insight from the network fragility literature \citep{Elliott2022_Fragility}. 
Crucially, supply chain resilience directly affects output and, in turn, welfare. Output in this economy can be decomposed into four terms:
\begin{align*}
    Y = \underbrace{\vphantom{\frac{1}{\delta}} \mu(f;N)}_{\text{market size}} \ \  
    \underbrace{\vphantom{\frac{1}{\delta}} \mathcal{R}(f;N)}_{\text{resilience}} \ \ \underbrace{\frac{1}{\delta}}_{\text{robustness}} \ \
    \underbrace{\vphantom{\frac{1}{\delta}} N\mathbb{\hat{E}}[A(s^\star(z);z)]}_{\text{expected match surplus| active}}.
\end{align*}
First, output depends positively on market size (for investment in specialization). Market size increases if the mass of searching final producers $\mu$ increases. 
Second, output is increasing in supply chain resilience.  Third, output increases with supply chain robustness, i.e., the average duration of a match between final and intermediate producers. Finally, output is increasing in the expected match surplus conditional on production.
Hence, our model justifies policymakers' interest in supply chain resilience as a welfare-relevant statistic.

\subsection{Efficiency}
To establish the efficiency properties of the equilibrium allocation in the dynamic economy, we consider a social planner solving the following recursive problem:
\begin{align*}
\mathcal{W}\left(\boldsymbol{\mathcal{D}}_{(-1)},\mu\right) = \max_{ \substack{
s_j(z),\; j = 1,\dots,N \\
\quad z \in [\underline{z},\bar{z}]
}  } & \ m \sum_{n=1}^N \int D_n(z)A(s_n(z);z) \ \gamma(z)dz + \psi \log(1-\ell) +\beta \mathcal{W}\left(\boldsymbol{\mathcal{D}},\mu_{(+1)}\right) \\[.4cm]
\text{s.t.} \ & D_j(z) = (1-\delta) D_{j(-1)}(z) + \frac{\mu}{m} \lambda \phi(s_j(z)) e^{-\lambda \hat{\phi}(z,\bar{z})} \prod_{v\neq j} f_v, \ \forall j, z, \\[.2cm]
& \mu_{(+1)} = \delta \left[1-\left(1-\prod_{n=1}^N f_{n}\right)\mu\right]+\left(1-\prod_{n=1}^N f_{n}\right)\mu, \\[.2cm]
& \ell = m \sum_{n=1}^N \int q(s_n(z)) \gamma(z) dz,
\end{align*}
where $\boldsymbol{\mathcal{D}}$ is the vector of productivity-specific demand for all inputs.

In words, the social planner chooses the specialization $s_j(z)$ of each intermediate producer
with productivity $z$ supplying input $j$, in order to maximize the lifetime utility of the representative household. 
The state variables are the vector of demand for all intermediate producers one period earlier, $\boldsymbol{\mathcal{D}}_{(-1)}$, and the share of searching final producers in the current period, $\mu$.
The maximization problem is subject to the laws of motion for firm-level demand and for the share of searching final producers, as well as to the labor resource constraint.

The efficient specialization of intermediate producers with productivity $z$ in steady state is given by:
\begin{align} \nonumber
      & \mathcal{D}(z;N) \bigg[A^\prime(\Sc(z))+\delta[1+\beta(1-\delta)]\frac{\phi^\prime(\Sc(z))}{\phi(\Sc(z))}\bigg(A(\Sc(z);z) - f(z)\mathbb{E}_{\max\{\tilde{z}\}|\tilde{z} \leq z}[A(\Sc(\tilde{z});\tilde{z})]  \\ \nonumber
    & 
    +(N-1)\left(1-f(z)\right)\mathbb{\hat{E}}[A(\Sc(\tilde{z});\tilde{z})]-\frac{\beta(1-\delta)}{1+\beta(1-\delta)}f^N N \left(1-f(z)\right) \mathbb{\hat{E}}[A(\Sc(\tilde{z});\tilde{z})]\bigg)\bigg] \\  \label{eff_s_dyn}
    &
    = \frac{\psi}{1-N m \bar{q}} \  q^\prime(\Sc(z)). 
\end{align}
Relative to the statically efficient specialization condition (\ref{eff_s}), an additional term arises in the dynamic setting capturing the marginal effect of specialization on the meeting probability through changes in the share of searching final producers, $\mu$.

To characterize the efficiency properties of the equilibrium of our dynamic economy, 
we follow the same steps as in the static case and focus on the marginal welfare effect of equilibrium specialization. We decompose the marginal welfare effect into five externalities. Four of them are inherited from the static model (see Proposition \ref{prop:offsetting}).
A fifth  arises only in the dynamic model: intermediate producers, by choosing their specialization, affect market tightness through changes in the probability that inputs are compatible. This effect, analogous to the standard \textit{search externality} in models of endogenous search effort \citep{pissarides2000}, is not internalized by intermediate producers. Specifically, greater specialization of intermediate producers reduces the probability that final producers successfully source their inputs. On the one hand, this lowers the trading probability of the intermediate producers themselves, which they internalize. On the other hand, because the mass of searching final producers increases, all intermediate producers benefit from a \textit{thicker} market. This effect is not internalized by individual intermediate producers and therefore constitutes a positive externality. Notably, the search externality differs from the business-stealing externality: the former affects all intermediate producers, whereas the latter only impacts lower-productivity producers within the same input line.
\begin{prop}[Externalities from Specialization in the Dynamic Economy]\label{prop:externalities_dyn}
The marginal welfare effect of equilibrium specialization can be decomposed as:
\begin{align} \nonumber
     \pd{\Wc}{s(z)}\bigg|_{s(z)=s^\star(z)} \hspace{-.5cm} \propto \quad & \underbrace{\ f(z)\mathbb{E}_{\max\{\tilde{z}\}|\tilde{z}\leq z}\left[A(s^\star(\tilde{z});\tilde{z})\right]}_{\text{business-stealing externality}} \underbrace{- \ f(z)\mathbb{E}_{\max\{\tilde{z}\}|\tilde{z}\leq z}\left[A(s^\star(\tilde{z});\tilde{z})\right]}_{\substack{\text{appropriability externality} \\ \text{(within market)}}} \\ \nonumber
    & \underbrace{- \ \left(1-f(z)\right) (N-1)\mathbb{\hat{E}}[A(s^\star(\tilde{z});\tilde{z})]}_{\text{network externality}} \underbrace{+ \ \left(1-f(z)\right) (\alpha(N)N-1)\mathbb{\hat{E}}[x^\star(\tilde{z})]}_{\substack{\text{appropriability externality} \\ \text{(across markets)}}} 
 \\ \label{dwds_dynamic}
    & \underbrace{\vphantom{\int_{\underline z}^z}+ \ \frac{\beta(1-\delta)}{1+\beta(1-\delta)}f^N N\left(1-f(z)\right)\mathbb{\hat{E}}[A(s^\star(\tilde{z});\tilde{z})]}_{\text{search externality}} <0, \quad \forall N>1.
\end{align}
\end{prop}
%
The decomposition of the marginal welfare effect of equilibrium specialization sheds light on the sources of inefficiency in the equilibrium allocation.
 As before, the business-stealing and within-market appropriability externality cancel out, thanks to the posting mechanism and simultaneous search. 
 When production is not complex ($N=1$), the network externality vanishes. Thus, the only sources of inefficiency are the appropriability externality across markets -- driven by a positive reservation surplus, $x_0 = (1-\alpha(1))\mathbb{\hat{E}}[x^\star(z)]> 0$ -- and the new search externality. The appropriability externality across markets enters equation (\ref{dwds_dynamic}) negatively, pushing the planner to lower specialization relative to its equilibrium level. In contrast, the search externality enters equation (\ref{dwds_dynamic})  positively, inducing the planner to increase specialization over its equilibrium level. Depending on the relative magnitude of these two externalities, the equilibrium with non-complex production can display either under- or over-specialization. 
 When production is complex, the sign of the marginal welfare effect of equilibrium specialization is determined by the relative magnitudes of the network externality, the appropriability externality across markets, and the search externality. We show that the network externality always dominates the sum of the other two externalities. 
 Therefore, when production is complex, the equilibrium displays over-specialization.
\begin{thm}[Efficiency of the Dynamic Economy] \label{thm:eff_dynam}
   The equilibrium allocation is constrained-efficient if and only if $(N\mathcal{M}(N)-1)\mathbb{\hat{E}}[A(s^\star(z);z)]+x_0=0$, where 
 $\mathcal{M}(N) \in\left(\frac{1}{2},1\right]$. If the production process is complex, i.e., N > 1, then the dynamic equilibrium features over-specialization and under-resilience.
\end{thm}
We conclude that the main result from the static model carries over to the dynamic setting: when the production process is complex, the equilibrium allocation features over-specialization. Moreover, the dynamic model allows us to establish a one-to-one connection between input over-specialization and the under-resilience of supply chains.

\subsection{Dynamic Model with Link Destruction}
In the dynamic model presented so far, we assumed that once final producers experience a disruption, they sever ties with all input providers and must re-establish every supplier relation anew. While this assumption provides analytical transparency, it is quite strong, as in the real world, many supply chain disruptions arise from the breakdown of single supplier relations. Since production inputs are complements, our model is well-suited for studying supply chain disruptions driven by the failure to replace single input providers. To do so, we consider a setting in which each individual supplier relation breaks down stochastically. In this case, the final producer only needs to replace the broken supplier relations to continue production. We show that the same sources of inefficiency and key results of the baseline dynamic model carry over to this more realistic setting.


The main innovation with respect to the baseline dynamic model is that $\delta \in (0,1)$ now denotes the probability that one of the supplier links between intermediate and final producers gets destroyed.
Hence, a disruption happens if the final producer is unable to replace the input providers who separate each period.\footnote{We assume that idle links are severed, so that a final producer unable to source all inputs in a given period loses contact with all its input providers.} 
Let $\rho(\boldsymbol{f})$ be the probability that an active final producer is operational next period:
\begin{align}\label{prob_operation}
    \rho(\boldsymbol{f}) =\prod_{i=1}^N [1-\delta(1-f_i)].
\end{align}
It follows that the probability of being operational conditional on separating with the input provider of good $i$ equals
$f_i \rho(\boldsymbol{f}_{-i})$.
The market tightness in each input market $i$ equals $\theta_i = \frac{\mu_i}{m}$. Let $\nu_i \in [0,1]$ denote the share of searching final producers of input $i$ that were active in the previous period. The mass of searching final producers evolves according to the following law of motion:
%
\begin{align} \nonumber
    \mu_{i (+1)} = & \ \delta \left(1-\left[\left(1-\prod_{j=1}^N f_j\right)(1-\nu_i)+\left(1-f_i \rho (\boldsymbol{f}_{-i})\right)\nu_i\right]\mu_i\right) +\Bigg[\left(1-\prod_{j=1}^N f_j\right)(1-\nu_i) \\ \label{mu_links}
     & +\left(1-f_i \rho (\boldsymbol{f}_{-i})\right)\nu_i\Bigg]\mu_i.
\end{align}
The law of motion (\ref{mu_links}) differs from its baseline counterpart (\ref{mu_dyn}) due to the different probability that searching final producers are active in the current period. This probability is a weighted average of the probabilities that searching final producers successfully source all required inputs, with weights given by the shares that were active and inactive one period earlier. 
Searching final producers who were active one period earlier are more likely to remain so, because they have to find fewer inputs than those who were inactive, i.e., $f_{i}\rho(\boldsymbol{f}_{-i})>\prod_{j=1}^N f_j$. Hence,
the composition of the pool of searching final producers matters for the probability that one of them is active. 
 %
Demand of an intermediate producer of input $i$ choosing specialization $s$ and surplus offered $x$ is
\begin{align} \label{demand_links}
    \mathcal{D}_{i}\left(\mathcal{D}_{i(-1)},s,x\right) = (1-\delta)\rho(\boldsymbol{f}_{-i}) \mathcal{D}_{i(-1)} +
 \theta\lambda \tilde{\mathcal{P}}_{i}(s,x;\boldsymbol{f}_{-i}).
\end{align}
This expression differs from its counterpart in the baseline dynamic model (\ref{demand_dyn}) in two respects. First, the stock of existing customers is discounted by the probability that attached final producers either do not face disruptions of other supplier relationships or they are able to replace them in the current period.
Second, the expression for the trading probability incorporates the differential probability that a searching final producer is active depending on its previous state. The probability that a searching final producer is active upon finding input $i$ in the market equals $\tilde{f}_{i}(\boldsymbol{f}_{-i}) \equiv \nu_{i}\rho(\boldsymbol{f}_{-i})+(1-\nu_{i}) \prod_{n\neq i}f_{n}$. Then, the trading probability of an intermediate producer of good $i$ is given by:
\begin{align} \label{trading_prob_links}  \tilde{\mathcal{P}}_i(s,x;\boldsymbol{f}_{-i}) = \phi(s) e^{-\lambda \bar \phi (1-G(x))} \tilde{f}_{i}(\boldsymbol{f}_{-i}).
\end{align}
Up to the different law of motion of demand, the profit maximization problem of intermediate producers is the same as in the baseline dynamic model, as given by (\ref{pmp_dyn}). In this dynamic environment, intermediate producers meet final producers that are heterogeneous with respect to their other input providers. Yet our surplus/price posting protocol implies that intermediate producers cannot price discriminate across potential new customers.\footnote{Hence, in this dynamic framework, price posting is no longer isomorphic to sealed-bid auctions run by final producers.}

\paragraph{Stationary equilibrium.}
As in the baseline dynamic model, the value of final producers embeds an option value of search, which reflects in the reservation surplus (see Appendix \ref{app:links} for the derivation).
\begin{lemma}[Dynamic Reservation Surplus with Link Destruction]
\label{lem:lem:res_surplus_link}
The equilibrium reservation surplus equals:
\begin{align} \label{res_surplus_links}
  x_0 = -\left(\tilde{\alpha}(N)N-1\right)\mathbb{\hat{E}}[x^\star(z)],
\end{align}
where $\tilde{\alpha}(N) \equiv 1-\frac{\beta (1-\beta) \left[f^N-\delta f \rho(N-1)\right]}{1-\beta (\rho(N)-f^N)} \in (0,1)$ summarizes the influence of the option value of search on the reservation surplus.
\end{lemma}

Since the disruption probability is the same across inputs, the stationary equilibrium is symmetric, i.e., $\mu_i = \mu, \nu_i = \nu,  \forall i$. The equilibrium specialization function 
 is implicitly defined by the following condition:
\begin{align} \nonumber
 & \mathcal{D}(z;N)\bigg[A^\prime (s^\star(z))+[1-(1-\delta)\rho(f;N-1)][1+\beta(1-\delta)\rho(f;N-1)]\frac{\phi^\prime(s^\star(z))}{\phi(s^\star(z))} \\ \label{eq_s_links}
 & (A(s^\star(z);z)-x^\star(z))\bigg] = w q^\prime(s^\star(z)).
\end{align}
Equilibrium specialization (\ref{eq_s_links}) differs from its counterpart in the baseline dynamic model (\ref{eq_s_dyn}) due to (i) the different expression for stationary demand, $\mathcal{D}$, and (ii) the multiplier on the component governing the reduction in the trading probability, as the retention probability of an existing customer compounds the probability that the supplier link is undisrupted, $1-\delta$, with the retention probability of the other inputs, $\rho(f;N-1)$.

Since disruptions happen at the level of single supplier relationships, supply chain resilience has a different expression:
\begin{align} \label{resilience_links}
    \tilde{\mathcal{R}}(f;N) \equiv 
    \nu f \rho(f;N-1) + (1-\nu) f^N.
\end{align}
Supply chain resilience is the weighted average of two terms:  (i) the probability that an active final producer is able to produce, conditional on losing one input provider, $f\rho(f;N-1) = f[1-\delta(1-f)]^{N-1}$, and (i) the probability that an inactive final producer is able to produce, $f^N$. These two terms are weighted by the stationary shares of searching final producers that were active and inactive one period earlier, respectively. This expression of supply chain resilience nests the baseline (\ref{resilience}) as the share of searching inactive final producers approaches one.\footnote{Our assumption that firms missing a single input lose contact with other suppliers is key to the tractability of the resilience function. A more general model would allow firms to lose contact with any number of input providers in $[0, N]$, in which case the resilience function would include a combinatorial term. Since that extension adds complexity without providing additional economic insight, we have chosen not to include it. The analytical characterization of that model is available upon request.} In addition to average specialization, production complexity and search frictions, supply chain resilience in the dynamic model with link destruction is determined by the frequency of supplier link disruptions. More frequent disruptions ($\delta \ \uparrow$) reduce supply chain resilience by raising the likelihood that final producers have to replace multiple input providers at the same time.

Output in this economy is given by:
\begin{align} \label{output_dyn_links}
    Y = \underbrace{\vphantom{\frac{1}{\delta}} \mu(f;N)}_{\text{market size}} \ \  
    \underbrace{\vphantom{\frac{1}{\delta}} \tilde{\mathcal{R}}(f;N)}_{\text{resilience}} \ \ \underbrace{\frac{1}{1-(1-\delta)\rho(f;N-1)}}_{\text{robustness}} \ \
    \underbrace{\vphantom{\frac{1}{\delta}} N\mathbb{\hat{E}}[A(s^\star(z);z)]}_{\text{expected match surplus| active}},
\end{align}
where $\mu(f;N)$ is the stationary share of searching final producers.
Output features the same four components discussed in the baseline dynamic model. Hence, supply chain resilience remains a key determinant of aggregate output and welfare.

\paragraph{Efficiency.}
We examine the efficiency of the equilibrium allocation in the dynamic economy with link destruction by comparing it to the allocation chosen by a social planner. The social planner maximizes the utility of the representative household subject to (i) the law of motion of demand for each intermediate producer (\ref{demand_links}), (ii) the law of motion for the mass of searching firms (\ref{mu_links}), (iii) the stock identity for the share of final producers that were active in the previous period,  and (iv) the labor resource constraint. 
See Appendix \ref{app:links} for the technical details. 

Comparing equilibrium and efficient specialization, we find that the dynamic model with link destruction features the same sources of externality as the baseline dynamic model. However, the search externality no longer has a clear sign. In this setting, the search externality includes an additional component that captures the marginal effect of specialization on welfare via the changing composition of the pool of searching final producers, as summarized by $\nu$.
Since higher specialization increases the share of inactive firms, this additional component pushes in the opposite direction relative to the baseline effect transmitted through market tightness.
We summarize the efficiency properties of this economy in the next theorem.
\begin{thm}[Efficiency of the Dynamic Economy with Link Destruction] \label{prop:eff_dynam_links}
   The equilibrium allocation is constrained-efficient if and only if $\chi_1(f;N) \left(N\widetilde{\mathcal{M}}(N)-1\right)\mathbb{\hat{E}}[A(s^\star(z);z)]+x_0=0$, where $\chi_1(f;N) \in (0,1)$, and $\widetilde{\mathcal{M}}(N) \in \left(\frac{1}{2},1\right]$. There exists a finite level of complexity $\underline{N}>1 $ such that if the production process is sufficiently complex, i.e. $N>\underline{N}$, then the dynamic equilibrium with link destruction features over-specialization and under-resilience. 
\end{thm}
\noindent
In Appendix \ref{app:links}, we characterize the multipliers $\widetilde{\mathcal{M}}(N)$, $\chi_1(f;N)$ and $\chi_2(f;N)$, and show that $N>1$ remains a sufficient condition for the equilibrium allocation to exhibit over-specialization and under-resilience over most of the parameter space.

Overall, we have established that endogenous specialization choices of intermediate producers generate a net negative externality on complementary input providers, even in a dynamic economy where the source of risk in supply chains is the potential disruption of single supplier links. We conclude this section with three remarks. First, we illustrate how final producers can mitigate the inefficiency induced by the network externality through risk-management strategies within our model. Second, we revisit the concept of the \textit{weakest link} to characterize supply chain fragility. Finally, we emphasize that the fundamental wedge between equilibrium and planner allocations reflects a trade-off between the value of production in good times and the economic cost of supply chain disruptions.


\subsubsection{Strategies to Mitigate the Impact of Supplier Breakdowns}
In our model, final producers occasionally halt production when they fail to source one or more key inputs. Because intermediate producers do not account for the foregone surplus from complementary inputs when their own input is missing, they tend to design overly specialized product varieties. As a result, production halts occur more frequently than socially desirable. However, if final producers could continue operating despite disruptions in some supplier relationships, the network externality would not arise. We now examine two key strategies that final producers can adopt to mitigate the impact of supplier breakdowns: holding inventories and sourcing inputs from multiple intermediate producers (multi-sourcing). Appendix \ref{app:inv_mult} provides a formal analysis.

\paragraph{Inventories.}
Suppose that active final producers can choose to purchase more units of each input than needed for current production, setting aside inventories as a buffer.\footnote{The role of inventories in supply chains has recently  been the subject of an active literature. See, for example, \cite{alessandria2023aggregate,ferrari2022inventories,carreras2021increasing,carreras_ferrari}.} If a supplier relationship breaks down, these inventories allow production to continue until a new supplier is found or until their stock is depleted. However, due to search frictions, there remains a positive probability that a final producer is unable to locate a compatible intermediate producer for any given number of periods. As a result, no finite level of inventories can completely hedge final producers from disruptions.

\paragraph{Multi-sourcing.}
Suppose final producers can establish links with multiple compatible suppliers met during a supplier search. Specifically, a final producer might maintain an idle backup link with the second-best compatible supplier at a fixed per-period cost. If the link with the best supplier breaks down first, the final producer can immediately switch to the backup supplier and continue production \citep{Elliott2022_Fragility,grossman2023resilience}. However, due to search frictions, the probability of finding multiple compatible intermediate producers is less than one. As a result, not all final producers are able to multi-source all inputs. It follows that multi-sourcing alone cannot completely hedge final producers from disruptions. 
 
\begin{remark} [Inventories \& Multi-sourcing] \label{rem:inv_mult}
    The ability of final producers to hold inventories or source the same input from multiple suppliers can alleviate, but not eliminate, the network externality.
\end{remark}


\subsubsection{Weakest Links and Supply Chain Fragility}

Much of the policy and academic discussion on supply chain fragility has focused on the concept of the \textit{weakest link} \citep{kremer1993ring,levine2012production,jones2011weak}. Due to the complementarities among production inputs, a supply chain is as fragile as its weakest link. Static models of supply chains identify the weakest link, $\hat{j}$, as the input \textit{provider} with the highest disruption probability:
\[
\hat{j}^{\text{static}} = \arg\max_j \delta_j.
\]
For instance, if two final producers source inputs from suppliers located in an area prone to natural disasters or subject to geopolitical tensions, their supply chains would be considered equally fragile. In contrast, our model identifies the weakest link as the \textit{input} that is most likely to be missing. This is a combination of the exogenous disruption probability of the current input provider, $\delta_j$, and the endogenous compatibility probability of alternative suppliers, $\bar\phi_j$:
\[
\hat{j}^{\text{dynamic}} = \arg\max_j \delta_j \left[1 - f(\bar{\phi}_j)\right].
\]
In other words, the weakest link is the input for which the probability that its current provider is disrupted \textit{and} no replacement is found is highest.

Accordingly, even two supply chains sourcing inputs from suppliers in the same ``risky'' area may differ significantly in fragility, depending on how difficult it is to replace those suppliers if needed. Supply chain fragility depends on the joint distribution of supplier disruption probabilities and input finding probabilities, reflecting endogenous specialization decisions.
Since optimal specialization is decreasing in the disruption probability, inputs whose providers are frequently disrupted may not be the \textit{weakest link}, as they endogenously choose to specialize less.

\begin{remark}[Weakest Link]
\label{rem:weakest}
Suppose that the disruption probability  $\delta_j$ is heterogeneous across input markets $j$, while all other parameters are symmetric. The input with the highest disruption probability is not necessarily the weakest link of the supply chain: it is possible that $\delta_j \left[1 - f(\bar{\phi}_j)\right]>\delta_k \left[1 - f(\bar{\phi}_k)\right]$ \ while \ $\delta_j<\delta_k$.
\end{remark}

\subsubsection{Productivity vs Resilience: A Dynamic Trade-off}
Our key result from Theorem \ref{prop:eff_dynam_links}  is that endogenous specialization distorts the efficient balance between productivity in good times and resilience in bad times -- that is, the speed of recovery from disruptions. Figure (\ref{fig:resil_vs_prod}) summarizes this result graphically.
When operations run smoothly, supply chains in the market equilibrium are \textit{more} productive than socially desirable, due to the high specialization of the inputs used in final production. However, this higher productivity comes at a cost: when a supplier relationship breaks down, the same specialization that boosts productivity makes recovery slower than efficient, delaying the return to normal operations.
\begin{figure}[ht]
\begin{center}
\caption{Productivity-Resilience Trade-off}
\label{fig:resil_vs_prod}
\includegraphics[width=0.6\textwidth,keepaspectratio]{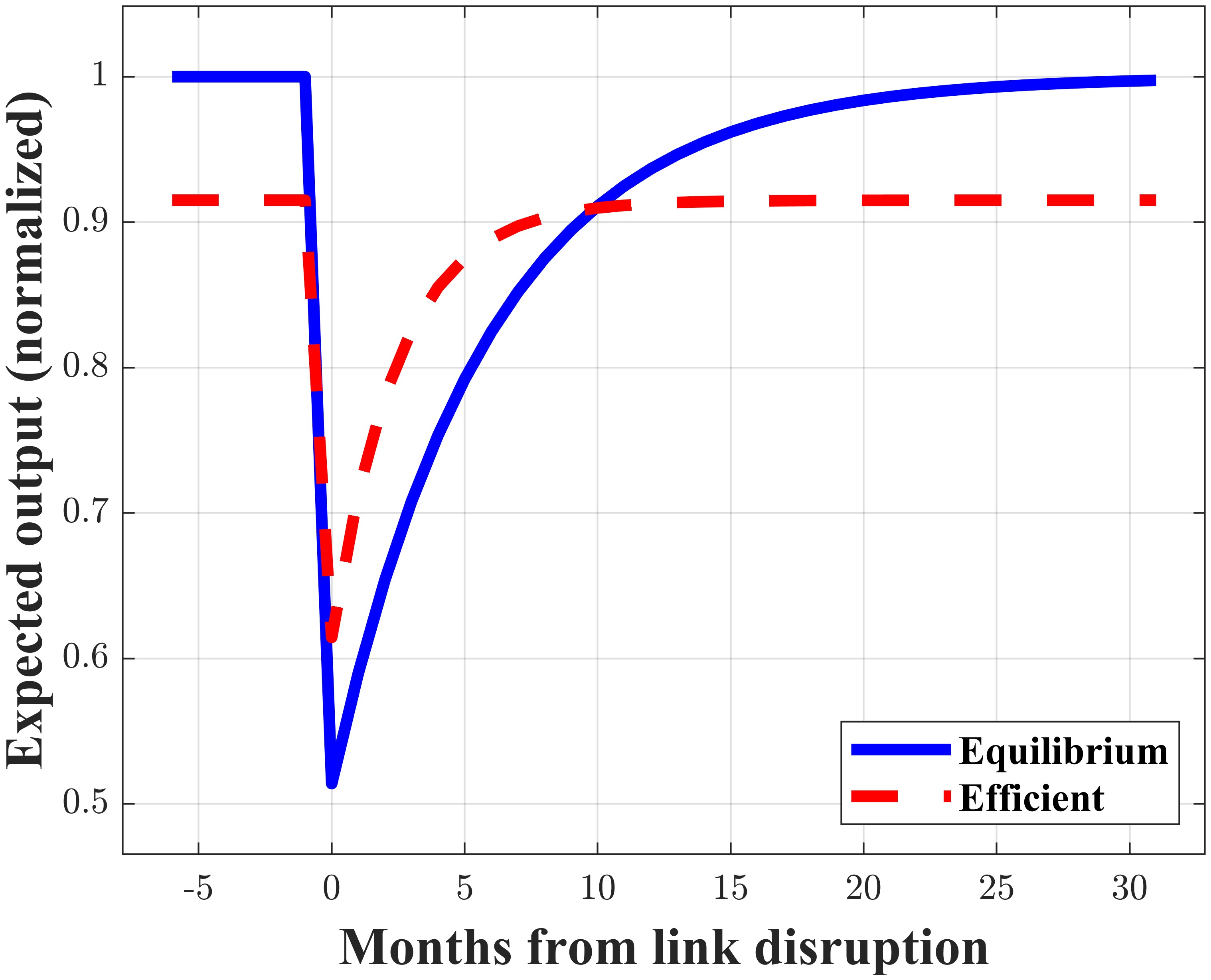}
\begin{minipage}{0.7\textwidth} \scriptsize{} \textit{Note}:  The graph depicts the expected output of an active final producer conditional on facing one link disruption shock in month $0$ and no disruptions in other months. The background data are generated from a quantitative version of the dynamic model with link destruction in steady state, assuming a monthly link disruption probability $\delta=5\%$, a monthly number of potential suppliers screened  $\lambda=24$, and a number of critical inputs sourced in the market $N=3$. The equilibrium input finding probability equals $f=54\%$.
\end{minipage} 
\end{center}
\end{figure}

This distortion in the efficient balance between productivity and resilience is driven by the network externality: single intermediate producers do not internalize how their specialization choices affect the likelihood that final producers 
source all the required inputs, and therefore the unrealized value added from complementary
inputs if production halts.  Crucially, this externality emerges precisely when a final producer is unable to find a compatible supplier -- meaning it cannot be corrected through a bilateral contractual agreement, as there are simply no two counterparties to sign one.
 On the other hand, 
 final producers can partially offset this externality through strategies to mitigate the impact of supplier breakdowns, such as holding inventories and multi-sourcing. In the next section, we study how a constrained social planner can fix the network externality by inducing specific terms-of-trade in firm-to-firm transactions.

\section{Normative Analysis}\label{normative_sec}

In this section, we study two normative approaches to address the inefficiency induced by equilibrium specialization choices. First, we characterize the optimal policy that decentralizes the constrained-efficient allocation in the economies we have studied in the previous sections. Next, we consider a planner setting \textit{standards} or \textit{minimum interoperability} requirements.

\subsection{Optimal Policy}

\paragraph{Static economy.} 
From Theorem \ref{thm:over-specialization}, it follows that intermediate producers offer too high a surplus $x^\star(z)$ to final producers in equilibrium. The constrained-efficient surplus offered equals:
\begin{align} \nonumber
    \mathcal{X}(z)=& \ x^\star(z)-\left(1-f(z)\right) \left[(N-1) \mathbb{\hat{E}}[A(s^\star(\tilde{z});\tilde{z})]+x_0\right] \\ \label{eff_surplus_static}
    =& \ -\left(1-f(z)\right) (N-1) \mathbb{\hat{E}}[A(s^\star(\tilde{z});\tilde{z})] + f(z)\mathbb{E}_{\max\{\tilde{z}\}|\tilde{z}\leq z}[A\left(s^\star(\tilde{z});\tilde{z}\right)],
\end{align}
where the second equality follows from substituting for the equilibrium surplus offered (\ref{eq_x_solved}) and reservation surplus (\ref{eq:res_surplus}). Notice that the constrained-efficient surplus offered has the same structure as the equilibrium surplus offered: it consists of a weighted average between a \textit{shadow} reservation surplus, $-(N-1) \mathbb{\hat{E}}[A(s^\star(\tilde{z});\tilde{z})]$, and the expected match surplus from the second-best compatible intermediate producer in contact with a final producer. The weights correspond to the probability that a final producer is not in contact with an alternative compatible intermediate producer with productivity lower than $z$, $1-f(z)$, and at least one such producer, $f(z)$, respectively. However, the constrained-efficient and equilibrium surplus differ in the level of the reservation surplus: in equilibrium, final producers are willing to bear losses on each input line as large as the expected surplus \textit{offered} by the other input lines, $-(N-1) \mathbb{\hat{E}}[x^\star(\tilde{z})]$, whereas a social planner would set the reservation surplus equal to the expected \textit{match} surplus from the other input lines. These observations allow us to establish the following result.
%
\begin{prop}[Optimal Subsidy Schedule in a Static Economy]\label{prop:norm_static}
The social planner can decentralize the constrained-efficient allocation with a lump-sum transfer to active final producers:
\begin{align*}
    T^\star= (N-1) \left(\mathbb{\hat{E}}[A(s^\star(\tilde{z});\tilde{z})]-\mathbb{\hat{E}}[x^\star(\tilde{z})]\right)>0, \quad \forall N>1.
\end{align*}
\end{prop}
Proposition \ref{prop:norm_static} characterizes the lump-sum transfer to active final producers that decentralizes the constrained-efficient allocation. The intuition is as follows. The social planner aims to increase the surplus appropriated by intermediate producers when final production takes place, thereby giving them more \textit{skin in the (final production) game}. As intermediate producers appropriate a larger share of surplus, the opportunity cost of marginally increasing specialization rises, while the marginal benefit from higher match surplus remains unchanged. To achieve this, the planner can implement a lump-sum subsidy to \textit{active} final producers. Because final producers receive additional compensation upon successfully sourcing all required inputs, they optimally lower their reservation surplus---or, equivalently, raise their reservation price---for each input. Anticipating this lower reservation surplus, intermediate producers reduce their surplus offered to final producers.

 Since the sensitivity of the surplus offered to the reservation surplus increases with the probability that a final producer does not meet alternative compatible suppliers with lower productivity, the reduction in surplus offered is heterogeneous across intermediate producers. Specifically, lower-productivity intermediate producers reduce their surplus offered by more. As a result, the lump-sum subsidy induces a decline in specialization that is concentrated among low-productivity intermediate producers and gradually fades out for high-productivity ones. Figure (\ref{fig:opt_subs}) illustrates this differential incidence: while the constrained-efficient allocation calls for higher (lower) prices (specialization) across \textit{all} intermediate inputs, the adjustment is substantially larger for low-productivity firms.
\begin{figure}[ht]
\begin{center}
\caption{Implicit Targeting Principle - Optimal Lump-Sum Subsidy}
\label{fig:opt_subs}
\includegraphics[width=0.6\textwidth,keepaspectratio]{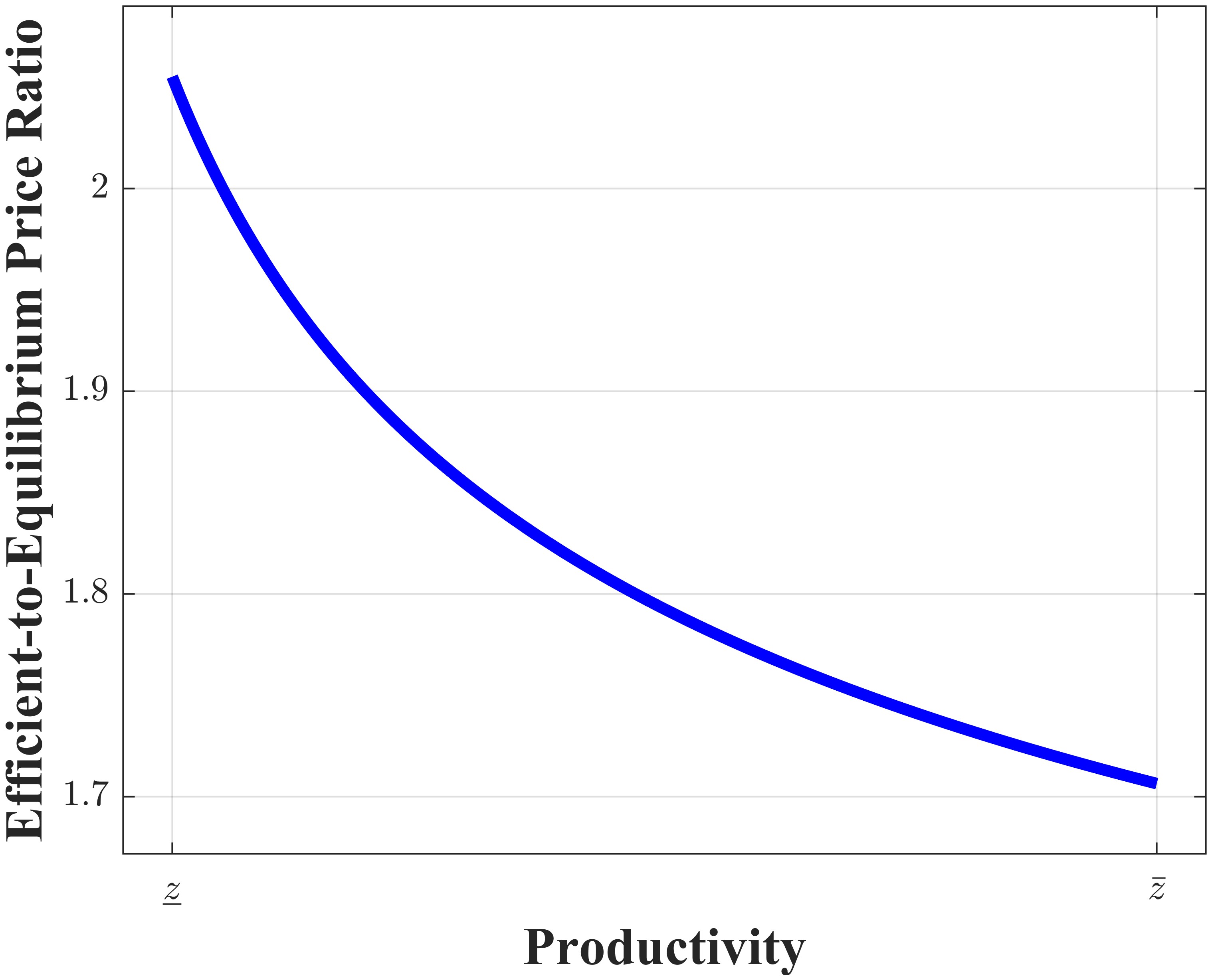}
\begin{minipage}{0.7\textwidth} \scriptsize{} \textit{Note}:  The graph shows the efficient price $\mathbf{P}(z) \equiv A(\Sc(z))-\mathcal{X}(z)$ scaled by the equilibrium price $p^\star(z) \equiv A(s^\star(z))-x^\star(z)$ over the support of the productivity distribution. Therefore, the market equilibrium prices are implicitly a horizontal line with an intercept equal to 1. The background data are generated from a quantitative version of the dynamic model with link destruction in steady state, assuming a monthly link disruption probability $\delta=5\%$, a monthly number of potential suppliers screened $\lambda=24$, and a number of critical inputs sourced in the market $N=3$. The equilibrium input finding probability equals $f=54\%$.
\end{minipage} 
\end{center}
\end{figure}

The economic rationale for this differential incidence is twofold. First, low-productivity firms are farther from their socially desirable specialization level, because the intensity of the net (negative) externality depends on the probability that a final producer fails to find a compatible supplier with lower productivity, which is particularly high for such firms. As a consequence, they offer excessively high surplus relative to the efficient benchmark, leaving them with too little skin in the final-production game. Second, under price posting, the price set by a firm with productivity $z$ depends on the specialization choices of all potential competitors with productivity $z^\prime < z$. Distorting the prices of the least productive firms, therefore, shifts the entire price distribution, whereas distorting the prices of the most productive firms does not trickle down to less productive competitors.

\paragraph{Dynamic economy.}
By combining (\ref{eq_s_dyn}) with (\ref{eff_s_dyn}), we can follow the same steps as in the static model to obtain the surplus offered that decentralizes the constrained-efficient allocation:
\begin{align} \nonumber
    \mathcal{X}(z)=& \ x^\star(z)-\left(1-f(z)\right) \left[\left(\mathcal{M}(N) N-1\right)\mathbb{\hat{E}}[A(s^\star(\tilde{z});\tilde{z})]+x_0\right] \\ \label{eff_surplus_dynamic}
    =& \ -\left(1-f(z)\right) (\mathcal{M}(N)N-1) \mathbb{\hat{E}}[A(s^\star(\tilde{z});\tilde{z})] + f(z)\mathbb{E}_{\max\{\tilde{z}\}|\tilde{z}\leq z}[A\left(s^\star(\tilde{z});\tilde{z}\right)],
\end{align}
where the second equality follows from substituting for the equilibrium surplus offered (\ref{eq_x_solved}) and reservation surplus (\ref{res_surplus_dyn}).
This result allows us to derive the optimal lump-sum subsidy in the dynamic context.
\begin{prop}[Optimal Subsidy Schedule in a Dynamic Economy]\label{prop:norm_dynamic}
The social planner can decentralize the constrained-efficient allocation with a lump-sum transfer to active final producers:
\begin{align*}
    T^{\star}= \frac{(\mathcal{M}(N)N-1)\,\mathbb{\hat{E}}[A(s^\star(\tilde{z});\tilde{z})]-\bigl(\alpha(N)N-1\bigr)\,\mathbb{\hat{E}}[x^\star(\tilde{z})]}{\alpha(N)}>0, \quad \forall N>1,
\end{align*}
where $\alpha(N) \equiv  1-\frac{\beta(1-\delta) f^N}{1-\beta(1-\delta)(1-f^N)}$.
\end{prop}
Proposition \ref{prop:norm_dynamic} shows that the optimal subsidy in the dynamic economy factors in the additional inefficiency brought about by the search externality -- subsumed by the multiplier $\mathcal{M}(N) \in \left(\frac{1}{2},1 \right]$. For the dynamic model with link destruction, all the results carry over by substituting $\chi_1(f;N)(\widetilde{\mathcal{M}}(N)N-1)$ and $\tilde{\alpha}(N)$ in place of $(\mathcal{M}(N)N-1)$ and $\alpha(N)$, respectively.

Overall, our results show that a properly calibrated lump-sum subsidy to active final producers can decentralize the constrained-efficient level of specialization. For the government budget to be balanced, the lump-sum subsidy can be financed through a lump-sum tax on households.

Finally, note that the possibility of decentralizing the efficient allocation by redistributing surplus across trading partners relates to the well-known \textit{Hosios condition} in matching models \citep{hosios1990efficiency}.\footnote{Rather than resorting to a lump-sum transfer to active final producers, the social planner may also decentralize the constrained-efficient allocation with a targeted transaction subsidy schedule $\tau^\star(z)$ such that the price of the transactions is given by $\mathbf{P}(z)=p^\star(z)+\tau^\star(z)$, where
$\tau^\star(z)= (1-f(z)) T^\star>0$ in the static economy, and
$\tau^\star(z)= (1-f(z)) \  \alpha(N) T^\star>0$ in the dynamic economy.} However, the economic intuition behind the two results is very different. In matching models, the inefficiency of the equilibrium comes from \textit{thin} and \textit{thick market} effects governed by the shape of the matching function. Accordingly, efficient entry of sellers obtains when the surplus share accruing to them equals their marginal contribution to matching (\textit{matching elasticity}). In our framework, there is no matching function and the number of meetings per intermediate producer is either exogenously given (in the static model) or determined by the relative mass of searching agents on each side of the market (in the dynamic model). Consequently, the matching elasticity is zero in the static model and proportional to the search externality in the dynamic model. On the contrary, the key source of inefficiency -- the network externality -- arises from the complementarity among production inputs. Hence, efficient specialization requires that the surplus share accruing to intermediate producers also compensates for the foregone surplus from complementary inputs caused by higher specialization of their own variety. \cite{mangin2021efficiency} provide a framework to reconcile these two different economic forces through the notion of a \textit{generalized} Hosios condition, according to which the efficient surplus share equals the sum of matching elasticity \textit{and surplus elasticity}.

\subsection{Standards and Minimum Interoperability}
The optimal policy to decentralize the constrained-efficient allocation requires subsidizing profit-making firms by taxing households. This may be politically infeasible. A common non-fiscal, regulatory approach used by governments is setting standards or minimum interoperability requirements.\footnote{A recent example of such policies is the implementation of the EU Common Charger Directive in 2024, imposing compatibility across producers on the use of USB-C chargers.} 
Through the lens of our model, we interpret these standards as caps on the degree of incompatibility, or, equivalently, a maximum level of specialization. Formally, we consider a Ramsey planner choosing a standard $\bar s$ associated with some minimum level of compatibility $\underline{\phi}$. The planner is constrained by each agent's equilibrium choices given the policy. We first prove that standards are welfare-improving in our framework. Then, we provide suggestive evidence that standards can bring the equilibrium significantly closer to the constrained-efficient allocation.
\begin{prop}[Standards Improve Welfare] 
\label{prop:stds}
   If the production process is complex,  setting a standard or a minimum interoperability requirement is welfare-improving. 
\end{prop}
\noindent
The intuition behind Proposition \ref{prop:stds} is straightforward. Since all intermediate producers design overly specialized product varieties in equilibrium, constraining the degree of specialization of the most specialized producers brings about a welfare improvement. 

Although uniform standards can improve welfare, they remain second-best policies in our framework, since implementing the efficient allocation requires a transfer to final producers, as shown in Propositions \ref{prop:norm_static}-\ref{prop:norm_dynamic}.
In Appendix \ref{app:norm} we solve the social planner problem when only standards are available as policy tools. The condition characterizing the optimal standard for each input comprises two components. On the one hand, the optimal standard balances the average marginal benefit and marginal cost of loosening the standard for constrained intermediate producers, mirroring the privately optimal specialization condition \eqref{eq_s}. On the other hand, it additionally internalizes the marginal impact of aggregate specialization on network and appropriability externalities across markets. Equation \eqref{opt_std} in Appendix \ref{app:norm} provides a formal characterization of these two components.

The extent to which standards can bring the equilibrium allocation closer to the efficient one is inherently a quantitative question, which hinges upon the specific market structure.
Figure (\ref{fig:opt_stds}) shows that setting standards can be a powerful instrument -- at least for some market structures.
Figure (\ref{fig:spec_stds}) sheds light on the cross-sectional impact of setting standards along the productivity distribution. The optimal standard implies that the most productive intermediate producers are less specialized (more compatible) than efficient, while the least productive are more specialized (less compatible) than in the first-best allocation. Hence, as long as specialization is increasing in productivity, the targeting principle of the optimal standard stands in contrast with that of the first-best policy (see Figure (\ref{fig:opt_subs})): rather than mainly targeting the specialization decisions of low-productivity intermediate producers (and let them spill over higher up in the productivity distribution through the price posting mechanism), the optimal standard brings aggregate specialization closer to its constrained-efficient level by distorting the specialization decisions of high-productivity intermediate producers.
In Table \ref{tab:stds}, we explore how both the optimal level of standards and their welfare effects vary with the structural parameters of the economy.
Implementing optimal standards is more effective
in bringing the economy closer to its efficient benchmark the higher the search efficiency, the
lower the disruption probability, and the lower the production complexity.

\begin{figure}[htb]
\begin{center}
\caption{Welfare and Compatibility: first-best vs standard}

\begin{subfigure}{0.485\textwidth}
\includegraphics[width=\textwidth,keepaspectratio]{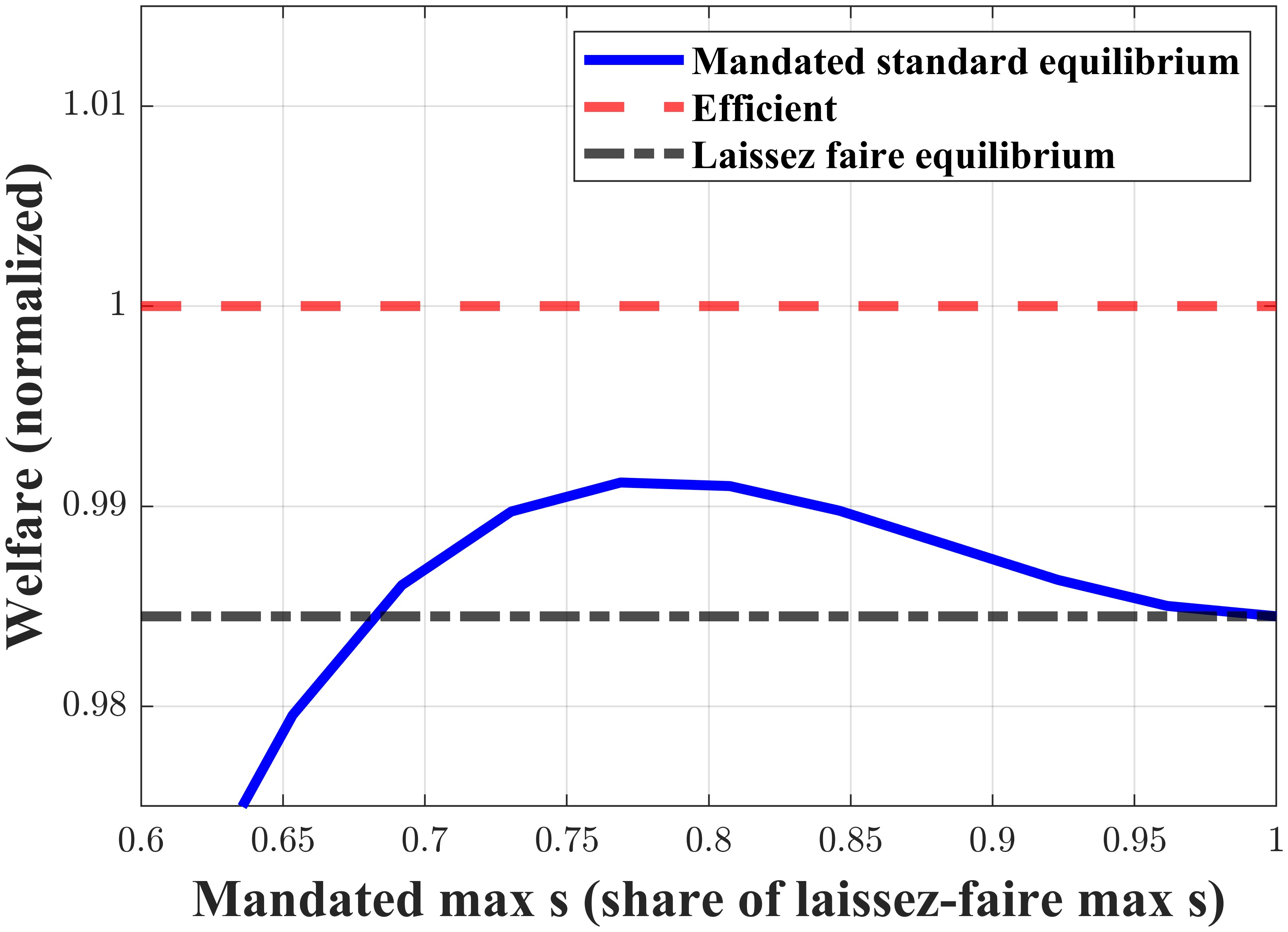}  
\caption{Welfare}
\label{fig:opt_stds}
\end{subfigure}
\begin{subfigure}{0.485\textwidth}
\includegraphics[width=\textwidth,keepaspectratio]{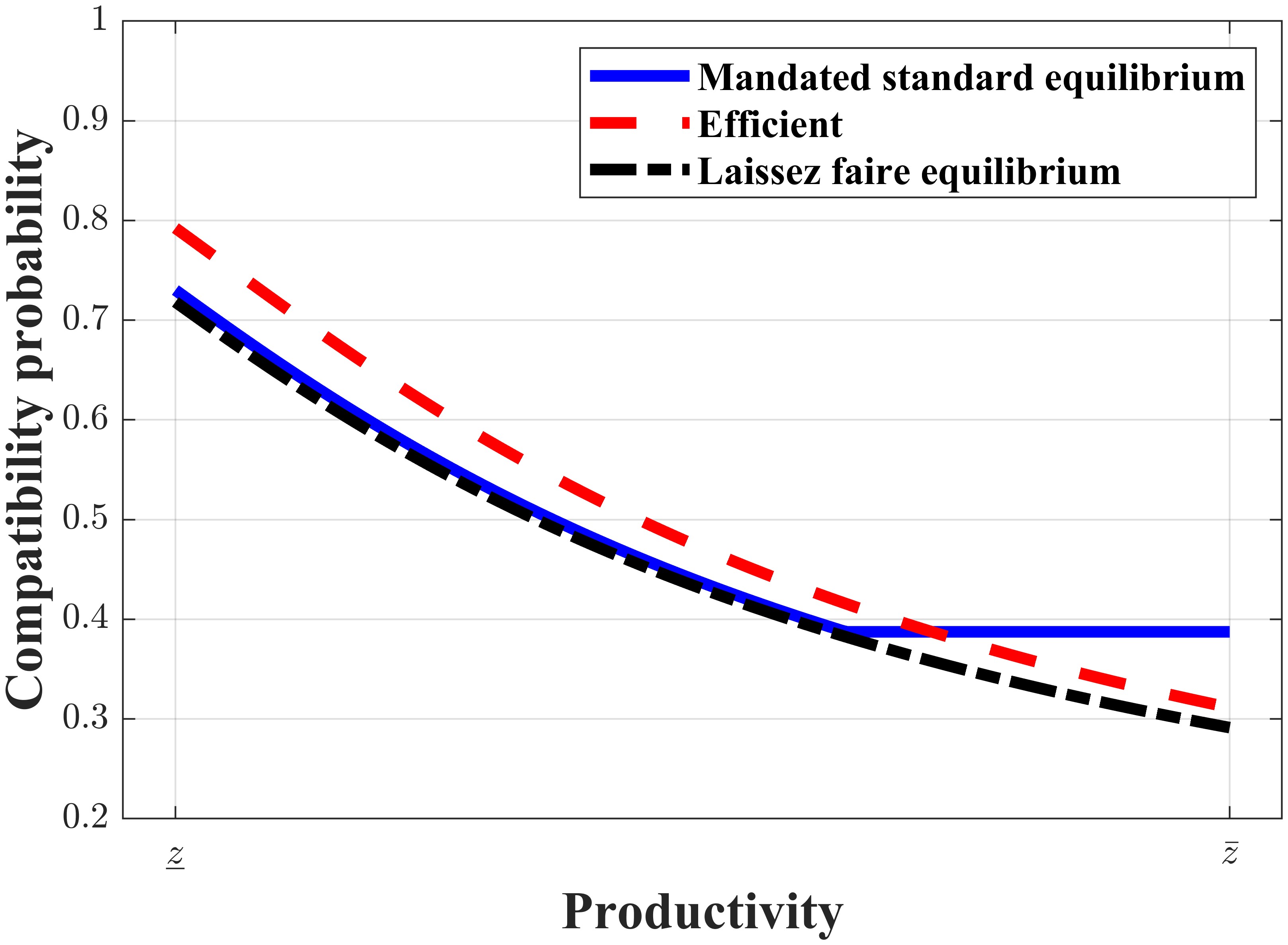}  
\caption{Compatibility}
\label{fig:spec_stds}
\end{subfigure}

\begin{minipage}{\textwidth} \scriptsize{} \textit{Note}:  Panel a) reports in solid blue aggregate steady-state welfare $\mathcal{W}$ for each mandated standard (i.e., cap on specialization) ranging from the lowest to the highest specialization in laissez faire equilibrium. The upper bound in dashed red represents the constrained-efficient steady-state welfare, the lower bound in dot-dashed black represents the equilibrium welfare with no standards in place (laissez faire). Welfare is measured in consumption equivalent variation. Panel b) displays the compatibility probability function under three scenarios: in the mandated standard equilibrium, $\phi(s^\star(z);\bar{s})$ (solid blue); in the constrained-efficient allocation, $\phi(\mathcal{S}(z))$ (dashed red); and in laissez-faire equilibrium $\phi(s^\star(z))$ (dot-dashed black).
The background data are generated from a quantitative version of the dynamic model with steady-steady behavior (see Section \ref{sec:dyn_ss}) assuming a monthly disruption probability $\delta=10\%$, a monthly number of potential suppliers screened $\lambda=3$, and a number of critical inputs sourced in the market $N=6$. The equilibrium input finding probability equals $f=75\%$. Under steady-state behavior, the constrained-efficient steady-state welfare is the maximum attainable.
\end{minipage} 
\end{center}
\end{figure}




\section{Relation to Empirical Evidence}  

We conclude by highlighting how our model helps rationalize existing empirical evidence. 
First, our modeling of supply chains as complex final production processes involving multiple complementary inputs is consistent with the shock-propagation patterns observed after the 2011 Tohoku Earthquake, as documented by \cite{boehm2019propagation} and \cite{carvalho2021supply}. In particular, \cite{boehm2019propagation} finds that ``\textit{the strong complementarity across material inputs implied that non-Japanese imported input use also fell nearly proportionately, thereby propagating the shock to other upstream firms}.'' In our dynamic model with link destruction, a shock that disrupts a supplier link reduces, on average, the sales of other suppliers to the same final producer, because the latter is sometimes unable to replace the missing input. This empirical evidence is therefore consistent with the propagation mechanism underpinning our horizontal network externality. 
Similarly, using U.S.\ Compustat data, \cite{barrot2016input} shows that shocks to suppliers propagate to firms sharing common customers with the disrupted firms only when the latter produce relationship-specific inputs that are not easily substitutable. In the same vein, \cite{Sanz2023resilience} documents that American firms with global supply chains reduce the number of shipments of ``specialized" (i.e., with a few potential suppliers worldwide) inputs by $73\%$ in the month following a flood hitting their input provider -- the reduction being negligible for generic inputs due to supplier switching. These findings provide an important validation of our dynamic notion of resilience: supplier-link disruptions translate into supply chain disruptions only if the final producer is unable to replace the broken link. The probability of replacement is inversely related to the degree of specialization of the respective input.  

Using firm-to-firm transaction data from India, \cite{khanna2022india} document that (i) robustness of final production is positively correlated with the specialization of intermediate inputs, and (ii) final producers whose suppliers were more exposed to Covid-19 restrictions reduced output and input purchases, as some supplier links were severed and replacements could not be found. Fact (i) aligns with the endogenous choice of input specialization in our dynamic models (see equations (\ref{eq_s_dyn}) and (\ref{eq_s_links})): the higher the robustness of the supplier link, the higher the optimal specialization. Fact (ii) reinforces the idea that replacing missing input providers is costly (in terms of final output) and time-consuming.  



\section{Conclusions}
In this paper, we study a model of frictional supply chain formation in which heterogeneous intermediate producers choose the specialization of their goods. Higher specialization increases value added but reduces the pool of compatible final producers. Final producers operate complex supply chains, defined as production processes requiring multiple complementary inputs. This framework delivers a precise and welfare-relevant notion of supply chain resilience: the probability that final production is restored in a given period
following a disruption. In the model, supply chain resilience is negatively affected by the average specialization of goods, the extent of search frictions, the frequency of disruptions, and the complexity of the production process. 

We study the efficiency properties of equilibrium specialization. We show that the combination of endogenous specialization and price-posting in frictional markets with simultaneous search implies that appropriability and business-stealing externalities perfectly offset each other within input markets. However, individual intermediate producers do not internalize how their specialization choices affect the likelihood that final producers
source all the required inputs, and therefore the unrealized value added from complementary inputs if production halts.
This network externality induces over-specialization in equilibrium. As a result, supply chains are excessively productive during normal times, when operations run smoothly.  Yet, over-specialization also makes supply chains less resilient than is socially efficient: when a disruption occurs, production takes too long to recover. A social planner would re-balance this trade-off between productivity and resilience, placing greater weight on the latter. We characterize how a social planner can decentralize efficient supply chain resilience through a lump-sum subsidy to active final producers. Restricting the set of policy tools to budget-neutral instruments, we show that setting compatibility standards is an effective way to mitigate over-specialization.

This paper sheds light on the sources of inefficiencies that are likely to affect supply chain resilience. 
Our analytical results call for further work to quantify the qualitative mechanisms we have highlighted. An interesting avenue for future research is to revisit the social gains from vertical and horizontal integration within supply chains, accounting for the network externality.

\bibliographystyle{apalike}
\bibliography{literature}
\clearpage
\newpage

\appendix
\newpage

\renewcommand{\thesection}{\Alph{section}}
\renewcommand{\thefigure}{A.\arabic{figure}}
\renewcommand{\thetable}{A.\arabic{table}}
\renewcommand{\thelemma}{A.\arabic{lemma}}
\renewcommand{\theprop}{A.\arabic{prop}}
\renewcommand{\theequation}{A.\arabic{equation}} 

\setcounter{figure}{0}
\setcounter{table}{0}
\setcounter{prop}{0}
\setcounter{lemma}{0}
\setcounter{equation}{0} 

\renewcommand{\theappprop}{A.\arabic{appprop}}
\renewcommand{\theapplemma}{A.\arabic{applemma}}

{
\begin{center}

{\Huge {Appendix} }
\end{center}}

\justifying

\section{Additional Material and Derivations}
\label{sec:additional_material}

\subsection{Definition of Expectation Operators}
\label{app:static}
Throughout the text, the expectation operator $\mathbb{E}[.]$ refers to the productivity distribution $\Gamma(z)$:
$$\mathbb{E}[X] \equiv \int_{\underline{z}}^{\bar{z}} X(z) \gamma(z)dz.$$
The expectation operator $\mathbb{E}_{\max\{\tilde{z}\}}$ refers to the productivity distribution of the highest-surplus-offering compatible intermediate producer:
$$\mathbb{E}_{\max\{\tilde{z}\}}[X] \equiv \int_{\underline{z}}^{\bar{z}} X(\tilde{z}) e^{-\lambda\hat{\phi}(\tilde{z},\bar{z})} \lambda\phi(s(\tilde{z}))\gamma(\tilde{z})d\tilde{z}.$$
The expectation operator $\mathbb{E}_{\max\{\tilde{z}\}|\tilde{z}\leq z}$ refers to the productivity distribution of the highest-surplus-offering compatible intermediate producer with lower productivity than $z$ contacted, conditional on its presence:
$$\mathbb{E}_{\max\{\tilde{z}\}|\tilde{z}\leq z}[X] \equiv \int_{\underline{z}}^z X(\tilde{z}) \frac{e^{-\lambda\hat{\phi}(\tilde{z},z)} \lambda\phi(s(\tilde{z}))\gamma(\tilde{z})}{f(z)}d\tilde{z}.$$
Finally, the expectation operator $\mathbb{\hat{E}}[.]$ refers to the productivity distribution of active matches:
$$\mathbb{\hat{E}}[X] \equiv \frac{\mathbb{E}_{\max\{\tilde{z}\}}}{f} = \int_{\underline{z}}^{\bar{z}} X(\tilde{z}) \frac{e^{-\lambda\hat{\phi}(\tilde{z},\bar{z})} \lambda\phi(s(\tilde{z}))\gamma(\tilde{z})}{f}d\tilde{z}.$$

\subsection{Generalized Production Function} \label{app:ces_prod_function}
In this section, we show how our final production structure can be derived from a more general class of models characterized by complementarity across indivisible production inputs. 

Each final producer $i$ operates a constant-elasticity-of-substitution (CES) production function combining $N>1$ complementary inputs into a consumption good:
\begin{align} \label{def:prod_fnct_ces}
Y_i = \left(\sum_{j=1}^N A_j \right)^{-\frac{1}{\eta}+1}\left(\sum_{j=1}^N A_j y_j^\eta \right)^{\frac{1}{\eta}},
\end{align}
where $\eta \leq 0$. The first multiplier makes sure that final producers exhibit no (dis-)like for variety. \\ \indent
Final producers solve the following profit maximization problem:
\begin{align}
    \pi_i = \max_{y_j \in \mathbb{N} \mathbbm{1}_j} \ Y_i-\sum_{j=1}^N p_j y_j.
\end{align}
A final producer can demand any quantity from an intermediate producer of good 
$j$ with whom she is connected, i.e., for which $\mathbbm{1}_j=1$.
Demand for each intermediate good is subject to an integer constraint, reflecting the indivisibility of goods typically traded in supply chains (e.g., engines, tires, chassis). The optimal demand for an intermediate good $j$ reads:
\begin{align*}
    y_j^\star: \begin{cases}
        Y(y_j^\star;\boldsymbol{y}_{-j}^\star)-Y(y_j^\star-1;\boldsymbol{y}_{-j}^\star) \geq p_j, \\
        Y(y_j^\star+1;\boldsymbol{y}_{-j}^\star)-Y(y_j^\star;\boldsymbol{y}_{-j}^\star) < p_j.
    \end{cases}
\end{align*}
If the distribution of quality-adjusted prices, $p_j/A_j$, has sufficiently low variance -- which is the case in equilibrium if the variance of the intermediate productivity distribution is small --, then the optimal demand for each intermediate good is price-insensitive -- just like in our baseline Leontief specification of the production function. 
Hence, output varies across final producers only because of the value added of inputs sourced: 
$$y_j^\star = y^\star \ \forall j \implies Y_i = \mathcal{Q}_i y^\star, \ \mathcal{Q}_i \equiv \sum_j A_j.$$
According to this generalized CES production function, the input value added $A_j$ increases the \textit{quantity} produced by the respective final producers, being the quality of the consumption good constant. Hence, this economy differs from our baseline only in terms of quantities produced, while being observationally equivalent as far as total expenditures and welfare are concerned.

\subsection{Specialization Function} \label{app:spec_function}
\begin{figure}[H]
\begin{center}
\caption{Specialization function: the role of the productivity distribution}
 \label{fig:spec_fnct} \centering
    \begin{minipage}{0.3\textwidth}
        \centering
         \includegraphics[width=\textwidth]{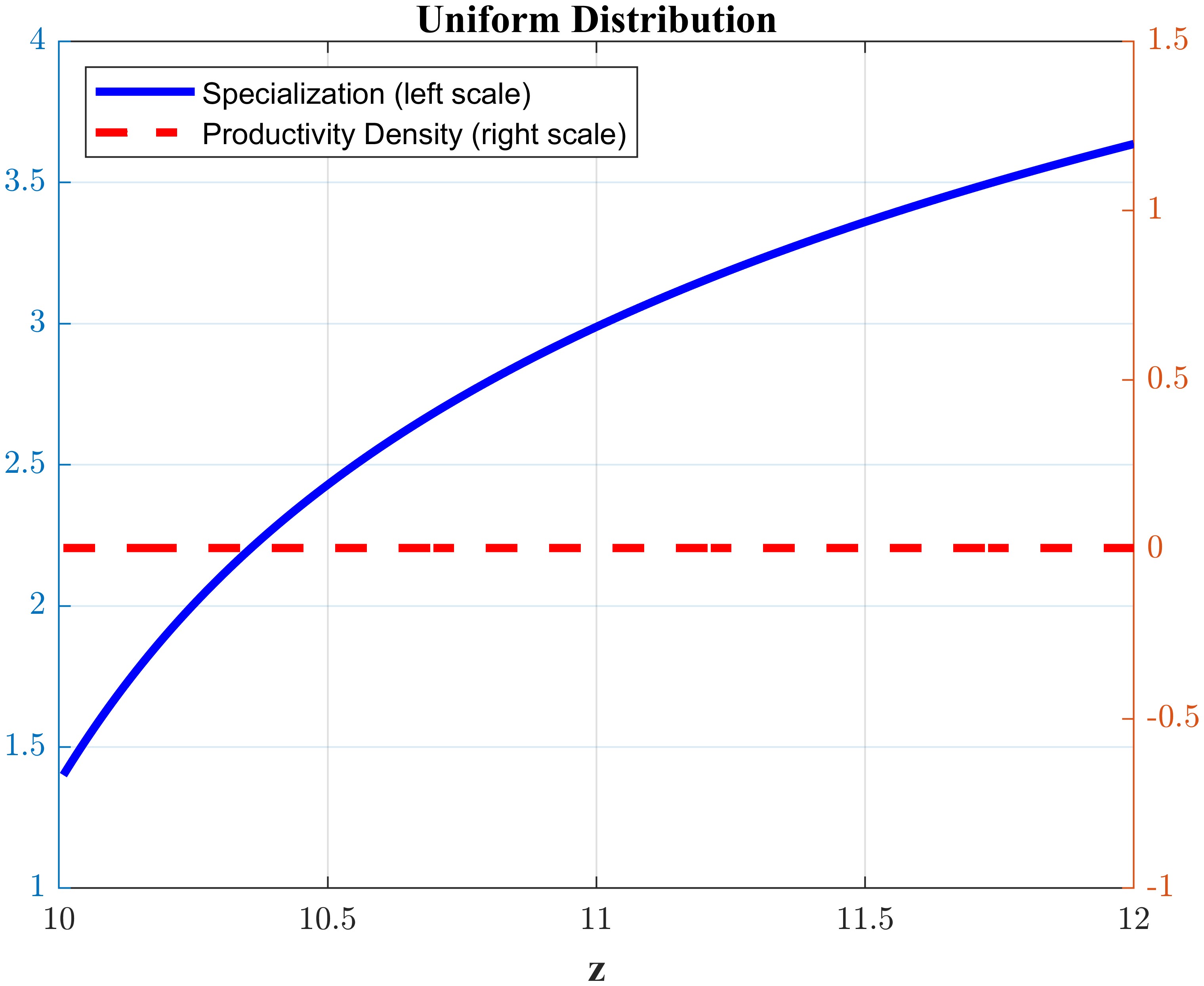} 
    \end{minipage}\hfill
    \begin{minipage}{0.3\textwidth}
        \centering
         \includegraphics[width=\textwidth]{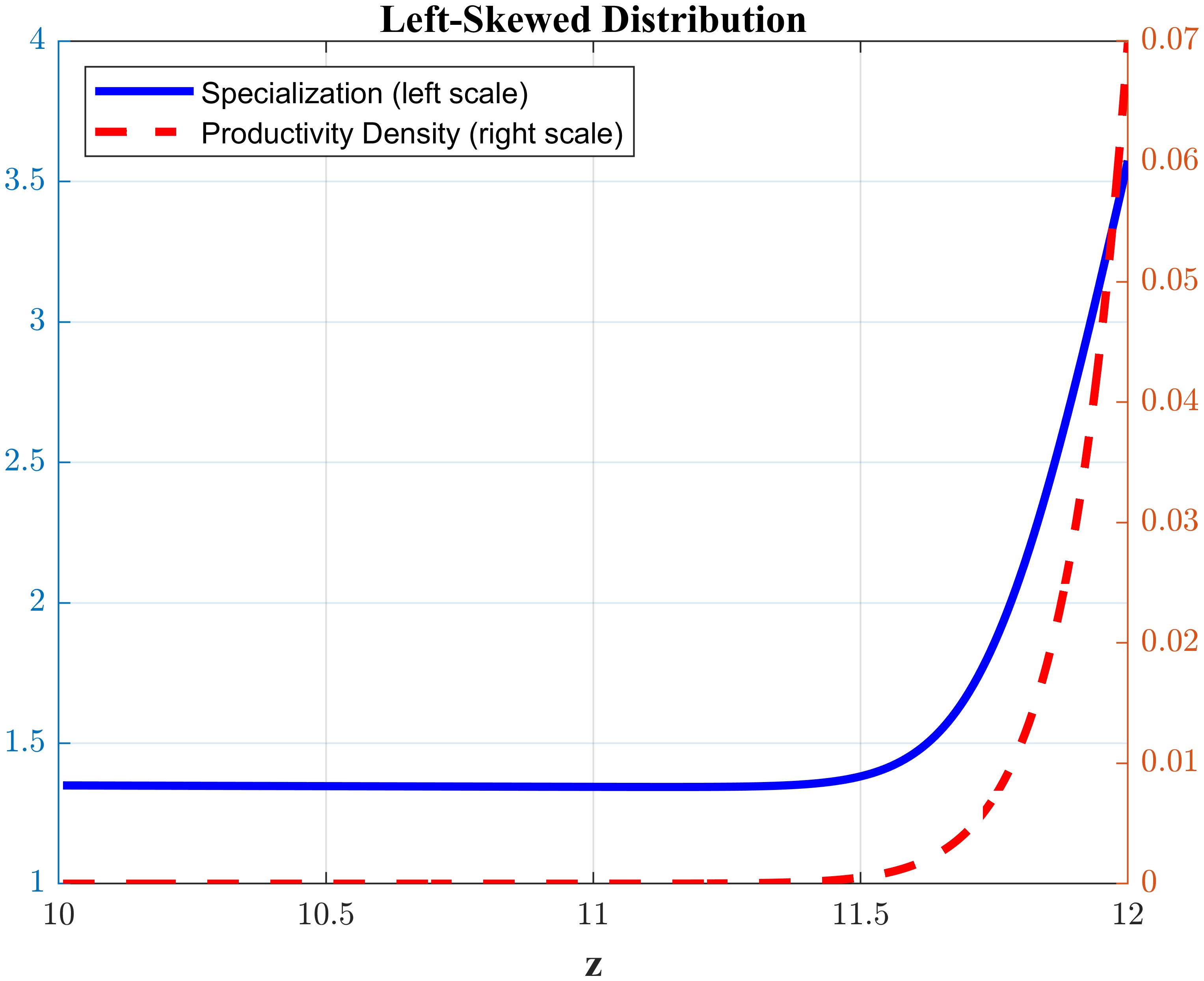}
    \end{minipage}
    \hfill
    \begin{minipage}{0.3\textwidth}
        \centering
         \includegraphics[width=\textwidth]{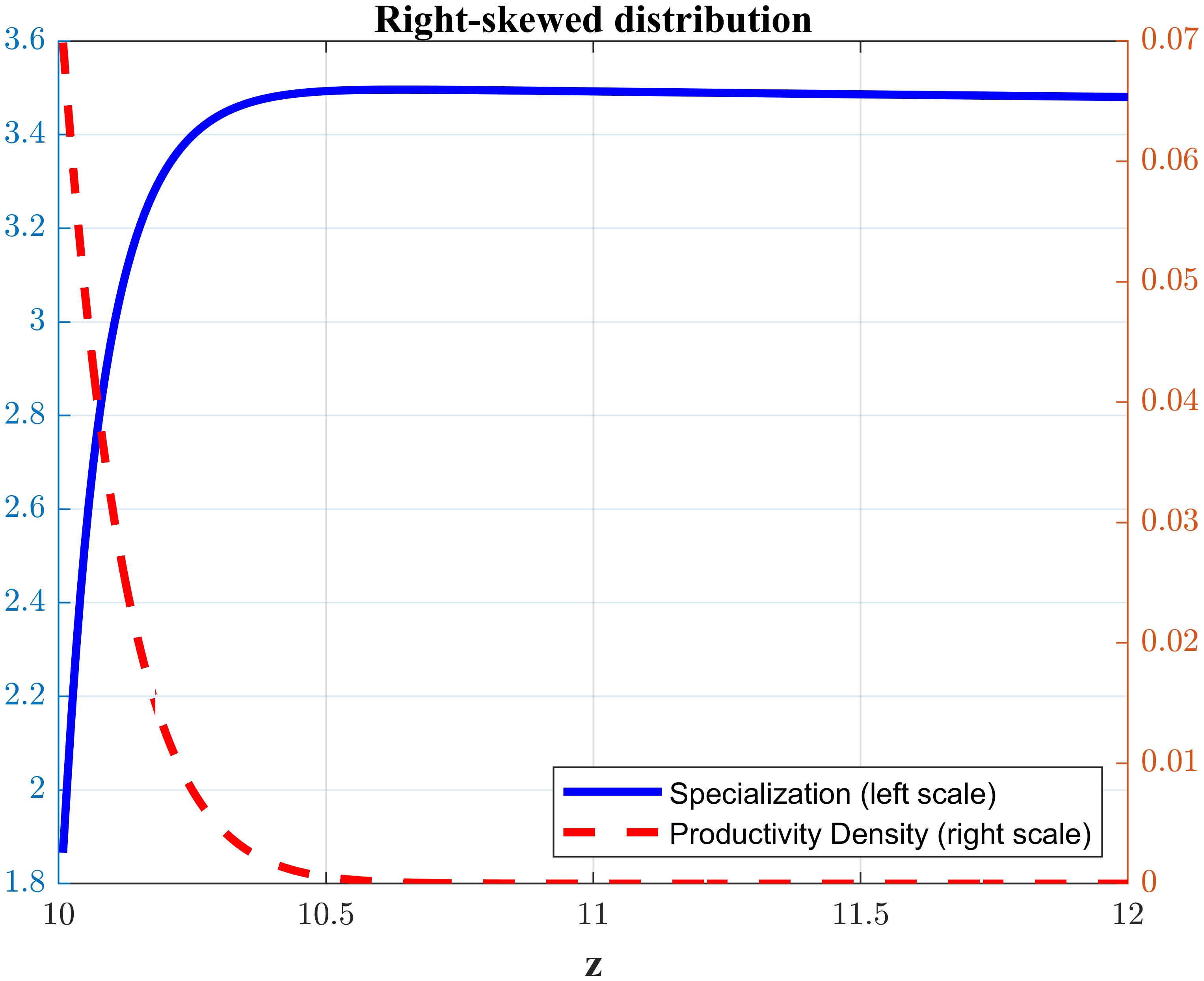}
    \end{minipage}
\begin{minipage}{\textwidth} \footnotesize{} \textit{Note}:  The three panels show the cross-sectional shape of the specialization function for different productivity distributions. The left panel considers a uniform distribution $z \sim U[10,20]$, the middle panel a left-skewed shifted Beta distribution $z \sim \mathcal{B}(1,18;10,20)$, the right panel a right-skewed shifted Beta distribution $z \sim \mathcal{B}(18,1;10,20)$. All the other structural parameters are the same across panels.
\end{minipage} 
\end{center}
\end{figure}
%
Figure (\ref{fig:spec_fnct}) shows how the specialization function is influenced by the shape of the productivity distribution. With a uniform density (left panel), specialization is (approximately) linearly increasing in productivity. A left-skewed productivity distribution (middle panel) entails a nonmonotonic behavior of specialization, which is decreasing in productivity for low productivity levels and exponentially increasing thereafter. Symmetrically, a right-skewed productivity distribution (right panel) gives rise to a mirror-like, nonmonotonic behavior of specialization, which is strongly increasing in productivity for low productivity levels and then gradually decreasing. The qualitative differences across productivity distributions stress the importance of firm heterogeneity as a determinant of specialization choices.

\subsection{Derivation and Interpretation of Optimal Surplus Offered}
\label{app:x_interpretation}
The first-order condition of the intermediate producer problem (\ref{static:pmp}) with respect to the surplus offered to final producers reads:
\begin{align}
\label{static:foc_x}
    \lambda \bar{\phi} G^\prime(x)[A(s;z)-x]=1.
\end{align}
Let $\mathcal{H}(z)$ denote the equilibrium distribution of surplus offered to final producers, i.e., $\mathcal{H}(z) \equiv G(x^\star(z))$. Under the guess that ${x^\star}^\prime(z)>0$, the equilibrium distribution of surplus offered to final producers equals the distribution of productivity weighted by the compatibility probability function, i.e., $\mathcal{H}(z)=\hat{\phi}(\underline{z},z)/\bar{\phi}$. It follows that $G^\prime(x(z))=\frac{\mathcal{H}^\prime(z)}{x^\prime(z)}$ and $\mathcal{H}^\prime(z)= \frac{\phi(s(z))\gamma(z)}{\bar{\phi}}$. Substituting these expressions into equation (\ref{static:foc_x}) yields the differential equation reported in the main text (\ref{eq_x}). 
Intermediate producers choose the optimal surplus offered $x^\star(z)$ by solving that differential equation, taking the reservation surplus $x_0$ (boundary condition) as given:
\begin{align*}
   & x^\prime (z) = \lambda \phi(s(z)) \gamma(z) [A(s(z))-x(z)] \\[.2cm]
   & x^\prime (z)+\lambda \phi(s(z)) \gamma(z) x(z) = \lambda \phi(s(z)) \gamma(z) A(s(z)) \\[.2cm]
   & e^{\lambda \hat{\phi}(\underline{z},z)} [ x^\prime (z)+\lambda \phi(s(z)) \gamma(z) x(z)] =  e^{\lambda \hat{\phi}(\underline{z},z)}   \lambda \phi(s(z)) \gamma(z) A(s(z)) \\[.2cm]
   & \frac{\partial}{\partial z} \left[e^{\lambda \hat{\phi}(\underline{z},z)} x(z) \right] = e^{\lambda \hat{\phi}(\underline{z},z)}   \lambda \phi(s(z)) \gamma(z) A(s(z)) \\[.2cm]
   & \int_{\underline{z}}^z \frac{\partial}{\partial \tilde{z}} \left[e^{\lambda \hat{\phi}(\underline{z},\tilde{z})}  x(\tilde{z})\right] d \tilde{z} = \int_{\underline{z}}^z e^{\lambda \hat{\phi}(\underline{z},\tilde{z})}    \lambda \phi(s(\tilde{z})) \gamma(\tilde{z}) A(s(\tilde{z})) d\tilde{z} \\[.2cm]
   & 
   e^{\lambda \hat{\phi}(\underline{z},z)} x(z)-x_0 = \int_{\underline{z}}^z e^{\lambda \hat{\phi}(\underline{z},\tilde{z})}   \lambda \phi(s(\tilde{z})) \gamma(\tilde{z}) A(s(\tilde{z})) d\tilde{z} \\[.2cm]
   & x(z) = e^{- \lambda \hat{\phi}(\underline{z},z)} \left[x_0 + \int_{\underline{z}}^z A(s(\tilde{z})) e^{\lambda \hat{\phi}(\underline{z},\tilde{z})}   \lambda \phi(s(\tilde{z})) \gamma(\tilde{z}) d\tilde{z} \right] \\[.2cm]
    & x(z) = e^{- \lambda \hat{\phi}(\underline{z},z)} x_0 + \int_{\underline{z}}^z  A(s(\tilde{z})) e^{- \lambda \hat{\phi}(\tilde{z},z)}   \lambda \phi(s(\tilde{z})) \gamma(\tilde{z}) d\tilde{z} \\[.2cm]
    & x^\star(z) = \left(1-f(z)\right) x_0 + f(z)\mathbb{E}_{\max\{\tilde{z}\}|\tilde{z}\leq z}[A\left(s^\star(\tilde{z});\tilde{z}\right)] \\[.2cm]
    & x^\star(z)  =  x_0 + f(z)\mathbb{E}_{\max\{\tilde{z}\}|\tilde{z}\leq z}[A\left(s^\star(\tilde{z});\tilde{z}\right)-x_0].
\end{align*}
The second-to-last line expresses the optimal surplus offered as a weighted average between the reservation surplus and the expected match surplus from the second-best compatible intermediate producer contacted by a searching final producer, where the weights are represented by the probability that a final producers is in contact with at least another intermediate producer with lower productivity than $z$ and its complement to one, respectively.
The last line makes clear that the optimal surplus offered exceeds the reservation surplus by the expected gap between the match surplus from the second-best compatible intermediate producer contacted and the reservation surplus. 

The optimal surplus offered function can be alternatively derived as the solution to a first-price sealed-bid auction with an unknown number of bidders. For given specialization function, let $v(z)$ be the valuation of a successful match for an intermediate producer of productivity $z$. An intermediate producer with productivity $z$ offering surplus $x$ makes expected operating profits:
\begin{align*}
  \Pi(x;z) &= \phi(z)\left(\sum_{k=0}^\infty m_k G(x)^k\right)f^{N-1}(v(z)-x) \\
  &= \phi(z) f^{N-1} \theta \lambda  \sum_{k=0}^\infty p_k G(x)^k(v(z)-x),
\end{align*}
where $p_k \equiv e^{-\lambda \bar{\phi}}\frac{(\lambda \bar{\phi})^k}{k!}$ denotes the probability that the number of other compatible intermediate producers contacted by a final producer is $k$, which follows a Poisson distribution. Setting the first-order condition with respect to $x$ to zero yields:
$$\sum_{k=0}^\infty p_k G(x)^k = (v(z)-x)\sum_{k=0}^\infty p_k k G(x)^{k-1} G^\prime(x).$$
We follow the same steps as in the main text by guessing that the optimal surplus offered is monotonically increasing in productivity, i.e., $x^\prime(z)>0$, and defining the equilibrium distribution of surplus offered as $\mathcal{H}(z)=G(x(z))$. Substituting into the optimality condition yields:
$$x^\prime(z) \sum_{k=0}^\infty p_k \mathcal{H}(z)^k = (v(z)-x(z))\sum_{k=0}^\infty p_k k \mathcal{H}(z)^{k-1} H^\prime(z).$$
Let $\tilde{\gamma}(z) \equiv \frac{\sum_{k=0}^\infty p_k k \mathcal{H}(z)^{k-1} H^\prime(z)}{\sum_{k=0}^\infty p_k \mathcal{H}(z)^k}=\frac{\sum_{k=0}^\infty p_k d\mathcal{H}(z)^k/dz}{\sum_{k=0}^\infty p_k \mathcal{H}(z)^k}$.
Solving the differential equation with boundary condition $x(\underline{z})=x_0$, we obtain the standard expression for the optimal bidding strategy in first-price sealed-bid auctions with an unknown number of bidders \citep{krishna2009auction}:
\begin{align} \label{xstar:auction}
  x^\star(z) = e^{-\int_{\underline{z}}^z\tilde{\gamma}(\tilde{z})d\tilde{z}} \left(x_0+\int_{\underline{z}}^z v(\tilde{z}) e^{\int_{\underline{z}}^{\tilde{z}} \tilde{\gamma}(\hat{z})d\hat{z}}  \tilde{\gamma}(\tilde{z}) d\tilde{z} \right).
\end{align}
Since $v(z) = A(s(z);z)$, it follows that the productivity density of the highest-surplus-offering compatible intermediate producer with lower productivity than $z$ equals $\frac{e^{-\int_{\tilde{z}}^z \tilde{\gamma}(\hat{z})d\hat{z}}}{1-e^{-\int_{\underline{z}}^z \tilde{\gamma}(\hat{z})d\hat{z}}} \tilde{\gamma}(\tilde{z})$.

Unlike standard auction theory, the distribution of the number of bidders is not taken as given in our model but pinned down by the extent of (exogenous) search frictions and (endogenous) compatibility frictions. Therefore, it is instructive to break down equation (\ref{xstar:auction}) further into a weighted average of states of the world where the final producer is able to locate lower-productivity competitors and where is not:
\begin{align*}
    x^\star(z) = \text{Pr}(k(z)=0)  x_0 + \text{Pr}(k(z)=1)   \mathbb{E}_{\max{\tilde{z}}|\tilde{z} \leq z, \  k(z)=1}\left[A(s^\star(\tilde{z});\tilde{z})\right],
\end{align*}
where $k(z) = \mathbbm{1}\{\text{final producer contacts at least one firm w/ productivity }\leq z\}$.
 From this formulation, it is apparent how the optimal surplus offered equals the expected outside option of a final producer.
 With probability $\text{Pr}(k(z)=0)=1-f(z)$, a final producer does not meet any compatible intermediate producer with productivity lower than $z$, in which case the final producer's outside option when meeting an intermediate producer with productivity $z$ boils down to her reservation surplus, that is, $x_0$. With complementary probability $f(z)$, the final producer meets at least one compatible intermediate producer with lower productivity than $z$, in which case the final producer's outside option when meeting an intermediate producer with productivity $z$ equals the maximum surplus the best of such firms can afford offering, that is, the expected highest match surplus among firms with lower productivity than $z$.

 \subsection{Bargaining Model}
\label{app:bargaining}
Our baseline model of specialization posits that intermediate producers post the price of their goods (or, interchangeably, set the offered surplus) ex ante, i.e., before meeting a final producer. In this section, we explore the alternative trading protocol based on ex post bargaining.

We study a model that is otherwise the same as our static baseline, but let the transaction price be determined by Nash bargaining upon meeting between a final producer and a compatible intermediate producer. We assume that final producers can observe the match surplus $v$ from matching with any intermediate producer they meet. Let $\xi \in [0,1]$ denotes the exogenous bargaining power weight of final producers in surplus-splitting negotiations.
The surplus accruing to the final producers solves:
\begin{align*}
   X^\star(v) = \argmax_{x} \left(X-\omega(v;\vec{v}_k)\right)^\xi \left(v-X\right)^{1-\xi}, 
\end{align*}
where $\omega(v;\vec{v}_k)$ denotes the outside option of a final producer which is contact with $k$ compatible intermediate producers with match surplus less than $v$.

To align the final producer's outside option in the two models, we assume that, while negotiating with the highest-surplus compatible intermediate producer, final producers keep the option of bargaining with the other less productive compatible intermediate producers they are in contact with. According to the Nash bargaining protocol, the surplus accruing to the final producer therefore equals:
$$X^\star(v)=\xi v +(1-\xi)\omega(v;\vec{v}_k),$$
where the final producer's outside option in bargaining equals the match surplus from the second best compatible intermediate producer, i.e.,   $\omega(v;\vec{v}_k)= \mathbbm{1}_{\{k=0\}}X_0+\mathbbm{1}_{\{k>0\}}\max\{v_1,\dots,v_k\}$, where $X_0$ is the reservation surplus.\footnote{Since intermediate producers have zero outside option at the bargaining stage, the second best compatible intermediate producer is indifferent between not trading with the final producer and offering him the entire match surplus.}
Let $\mathcal{G}(v)$ denote the distribution of match surplus. The expected outside option of a final producer when bargaining with an intermediate producer of productivity $z$, while being in contact with $k$ other compatible intermediate producers, reads: $\mathbb{E}[\omega(v;\vec{v}_k)] = \mathbbm{1}_{\{k=0\}}X_0+\mathbbm{1}_{\{k>0\}}\int_{\underline{v}}^v \tilde{v} \frac{k \mathcal{G}^\prime(\tilde{v})\mathcal{G}(\tilde{v})^{k-1}}{\mathcal{G}(v)^k}d\tilde{v}$.

According to this market structure, the expected operating profits of intermediate producers with productivity $z$ choosing match surplus $v$ equal:
\begin{align*}
    \Pi(v;z) =& \ \phi(s(v;z)) \ f^{N-1}\left(\sum_{k=0}^\infty m_k   \ \mathcal{G}\left(v\right)^k (1-\xi) \left(v-\mathbb{E}[\omega(v;\vec{v}_k)]\right)\right) \\
    =& \ \theta \lambda \phi(s(v;z)) f^{N-1} (1-\xi) \left(\sum_{k=0}^\infty e^{-\lambda\bar{\phi}}\frac{(\lambda\bar{\phi})^k}{k!}\mathcal{G}\left(v\right)^k \left(v-\mathbb{E}[\omega(v;\vec{v}_k)]\right)\right)
    \\
    =& \ \theta \lambda \phi(s(v;z)) f^{N-1} (1-\xi) \left(
    e^{-\lambda \bar{\phi}(1-\mathcal{G}\left(v\right))}v-e^{-\lambda \bar{\phi}} X_0-
\sum_{k=1}^\infty e^{-\lambda\bar{\phi}}\frac{(\lambda\bar{\phi})^k}{k!}\mathcal{G}\left(v\right)^k \left[v - \int_{\underline{v}}^v \frac{\mathcal{G}\left(\tilde{v}\right)^k}{\mathcal{G}\left(v\right)^k} d\tilde{v} \right]
   \right) \\
    =& \ \theta \lambda \phi(s(v;z)) f^{N-1} (1-\xi) \left(
    e^{-\lambda \bar{\phi}} (v-X_0)+ \int_{\underline{v}}^v e^{-\lambda \bar{\phi}(1-\mathcal{G}\left(\tilde{v}\right))}d\tilde{v} -  e^{-\lambda \bar{\phi}} (v-v_0)
   \right) \\
    =& \ \theta \lambda \phi(s(v;z)) f^{N-1} (1-\xi) \left[e^{-\lambda \bar{\phi}} (v_0-X_0)+\int_{\underline{v}}^v e^{-\lambda \bar{\phi}(1-\mathcal{G}\left(\tilde{v}\right))}d\tilde{v} \right].
\end{align*}
Optimal match surplus (or, interchangeably, optimal specialization) solves the following intermediate producer's problem:
\begin{align*}
  v^\star(z) = \argmax_v \ \Pi(v;z)-w q(s(v;z)).
\end{align*}
Substituting for $v(z) = A(s(z);z)$, equilibrium match surplus/specialization solves: 
\begin{align} \label{eq_s_barg2}
    &\theta \lambda \mathcal{P}(z;N) (1-\xi) \bigg[A^\prime(s^\star(z))+\frac{\phi^\prime(s^\star(z))}{\phi(s^\star(z))} \left(A(s^\star(z);z)-\mathbb{E}[\omega(z)]\right)\bigg]=w q^\prime(s^\star(z)),
\end{align}
where$\left(\frac{\partial s(v;z)}{\partial v}\right)^{-1}=A^\prime(s(z))$. Since in equilibrium $v^\prime(z)>0$, then $\int_{\underline{v}}^v e^{-\lambda \bar{\phi}(1-\mathcal{G}\left(\tilde{v}\right))}d\tilde{v} = \int_{\underline{z}}^z e^{-\lambda \bar{\phi}(1-\mathcal{G}\left(\tilde{z}\right))} v^\prime(\tilde{z}) d\tilde{z}$ $=  e^{-\lambda \hat{\phi}(z,\bar{z})} A(s^\star(z);z)-e^{-\lambda \bar{\phi}} A(s^\star(\underline{z});\underline{z})-\int_{\underline{z}}^z \lambda \phi(s(z)) e^{-\lambda \hat{\phi}(\tilde{z},\bar{z})} v(\tilde{z})\gamma(\tilde{z})d\tilde{z}$.
It follows that 
$e^{-\lambda \bar{\phi}} (v_0-X_0)+\int_{\underline{v}}^v e^{-\lambda \bar{\phi}(1-\mathcal{G}\left(\tilde{v}\right))}d\tilde{v} =  e^{-\lambda \hat{\phi}(z,\bar{z})} A(s^\star(z);z)-e^{-\lambda \bar{\phi}} X_0-\int_{\underline{z}}^z \lambda \phi(s(\tilde{z})) e^{-\lambda \hat{\phi}(\tilde{z},\bar{z})} v(\tilde{z})\gamma(\tilde{z})d\tilde{z} 
  = \ e^{-\lambda \hat{\phi}(z,\bar{z})}$ $\big(A(s^\star(z);z) -\mathbb{E}[\omega(z)]\big)$.
Notice that the expected outside option $\mathbb{E}[\omega(z)]$ equals the surplus offered to final producers in our baseline model, i.e., $\mathbb{E}[\omega(z)] = \left(1-f(z)\right) X_0 + f(z)\mathbb{E}_{\max\{\tilde{z}\}|\tilde{z}\leq z}[A\left(s^\star(\tilde{z});\tilde{z}\right)] \equiv x^\star(z)$. 

Comparing equations (\ref{eq_s_barg2}) and (\ref{eq_s}), we observe that the two coincide up to the intermediate producer's bargaining power weight multiplying the marginal benefit of specialization. Therefore, the bargaining model nests our baseline model if intermediate producers enjoy all the bargaining power, i.e., $\xi=0$.

Since the planner is indifferent as to how the match surplus is split, efficient specialization is still determined by equation (\ref{eff_s}). This allows us to establish the following proposition.
\begin{appprop}[Efficiency of the Static Economy with Bargaining]\label{prop:bargaining}
The equilibrium is constrained-efficient if and only if production is not complex, i.e. $N=1$, and intermediate producers hold all bargaining power, i.e., $\xi=0$. 
\end{appprop}
Since the network externality pushes towards over-specialization and bargaining frictions push towards under-specialization, there may be a bargaining power weight $\xi^\star$ such that the two externalities cancel out. However, notice that the bargaining power weight is independent of the trading counterpart, whereas the network externality is firm-specific. Hence, no bargaining power weight decentralizes the efficient allocation. 
If we allow bargaining weights to depend on the productivity of the intermediate producer involved in the transaction, there exists a unique schedule of bargaining weights that implements the efficient allocation:
\begin{align*}
    \xi^\star(z;N) = (N-1)\frac{\left(1-f(z)\right) \left(\mathbb{\hat{E}}[A(s^\star(\tilde{z});\tilde{z})]-\mathbb{\hat{E}}[x^\star(\tilde{z})]\right)}{A^\prime(s^\star(z))+\frac{\phi^\prime(s^\star(z))}{\phi(s^\star(z))}\bigg(A(s^\star(z);z) -  x^\star(z)\bigg)}.
\end{align*}
If intermediate producers do not hold all the bargaining power, the hold-up externality pushes equilibrium specialization to be inefficiently low. The bargaining weight schedule $\xi^\star(z;N)$ makes the hold-up externality exactly offsets the gap between network and across-market appropriability externalities. Importantly, since the hold-up externality does not vanish when production is not complex, the economy is always inefficient whenever $\xi(z)\neq \xi^\star(z;1) = 0$ for some $z$ when $N=1$.

\subsection{Directed Search}
\label{app:directed_search}
In this section we characterize the stationary equilibria of the models reported in the main text by assuming a directed search protocol. For consistency with the market structure of our baseline models, we assume that every input market is potentially segmented into submarkets characterized by pairs of expected surplus offered upon trading, $\mathbb{\hat{E}}[x]$, and input finding probability, $f(\bar{\phi})$, induced by the choice of a given reservation surplus, $x_0$. As in the baseline model, final producers still search simultaneously for potential suppliers and select as input provider the highest-surplus-offering one among the set of compatible supplier contacted. Unlike the baseline model, the reservation surplus is not pinned down by an expected break-even condition, but maximizes the joint surplus of the market participants in each input market. 

\paragraph{Static model.}
For each input market $j$, the equilibrium specialization and surplus offered functions can be determined through the following \textit{market utility approach}:\footnote{Formally, we guess that ${x^\star}^\prime(z)>0$ to replace $G(x(z))$ with $\hat{\phi}(\underline{z},z)/\bar{\phi}$ in the intermediate producers' profits.}
\begin{align*}
    \max_{s_j(z),x_j(z)} \ & \prod_{n=1}^N f\left(\bar{\phi}(s_n)\right) \sum_{n=1}^N \mathbb{\hat{E}}[x_n(z)] \\
    & \text{s.t.} \ \theta \lambda \phi(s_j(z)) e^{-\lambda \hat{\phi}(z,\bar{z})} \prod_{n \neq j} f\left(\bar{\phi}(s_n)\right) [A(s_j(z))-x_j(z)] - w q(s_j(z)) = U_j(z), \quad z \in [\underline{z},\bar{z}].
\end{align*}
Eliminating the surplus offered function $x_j(z)$ by substituting for the constraints into the objective function yields:
\begin{align*}
    \max_{s_j(z)} \ & \prod_{n=1}^N f\left(\bar{\phi}(s_n)\right) \left[ \mathbb{\hat{E}}[A_j(s(z);z)] + \sum_{n \neq j}^N \mathbb{\hat{E}}[x_n(z)] \right] - \frac{1}{\theta} \int \left[w q(s_j(z)) + U_j(z)\right] \gamma(z) dz.
\end{align*}
This problem yields a unique solution for the specialization function $s_j(z)$. Hence, each input market features just one active submarket (expected surplus offered-input finding probability pair). 
In a symmetric competitive search equilibrium, optimal specialization solves:
\begin{align*}
    & \theta \lambda \Pc(z;N)\bigg[A^\prime(s^\star(z))+\frac{\phi^\prime(s^\star(z))}{\phi(s^\star(z))} \bigg(A (s^\star(z);z) - f(z)\mathbb{E}_{\max{\tilde{z}}|\tilde{z}<z} [A(s^\star(\tilde{z});\tilde{z})] \\
    & + \left(1-f(z)\right) (N-1)  \mathbb{\hat{E}}[x^\star(\tilde{z})]\bigg)\bigg]=w q^\prime(s^\star(z)),
\end{align*}
which is the same as under random search. Hence, whether search is random or directed within input markets is irrelevant for our results.

Directed search would restore efficiency only if all input markets were integrated. In that case, the problem setup would be:
\begin{align*}
     & \max_{\left\{s_n,x_n\right\}_{n=1}^N} \  \prod_{n=1}^N f\left(\bar{\phi}(s_n)\right) \sum_{n=1}^N \mathbb{\hat{E}}[x_n(z)] \\
    & \text{s.t.} \ \theta \lambda \phi(s_j(z)) e^{-\lambda \hat{\phi}(z,\bar{z})} \prod_{n \neq j} f\left(\bar{\phi}(s_n)\right) [A(s_j(z))-x_j(z)] - w q(s_j(z)) = U_j(z), \quad z \in [\underline{z},\bar{z}], \ \boldsymbol{\forall j}.
\end{align*}
Unlike the previous problem setup, now the objective functions of intermediate producers in \textit{all} input markets $j$ are accounted for. As a result, markets are complete and the resulting allocation is constrained-efficient. Hence, constrained efficiency requires full (vertical and horizontal) integration of supply chains.

\paragraph{Static model with perfect information.}
Suppose final producers can observe the specialization and surplus offered of all the intermediate producers and, consequently, direct their search towards specific \textit{intermediate producers}. Hence, intermediate producers of equal productivity sort into a sub-market, which we index by their productivity $z$. Once final producers target a sub-market, we assume that all their $n \sim \text{Poisson}(\lambda)$ searches happen within the same sub-market (\cite{galenianos2009multiple} relax this assumption in a labor market setting with homogeneous agents). Since final producers randomize among the potential compatible suppliers contacted, the trading probability of intermediate producers is independent of the number of searches. The equilibrium specialization, surplus offered, and (sub-)market tightness can be determined through the usual market utility approach:
\begin{align*}
    \max_{s_j(z),x_j(z)} \ &  m\int \bigg(\frac{b(x_j(z))}{m \gamma(z)} f(s_j(z)) \prod_{n \neq j} f\left(s_n(z)\right) [A(s_j(z);z)-x_j(z)]-wq(s_j(z))\bigg) \gamma(z) dz\\
     \text{s.t.} & \ f(s_j(z)) \prod_{n \neq j} f\left(s_n(z)\right) \left[x_j(z)+\sum_{n \neq j} x_n(z) \right] = U_j, \quad z \in [\underline{z},\bar{z}], \\
    & \int b(x_j(z)) dz=1,
\end{align*}
where $b(x_j(z))$ is the measure of final producers queuing at the $z$ sub-market. It follows that the $z$-sub-market tightness equals $\theta(x_j(z)) = \frac{b(x_j(z))}{m \gamma(z)}$. In equilibrium, sub-markets offering higher surplus will have higher tightness, i.e., $b^\prime(x_j(z))>0$. Due to perfect information, final producers do not face any uncertainty about the input finding probability and surplus offered.

Eliminating the surplus offered through the first constraint, the problem reduces to:
\begin{align*}
    \max_{s_j(z)} \ &  \int \left[\bigg(b(s_j(z)) f(s_j(z)) \prod_{n \neq j} f\left(s_n(z)\right) \left[A(s_j(z);z)+\sum_{n \neq j} x_n(z) \right]-m\gamma(z)wq(s_j(z))\bigg) \right] dz-U_j\\
     \text{s.t.} 
    & \int b(s_j(z)) dz=1,
\end{align*}
where $b(s_j(z)) = b\left(\frac{U_j}{f(s_j(z)) \prod_{n \neq j} f\left(s_n(z)\right)}-\sum_{n \neq j} x_n(z)\right)$.
For all the active sub-markets, the first-order condition for optimal specialization in a symmetric equilibrium reads:
\begin{align*}
    & \ \theta(z) f^N\left[A^\prime(s(z))+\left(\frac{f^\prime(s(z))}{f(s(z))}+\frac{b^\prime(s(z))}{b(s(z))}\right) \left[A(s_j(z);z)+(N-1) x(z) \right]+\frac{\mu_b}{f^N}\frac{b^\prime(s(z))}{b(s(z))} \right] = wq^\prime(s(z)),
\end{align*}
where $\mu_b$ is the Kuhn Tucker multiplier attached to the adding-up constraint of the mass of final producers.

Again, directed search would restore efficiency only if all input markets were integrated. In that case, the problem setup would be:
\begin{align*}
    \max_{s_j(z),x_j(z)} \ &  m \int \bigg(\frac{b(x_j(z))}{m \gamma(z)} f(s_j(z)) \prod_{n \neq j} f\left(s_n(z)\right) [A(s_j(z);z)-x_j(z)]-wq(s_j(z))\bigg) \gamma(z) dz\\
     \text{s.t.} & \ f(s_j(z)) \prod_{n \neq j} f\left(s_n(z)\right) \left[x_j(z)+\sum_{n \neq j} x_n(z) \right] = U_j, \quad z \in [\underline{z},\bar{z}], \ \boldsymbol{\forall j}, \\
    & \int b(x_j(z)) dz=1, \quad \boldsymbol{\forall j}.
\end{align*}
Hence, our conclusions do not hinge upon the restriction on the information set.

\paragraph{Dynamic model.}
For each input market $j$, we follow again the market utility approach to determine the equilibrium specialization and surplus offered functions:
\begin{align*}
    & \max_{s_j(z),x_j(z)} \ \mu \prod_{n=1}^N f\left(\bar{\phi}(s_n)\right) \sum_{n=1}^N \mathbb{\hat{E}}[x_n(z)]+m\sum_{n=1}^N \int \hat{V}(\bar{X}_{n(-1)}(z))   \gamma(z) dz \\
    \text{s.t.} & \ D_j(z)A(s_j(z);z) - \bar{X}_{j(-1)}(z)-\theta \lambda \Pc_j(s_j(z);z)x_j(z)-w q(s_j(z)) + \beta V_j(D_j(z),\bar{X}_{j}(z);z)  = V_j(z), \\[.2cm]
    & \ D_j(z) = (1-\delta)  D_{j(-1)}(z) + \theta \lambda \Pc_j(s_j(z);z), \\[.2cm]
    & \ \bar{X}_j(z) = (1-\delta) \left[ \bar{X}_{j(-1)}(z)+\theta \lambda \Pc_j(s_j(z);z) x_j(z) \right], \quad z \in [\underline{z},\bar{z}] \\[.2cm]
    & \ \mu_{(+1)} = \delta \left[1-\left(1-\prod_{n=1}^N f(\bar{\phi}(s_n))\right)\mu\right]+\left(1-\prod_{n=1}^N f(\bar{\phi}(s_n))\right)\mu,
\end{align*}
where $\hat{V}(\bar{X}_{n(-1)}(z)) = \bar{X}_{n(-1)}(z) + \beta \hat{V}(\bar{X}_{n}(z))$ is the value for final producers of existing supplier relationships with  intermediate producers of input $n$ with productivity $z$.

Eliminating the surplus offered function $x_j(z)$ by substituting for the constraints into the objective function yields:
\begin{align*}
    \max_{s_j(z)} \ & m \int \left[ \mathcal{D}_j(z) A(s_j;z) - w q(s_j(z)) + \beta \tilde{V}(\mathcal{D}_j(z);z) - V_j(z)\right]\gamma(z)dz + \\
    & \mu \prod_{n=1}^N f\left(\bar{\phi}(s_n)\right) \sum_{n \neq j}^N \mathbb{\hat{E}}[x_n(z)]+\sum_{n \neq j}^N \int \hat{V}\left(\bar{X}_{n(-1)}(z)\right) \gamma(z) dz,
\end{align*}
where $ \tilde{V}(\mathcal{D}_j(z);z) \equiv V(\mathcal{D}_j(z),\bar{X}_j(z);z)+\hat{V}(\bar{X}_{j}(z))$ is independent of the surplus offered to existing customers.

In a symmetric competitive search equilibrium, optimal specialization solves:
\begin{align*}
    & \theta \lambda \Pc(z;N)\bigg[A^\prime(s^\star(z))+\delta[1+\beta(1-\delta)]\frac{\phi^\prime(s^\star(z))}{\phi(s^\star(z))} \bigg(A (s^\star(z);z) - f(z)\mathbb{E}_{\max{\tilde{z}}|\tilde{z}<z} [A(s^\star(\tilde{z});\tilde{z})] \\
    & + \left(1-f(z)\right) \left[(N-1)  \mathbb{\hat{E}}[x^\star(\tilde{z})]-\frac{\beta(1-\delta)}{1+\beta(1-\delta)}f^N \left(\mathbb{\hat{E}}[A(s^\star(\tilde{z});\tilde{z})+(N-1)\mathbb{\hat{E}}[x^\star(\tilde{z})]\right)\right]\bigg)\bigg]=w q^\prime(s^\star(z)).
\end{align*}
The competitive search equilibrium is implemented through a reservation surplus $x_0 =  -(N-1)  \mathbb{\hat{E}}[x^\star(\tilde{z})]+\frac{\beta(1-\delta)}{1+\beta(1-\delta)}f^N\left(\mathbb{\hat{E}}[A(s^\star(\tilde{z});\tilde{z})+(N-1)\mathbb{\hat{E}}[x^\star(\tilde{z})]\right)$.

Comparing the optimal specialization condition with its efficient counterpart (\ref{eff_s_dyn}), we draw two takeaways. First, directed search internalizes the search externality within each market, but not fully across markets. Hence, the equilibrium allocation is inefficient and displays over-specialization. Second, the equilibrium allocation depends on the search protocol. Hence, the irrelevance between random and directed search does \textit{not} extend to a dynamic setting.\footnote{The same takeaways hold also in the dynamic model with link destruction.}

\subsection{Contract Contingency and Completeness}
\label{app:non-cont}


In this section, we study the consequences of contracting frictions in supply chains on the efficiency properties of or baseline static equilibrium. First, we examine an alternative, non-contingent contractual arrangement in which suppliers are paid independently of whether all inputs are successfully sourced. In this setting, final producers bear all sourcing risk. In equilibrium, final producers who fail to source all required inputs incur losses, which are ultimately borne by the representative household in the form of lower profit rebates.

First, we assume that intermediate producers are paid independently of whether the final producer is able to source all the inputs successfully.
For intermediate producers, the trading probability is now independent of the complexity of final production, i.e.,  $\tilde{\mathcal{P}}(s,x;N)=\tilde{\mathcal{P}}(s,x)=\phi(s) e^{-\lambda \bar{\phi}(1-G(x))}$, since they are paid regardless of whether other inputs are successfully sourced. The consequences on the intermediate producers' policies are twofold. First, the optimal surplus offered increases due to a larger reservation surplus -- being its general structure still given by (\ref{eq_x_solved}). The reservation surplus for input $j$ is defined by the indifference condition between accepting the surplus offered and failing to source the input (thereby failing to produce at all) when in contact with a compatible intermediate producer: $$\prod_{n \neq j} f_n \left(A_{0,j} + \sum_{n \neq j} \mathbb{\hat{E}}[A_n]\right)-(A_{0,j}-x_{0,j})- \sum_{n \neq j} f_n \mathbb{\hat{E}}[A_n-x_n] = 0.$$
Hence, in a symmetric stationary equilibrium, the reservation surplus equals $x_0 =-(N-1)f\mathbb{\hat{E}}[x^\star(z)]+(1-f^{N-1})A(s^\star(\underline{z}))+(N-1)f(1-f^{N-2}) \mathbb{\hat{E}}[A(s^\star(z))]$.

Second, the optimal specialization $s_{NC}^\star(z)$ differs from the baseline model with contingent contracts, and is implicitly defined by:
    \begin{align}
        \label{eq_s_noncont}
    \theta \lambda \tilde{\Pc}(z)\bigg[A^\prime(s_{NC}^\star(z))+\frac{\phi^\prime(s_{NC}^\star(z))}{\phi(s_{NC}^\star(z))} \left(A \big(s_{NC}^\star(z);z \big)-x_{NC}^\star(z)\right)\bigg]=w q^\prime(s_{NC}^\star(z)).
    \end{align}
    %
Since the planner is indifferent as to how the match surplus is split, efficient specialization is still determined by (\ref{eff_s}). We proceed by evaluating the marginal welfare effect of specialization at the equilibrium solution with non-contingent contracts. Formally,
\begin{align} \nonumber
    \pd{\Wc}{s(z)}\bigg|_{s(z)=s_{NC}^\star(z)}\propto \quad &\underbrace{-\left(A^\prime(s_{NC}^\star(z))+\frac{\phi^\prime(s_{NC}^\star(z))}{\phi(s_{NC}^\star(z))}A(s_{NC}^\star(z);z)\right)(1-f^{N-1})}_{\text{contracting externality}} \\ \nonumber
    & -\frac{\phi^\prime(s_{NC}^\star(z))}{\phi(s_{NC}^\star(z))}\bigg[\underbrace{f^{N-1} f(z)\mathbb{E}_{\max\{\tilde{z}\}|\tilde{z} \leq z}\left[A(s_{NC}^\star(\tilde{z});\tilde{z})\right]}_{\text{business-stealing externality}} \\ \label{dwds_noncont}
&\underbrace{- x_{NC}^\star(z)}_{\substack{\text{appropriability} \\ \text{externality}}} \underbrace{-(N-1)f^{N-1}\left(1-f(z)\right)\mathbb{\hat{E}}[A(s_{NC}^\star(z);z)]}_{\text{network externality}}\bigg].
\end{align}
Note that Theorem \ref{thm:over-specialization} extends to an economy with non-contingent contracts. Just like the network externality and appropriability externality across markets, the \textit{contracting externality} vanishes when $N=1$. Hence, the equilibrium is efficient if and only if production is not complex. However, if the production process is complex, non-contingent contracts exacerbate over-specialization. In addition to the gap between the network externality and the appropriability externality across markets, two further forces  pushes toward excessive specialization. First, contracting frictions imply that the private marginal benefit of specialization on the conditional match surplus exceeds the social marginal benefit (contracting externality). Second, the appropriability externality within each market exceeds the business-stealing externality (notice that the latter is multiplied by $f^{N-1}<1$).
As intermediate producers are remunerated independently of whether sourcing is successful or not, they have stronger incentives to over-specialize when contracts are not contingent. 

Second, we consider an alternative extension to our baseline model, whereby intermediate producers can offer two surpluses: $x(z)$ in case of successful trade, and $x^\circ(z) \geq 0$ in case of no trade. Specifically, the intermediate producer offers $x^\circ(z)$ to the final producer upon meeting if the latter cannot find any compatible input suppliers. We view this contractual arrangement as a partial compensation for no-trade scenarios that would make incompatibility more costly for intermediate producers. In this context, we argue that positive no-trade payments are not sustainable in equilibrium; that is, $x^\circ(z)=0,\,\forall z$. Intermediate producers would only offer a positive $x^\circ$ if doing so increased the likelihood of being selected by final producers. However, final producers still find it optimal to choose the highest surplus-offering compatible supplier -- irrespective of no-trade payments. If the supplier is compatible, only the trade-case surplus $x$ is payoff-relevant. If the seller is not compatible, the final producer has no choice to make. Consequently, offering $x^\circ>0$ does not increase the trading probability, and, therefore, strictly decreases the intermediate producer's payoff. As a consequence, positive no-trade payments are not sustainable in equilibrium and do not solve the network externality.

\subsection{In-House Production Model}
\label{app:in-house}
Suppose final producers can produce any intermediate good of value added $\underline{A}$ in-house at cost $F > \underline{A}$. Since the surplus from producing intermediate goods in-house $\Delta \equiv F-\underline{A}$ is negative, they find it optimal to search the market for external providers of each input. In-house production possibilities set a lower bound to the reservation surplus, which now equals $x_0 = \max\{\Delta,-(N-1)\mathbb{\hat{E}}[x^\star(z)]\}$. If search for some input is unsuccessful, the final producer find it optimal to produce the missing input(s) in-house if and only if:
\begin{align*}
    MF \leq \sum_{n=1}^{N-M} x_n + M \underline{A} \iff \Delta \leq \frac{1}{M}\sum_{n=1}^{N-M} x_n,
\end{align*}
where $M \in [1,N]$ is the number of missing inputs (conditional on missing some input). Intuitively, a final producer produces the missing inputs in-house if the net cost of doing so (the left-hand side) is lower than the realized surplus from the other inputs (the right-hand side).  

Let $K(.)$ be the equilibrium distribution of $\frac{1}{M}\sum_{n=1}^{N-M} x_n$, with support $[(N-1)x_0,(N-1)x(\bar{z})]$. Under the intermediate producers' standpoint, the conditional trading probability becomes:
\begin{align*}
   \mathcal{P}(s,x;N) \equiv \ &  \phi(s) \ \left[1-G(x)\right] \ \left[f^{N-1}+(1-f^{N-1})(1-K(\Delta))\right] \\
   & = \phi(s) \ \left[1-G(x)\right] \ \left[1-K(\Delta)(1-f^{N-1})\right].
\end{align*}
Since $K(\Delta)<1$, the network externality operates even if final producers are allowed to produce inputs in-house. 

Moreover, the decision to turn on in-house production brings about a further inefficiency due to a hold-up problem. Indeed, a social planner would activate in-house production if and only if: 
\begin{align*}
    MF \leq \sum_{n=1}^{N-M} A_n + M \underline{A} \iff \Delta \leq \frac{1}{M}\sum_{n=1}^{N-M} A_n,
\end{align*}
Since $A_n \geq x_n \ \forall n$, in-house production is turned on too infrequently in equilibrium -- for given specialization function.

The equilibrium with in-house production is constrained-efficient if turning on in-house production is costless on net, i.e., $\Delta=0$. 

\subsection{Dynamic Model}
\label{app:dyn}
The value of an active final producer with suppliers of each input $j$ offering surplus $x_j$ in every period is:
\begin{align*}
    V(x_1,\dots,x_N) = \sum_n x_n + \beta \left[\delta S + (1-\delta) V(x_1,\dots,x_N) \right],
\end{align*}
where $S$ represents the value of a searching final producer. 
The expected value of an active final producer is:
\begin{align*}
    \mathbb{E}[V] = \frac{\sum_{n=1}^N \mathbb{\hat{E}}[x_n] + \beta \delta S}{1-\beta(1-\delta)}. 
\end{align*}
In turn, the value of a searching final producer is:
\begin{align*}
    S =& \prod_{n=1}^N f_n \mathbb{E}[ \max\left\{V, \beta S \right\}]+\left(1- \prod_{n=1}^N f_n \right) \beta S \\
    =&\prod_{n=1}^N f_n \frac{\sum_{n=1}^N \int_{x_{0,n}}^\infty x_n dF_n(x) + \beta \delta S}{1-\beta(1-\delta)} +\left(1- \prod_{n=1}^N f_n \right) \beta S,
\end{align*}
where the second equality follows from the optimality of a reservation surplus strategy (justified by the monotonicity of the value of an active final producer in the surplus offered on each input). A searching final producer is able to find all the inputs it needs with probability $\prod_{n=1}^N f_n$ -- each of which is expected to deliver surplus $x_n$ according to some distribution $F_n(x)$. If search is unsuccessful, the final producer will keep searching in the next period. The reservation surplus is defined by the indifference condition between accepting and continuing to search next period:
\begin{align*}
    V(x_{0,j},\mathbb{\hat{E}}[\boldsymbol{x}_{-j}]) = \beta S 
    \implies x_{0,j} = - \sum_{n \neq j}^N \mathbb{\hat{E}}[x_n] +\beta(1-\delta)(1-\beta)S.
\end{align*}
Since input markets are symmetric, the expected value of an active final producer in equilibrium is:
\begin{align*}
    \mathbb{E}[V] = \frac{N \mathbb{\hat{E}}[x] + \beta \delta S}{1-\beta(1-\delta)}. 
\end{align*}
In turn, the value of a searching final producer is:
\begin{align*}
    S = f^N \mathbb{E}[V]+(1-f^N) \beta S = \frac{f^N}{1-\beta (1-f^N)} \mathbb{E}[V].
\end{align*}
Plugging $S$ back into $\mathbb{E}[V]$ yields the equilibrium expression for the expected value of an active final producer:
\begin{align*}
    \mathbb{E}[V] = \frac{1-\beta(1-f^N)}{[1-\beta(1-\delta)][1-\beta(1-f^N)]-\beta\delta f^N} N \mathbb{\hat{E}}[x].
\end{align*}
Finally, the reservation surplus is given by:
\begin{align*}
\nonumber
    & x_0 = -\left[\left(1-\frac{\beta(1-\delta)(1-\beta) f^N}{[1-\beta(1-\delta)][1-\beta(1-f^N)]-\beta\delta f^N}\right)N-1\right]\mathbb{\hat{E}}[x^\star(z)] = -(\alpha(N)N-1) \mathbb{\hat{E}}[x^\star(z)],
\end{align*}
where $\alpha(N) \equiv 1-\frac{\beta(1-\delta) f^N}{1-\beta(1-\delta)(1-f^N)} \in (0,1)$ summarizes the influence of the option value of search on the reservation surplus.

We now solve the recursive profit maximization problem of intermediate producers (\ref{pmp_dyn}).
The first-order condition with respect to the surplus offered to the flow of new customers $x$ is:
\begin{align*}
  & \frac{\partial}{\partial x} \left[\theta \lambda \Pc(s,x;N)A(s;z)\right] + \beta \frac{\partial V(\mathcal{D},\bar{X})}{\partial \mathcal{D}}\frac{\partial}{\partial x} \left[\theta \lambda \Pc(s,x;N)\right]
  -\left[1-\beta(1-\delta) \frac{\partial V(\mathcal{D},\bar{X})}{\partial \bar{X}} \right]\frac{\partial}{\partial x} \left[\theta \lambda \Pc(s,x;N)x\right] \\
  & =0.
\end{align*}
The envelope conditions with respect to demand and the surplus offered to the stock of existing customers read:
\begin{align*}
\frac{\partial V(\mathcal{D}_{(-1)},\bar{X}_{(-1)})}{\partial \bar{X}_{(-1)}} &=-1, \quad 
  \frac{\partial V(\mathcal{D}_{(-1)},\bar{X}_{(-1)})}{\partial \mathcal{D}_{(-1)}} =(1-\delta)A(s;z).
\end{align*}
Rolling the envelope conditions forward and substituting them into the first-order condition yields:
\begin{align*}
   \left[1+\beta (1-\delta)\right] \frac{\partial}{\partial x} \left[\theta \lambda \Pc(s,x;N) \left(A(s;z)-x \right)\right]=0. 
\end{align*}
Developing the derivative allows recovering the same first-order condition as in the static model (\ref{static:foc_x}):
$\lambda \bar{\phi} G^\prime(x)[A(s;z)-x]=1.$
%
%
Hence, the equilibrium surplus offered to final producers in the dynamic model has the same structure as in the static model.

The stationary equilibrium of the dynamic model in Section \ref{dynamic_sec} is characterized by the following system of equations:
\begin{align*} \nonumber
    &   x_0 = -\left(\alpha(N)N-1\right)\mathbb{\hat{E}}[x^\star(z)] , \\[.2cm]  \nonumber
    &x^\star(z) = \left(1-f(z)\right) x_0+f(z)\mathbb{E}_{\max\{\tilde{z}\}|\tilde{z}\leq z}[A(s^\star(\tilde{z});\tilde{z})] , \\[.2cm]  \nonumber
    &  \mathcal{D}(z;N)\bigg[A^\prime (s^\star(z))+\delta[1+\beta(1-\delta)]\frac{\phi^\prime(s^\star(z))}{\phi(s^\star(z))}(A(s^\star(z);z)-x^\star(z))\bigg] = w q^\prime(s^\star(z)), \\[.2cm]
    & 1-\frac{\psi}{w}= N m \bar{q}, \\[.2cm]
    & C= \mu(f;N) f^{N}\frac{1}{\delta} N\mathbb{\hat{E}}[A(s^\star(\tilde{z});\tilde{z}), \quad \mu(f;N) = \frac{\delta}{1-(1-\delta)(1-f^N)}.
\end{align*}
The equilibrium conditions represent the optimal reservation surplus, surplus offered to final producers and specialization, along with the labor and goods market clearing conditions, respectively.

\subsubsection{Dynamic Model with Link Destruction}
\label{app:links}
The value of an active final producer is:
\begin{align*}
    V(x_1,\dots,x_N) = \sum_n x_n + \beta \left[(1-\rho(N)) S + \sum_{k=0}^N \binom{N}{k} (\delta f)^k (1-\delta)^{N-k}  V(\boldsymbol{x}_{N,k}) \right],
\end{align*}
where $\boldsymbol{x}_{N,k}$ denotes a vector with $k$ entries equal to $\mathbb{\hat{E}}[x]$ (the expected surplus offered by replaced inputs) and $N-k$ entries equal to the surplus offered by current inputs $x_i$ (under symmetry, the exact ordering of the $x_i$ does not matter).\footnote{By the Binomial theorem, it follows that $\sum_{k=0}^N \binom{N}{k} (\delta f)^k (1-\delta)^{N-k} = \rho(N)$.}
Notice that the retention probability function $\rho$ already allows for immediate replacement of missing suppliers. Hence, unlike in the baseline dynamic model, $S$ is interpreted as the value of an inactive final producer -- rather than a searching one.

In a symmetric equilibrium, the expected value of a final producer is:
\begin{align*}
    \mathbb{E}[V] = \frac{N \mathbb{\hat{E}}[x] + \beta (1-\rho(N)) S}{1-\beta\rho(N)}. 
\end{align*}
The value of an inactive final producer is:
\begin{align*}
    S = \beta \left[ f^N \mathbb{E}[V]+(1-f^N) S\right] = \frac{\beta f^N}{1-\beta (1-f^N)} \mathbb{E}[V]
\end{align*}
Plugging $S$ back into $\mathbb{E}[V]$, we get:
\begin{align*}
    \mathbb{E}[V] = \frac{1-\beta(1-f^N)}{(1-\beta)[1-\beta(\rho(N)-f^N)]} N \mathbb{\hat{E}}[x].
\end{align*}
For the computation of the reservation surplus, it is useful to characterize the expected value of an active producer whose input provider of input $j$ offers surplus $x_j$:
\begin{align*}
V(x_j,\mathbb{\hat{E}}[\boldsymbol{x}_{-j}]) = & \ x_0 + (N-1)\mathbb{\hat{E}}[\boldsymbol{x}] +\beta(1-\rho(N))S + \beta(1-\delta)\rho(N-1)V(x_j,\mathbb{\hat{E}}[\boldsymbol{x}_{-j}]) \\
    & +\beta[\rho(N)-(1-\delta)\rho(N-1)] \mathbb{E}[V],
\end{align*}
where $(1-\delta)\rho(N-1)$ denotes the probability that the link with the input provider of input $j$ is not severed and all the other inputs are sourced.
The reservation surplus is determined by the standard indifference condition between accepting the surplus offered and continuing to search in the next period:
\begin{align*}
    & V(x_0,\mathbb{\hat{E}}[\boldsymbol{x}]) = S , \\
    \implies &   x_0+(N-1)\mathbb{\hat{E}}[x] = \left[1-\beta \left(1-\rho(N)+(1-\delta)\rho(N-1)\right)\right] S -\beta \left[\rho(N)-(1-\delta)\rho(N-1)\right]\mathbb{E}[V], \\
    & x_0 = -\left[\left(1-\frac{\beta (1-\beta) \left(f^N-\rho(N)+(1-\delta)\rho(N-1)\right)}{1-\beta (\rho(N)-f^N)}\right)N-1\right]\mathbb{\hat{E}}[x^\star(z)] \\
    & \hphantom{x_0} = -\left(\tilde{\alpha}(N)N-1\right)\mathbb{\hat{E}}[x^\star(z)],
\end{align*}
where $\tilde{\alpha}(N) \equiv 1-\frac{\beta (1-\beta) \left[f^N-\delta f \rho(N-1)\right]}{1-\beta (\rho(N)-f^N)} \in (0,1)$ summarizes the influence of the option value of search on the reservation surplus, and follows from substituting $\rho(N)-(1-\delta)\rho(N-1)=\delta f \rho(N-1)$.

Let $\mu^A \in (0,1)$ denote the share of active final producers at the production stage. It follows that the share of searching final producers equals $\mu = (1-\mu^A_{(-1)})+\delta \mu^A_{(-1)}$. The dynamics of $\mu^A$ is governed by the following law of motion:
\begin{align}
    \mu^A = 1-\left[\left(1-\prod_{n=1}^N f_n \right)(1-\nu) +(1-f\rho(f;N-1))\nu\right]\mu,
\end{align}
where $\nu \equiv \frac{\delta\mu^A_{(-1)}}{\mu}$ is the share of searching final producers that were active in the previous period. Plugging the definition of $\nu$ into the stock identity for the share of searching final producers, we back out the relation between the share of searching final producers that were active in the previous period and the share of searching final producers:
\begin{align*}
    \mu_{(+1)} =& \ 1-\mu^A+\delta \mu^A = 1-\frac{\mu_{(+1)}\nu_{(+1)}}{\delta}+\mu_{(+1)}\nu_{(+1)}
    \implies \nu = \ \frac{\delta}{1-\delta}\left(\frac{1}{\mu}-1\right),
\end{align*}
where the implication follows from solving the identity for $\nu_{(+1)}$ and rolling all the time subscripts backward one period.

Exploiting symmetry across markets, the expected surplus offered of active matches, $\mathbb{\hat{E}}[x^\star(\tilde{z})]$, can be derived by applying the expectation operator $\mathbb{\hat{E}}[.]$ to the equilibrium surplus offered function (\ref{eq_x_solved}):
\begin{align} \nonumber
    \mathbb{\hat{E}}[x^\star(z)] =& -\mathbb{\hat{E}}[1-f(z)]\left(\tilde{\alpha}(N)N-1\right) \mathbb{\hat{E}}[x^\star(z)]  +  \mathbb{\hat{E}} \bigg[f(z) \mathbb{E}_{\max{\tilde{z}}|\tilde{z}<z} [A(s^\star(\tilde{z});\tilde{z})]\bigg] \\[.2cm]
    \nonumber
    =& \frac{\mathbb{\hat{E}} \left[f(z)\mathbb{E}_{\max{\tilde{z}}|\tilde{z}<z} [A(s^\star(\tilde{z});\tilde{z})]\right]}{1+\left(\tilde{\alpha}(N)N-1\right) \mathbb{\hat{E}}[1-f(z)]} \\[.2cm]
    \label{exp_surplus_links}
    =& \frac{\mathbb{\hat{E}} \left[f(z)\mathbb{E}_{\max{\tilde{z}}|\tilde{z}<z} [A(s^\star(\tilde{z});\tilde{z})]\right]}{1-\left(\tilde{\alpha}(N)N-1\right)\frac{1-f}{f}\ln{1-f}},
\end{align}
where we made use of the fact that $\mathbb{\hat{E}}[1-f(z)] = \frac{1}{f}\int e^{- \lambda \hat{\phi}(\underline{z},\tilde{z})}  e^{- \lambda \hat{\phi}(\tilde{z},\bar{z})} \lambda \phi(s(\tilde{z}))\gamma(\tilde{z})d\tilde{z} = \frac{e^{-\lambda \bar{\phi}}}{f} \int  \lambda \phi(s(\tilde{z}))\gamma(\tilde{z})d\tilde{z} = \frac{1-f}{f}\lambda \bar{\phi} = -\frac{1-f}{f}\ln{1-f}$.

Moreover, the stationary share of searching final producers, $\mu$, and the stationary share of searching final producers that were active in the previous period, $\nu$, read:
\begin{align*}
    \mu(f;N) =& \ \frac{\delta}{1-(1-\delta)\frac{1-f\rho(N-1)}{1-f \left(\rho(N-1)-f^{N-1}\right)}}, \\
    \nu(f;N) =& \ \frac{f^N}{1-f \left(\rho(N-1)-f^{N-1}\right)}.
\end{align*}

\paragraph{Efficiency.} The social planner solves the following recursive problem:
\begin{align*}
\mathcal{W}\left(\boldsymbol{\mathcal{D}}_{(-1)},\boldsymbol{\mu}\right) = \max_{ \substack{
s_j(z),\; j = 1,\dots,N \\
\quad z \in [\underline{z},\bar{z}]
}  } \ & m \sum_{n=1}^N \int D_n(z)A(s_n(z);z)\gamma(z)dz + \psi \log(1-\ell) +\beta \mathcal{W}\left(\boldsymbol{\mathcal{D}},\boldsymbol{\mu}_{(+1)}\right) \\[.4cm]
\text{s.t.} \ & D_j(z) = (1-\delta)\rho(\boldsymbol{f}_{-j}) D_{j(-1)}(z) + \frac{\mu_j}{m} \lambda \phi(s_j(z)) e^{-\lambda \hat{\phi}(z,\bar{z})} \tilde{f}_{j}(\boldsymbol{f}_{-j},\nu_j), \ \forall j,z, \\[.2cm]
& \mu_{j(+1)} = \delta+(1-\delta)\left[\left(1-\prod_{n=1}^N f_n\right)(1-\nu_j)+\left(1-f_j \rho (\boldsymbol{f}_{-j})\right)\nu_j\right]\mu_j
, \ \forall j, \\[.2cm]
& \tilde{f}_{j}(\boldsymbol{f}_{-j},\nu_j) = \nu_j \rho(\boldsymbol{f}_{-j})+(1-\nu_j) \prod_{n\neq j}f_n
, \ \forall j, \\[.2cm]
& \nu_j = \frac{\delta}{1-\delta} \left(\frac{1}{\mu_j}-1\right), \ \forall j, \\[.2cm]
& \ell = m \sum_{n=1}^N \int q(s_n(z)) \gamma(z) dz,
\end{align*}
where the second constraint follows from rearranging (\ref{mu_links}).

The first-order condition for efficient specialization of intermediate producers of input $j$ with productivity $z$ reads (suppressing notation when unambiguous):
\begin{align*}
    & \mu_j \lambda e^{-\lambda \hat{\phi}(z,\bar{z})} \tilde{f}_j \phi^\prime(s_j(z)) \left( A(s_j(z);z) \gamma(z) + \frac{\beta}{m} \frac{\partial \mathcal{W}}{\partial \mathcal{D}_j(z)}  \right)+ m \mathcal{D}_j(z) A^\prime(s_j(z)) \gamma(z) - \frac{\psi}{1-\ell} m q^\prime(s_j(z)) \gamma(z)  \\
    - & \mu_j \lambda \tilde{f}_j \lambda \phi^\prime(s_j(z)) \gamma(z) \int_{\underline{z}}^z  e^{-\lambda \hat{\phi}(\tilde{z},z)} \left(A(s_j(\tilde{z});\tilde{z})\gamma(\tilde{z})  +\frac{\beta}{m} \frac{\partial \mathcal{W}}{\partial \mathcal{D}_j(\tilde{z})} \right) d\tilde{z} \\
    + & \sum_{n \neq j} \int \mu_n \lambda \phi(s_n(\tilde{z})) e^{-\lambda \hat{\phi}(\tilde{z},\bar{z})} \frac{\partial \tilde{f}_n}{\partial s_j(z)} \left(A(s_j(\tilde{z});\tilde{z}) \gamma(\tilde{z})+\frac{\beta}{m} \frac{\partial \mathcal{W}}{\partial \mathcal{D}_n(\tilde{z})} \right) d\tilde{z} \\
    + & m \sum_{n \neq j} \int (1-\delta) \frac{\partial \rho(\boldsymbol{f}_{-n})}{\partial s_j(z)}  \mathcal{D}_{n (-1)}(\tilde{z}) \left(A(s_j(\tilde{z});\tilde{z})\gamma(\tilde{z})  +\frac{\beta}{m} \frac{\partial \mathcal{W}}{\partial \mathcal{D}_n(\tilde{z})} \right) d\tilde{z} \\
    - & \beta \sum_{n \neq j} \frac{\partial \mathcal{W}}{\partial \mu_{n(+1)}} (1-\delta) \mu_n \left[(1-\nu_n) \Pi_{s \neq j} f_s \frac{\partial f_j}{\partial s_j(z)}+ \nu_n f_n \frac{\partial \rho(\boldsymbol{f}_{-n})}{\partial s_j(z)} \right] \\
    - & \beta \frac{\partial \mathcal{W}}{\partial \mu_{j(+1)}} (1-\delta) \mu_j \left[(1-\nu_n) \Pi_{s \neq j} f_s \frac{\partial f_j}{\partial s_j(z)}+ \nu_n \rho(\boldsymbol{f}_{-j}) \frac{\partial f_j}{\partial s_j(z)}\right] = 0.
\end{align*}
The first line collects the direct effects of specialization on the firm making the specialization choice. The second line represents the spillover effect on the expected surplus from other intermediate producers
supplying the same input, holding constant both the conditional surplus and the compatibility
probability. The third and fourth lines capture the spillover effect
on the expected surplus from intermediate producers of complementary inputs. The fifth and sixth lines capture the spillover effect on meeting probabilities.

The envelope conditions read:
\begin{align*}
     \frac{\partial \mathcal{W}}{\partial \mathcal{D}_{j (-1)}} &= m(1-\delta) \rho(\boldsymbol{f}_{-j}) A(s_j(z);z) \gamma(z), \\
     \frac{\partial \mathcal{W}}{\partial \mu_j} &= \int  \left(\tilde{f}_j+\mu_j \frac{\partial \tilde{f}_j}{\partial \nu_j} \frac{\partial \nu_j}{\partial \mu_j}\right) \lambda \phi(s_j(z)) e^{-\lambda \hat{\phi}(z,\bar{z})} A(s_j(z);z)\gamma(z) dz.
\end{align*}
Rolling the envelope conditions forward, substituting them into the first-order condition, and using steady-state conditions, efficient specialization of intermediate producers with productivity $z$ is given by: 
%
\begin{align}
    \nonumber
      & \mathcal{D}(z;N) \bigg[A^\prime(\Sc(z))+[1-(1-\delta)\rho(f;N-1)][1+\beta(1-\delta)\rho(f;N-1)]\frac{\phi^\prime(\Sc(z))}{\phi(\Sc(z))}\bigg(A(\Sc(z);z)    \\ \nonumber
      & 
      - f(z)\mathbb{E}_{\max\{\tilde{z}\}|\tilde{z} \leq z}[A(\Sc(\tilde{z});\tilde{z})] + \chi_1(f;N)(N-1)\left(1-f(z)\right)\mathbb{\hat{E}}[A(\Sc(\tilde{z});\tilde{z})]   \\ \label{eff_s_dyn_links}
      & -\chi_2(f;N)\frac{\beta(1-\delta)}{1+\beta(1-\delta)}f^N N \left(1-f(z)\right) \mathbb{\hat{E}}[A(\Sc(\tilde{z};\tilde{z}))]\bigg)\bigg] = \frac{\psi}{1-N m \bar{q}} \  q^\prime(\Sc(z)).
\end{align}
The first-order condition characterizing efficient specialization in the dynamic model with link destruction differs from the baseline dynamic model (\ref{eff_s_dyn}) in two respects. First, the disruption probability of existing customers is $[1-(1-\delta)\rho(f;N-1)]$ instead of $\delta$: this reflects the fact that existing customers keep sourcing the input if the supplier link does not break down (which happens with probability $1-\delta$) and the other supplier links either do not face disruptions or are replaced in case of disruption (which happens with probability $\rho(f;N-1)$). Similarly, the marginal value of an existing customer is compounded by $\beta(1-\delta)\rho(f;N-1)$, rather than simply by $\beta(1-\delta)$, because a customer may separate if hit by a link disruption and unable to replace it. Second, the network and search externality are weighted by the multipliers $\chi_1(f;N) \in (0,1)$ and $\chi_2(f;N) \lesseqgtr 0$, respectively. These multipliers are defined as follows:
\begin{align*}
   \chi_1(f;N) \equiv & \ 1-\nu\frac{\rho(f;N-1)-\delta f \rho(f;N-2)}{\tilde{f}(f,\nu)}+(1-\delta)\frac{\delta f \rho(f;N-2)}{1-(1-\delta)\rho(f;N-1)}  \\
   = & \ 
    1-(1-\delta)\frac{f(1-\delta)\rho(f;N-2)(1-\rho(f;N-1))]}{1-(1-\delta)\rho(f;N-1)} \in (0,1),  \\[.6cm]
   \chi_2(f;N) \equiv & \ \frac{1+\beta(1-\delta)}{1+\beta(1-\delta)\rho(f;N-1)}\left(1-\frac{\rho(f;N-1)-f^{N-1}}{\tilde{f}(f;N-1)}\frac{\delta}{(1-\delta)\mu(f;N)}\right)\bigg[(1-\nu(f;N)) \\
   & \ +\frac{\nu(f;N)}{N}\bigg(\frac{\rho(f;N-1)}{f^{N-1}}+(N-1)\delta\frac{\rho(f;N-2)}{f^{N-2}}\bigg)\bigg] \\
   = & \ \frac{1+\beta (1-\delta)}{1+\beta (1-\delta) \rho(f;N-1)}\frac{f^{N-1}-\delta\rho(f;N-1)}{(1-\delta)f^{N-1}} \left(1-\frac{N-1}{N}f(1-\delta)\rho(f;N-2) \right) \\
   < & \ f^{N-2}[1-(1-f)N],
\end{align*}
where the simplifications follow from substituting the steady-state expressions for $\nu, \mu$, and $\tilde{f}$.
The multiplier $\chi_1(f;N)$ discounts, on the one hand, the network externality of the baseline dynamic model because a share $\nu$ of final searching producers need to source less than $N$ inputs in the market (second term) and, on the other hand, inflates it since existing customers have a lower probability of being active the higher the specialization of other inputs (third term). 
The multiplier of the search externality, $\chi_2(f;N)$, features two distinct components. The first multiplier discounts the search externality of the baseline dynamic model because a higher share of searching final producers, $\mu$, tilts the composition of the pool of searching final producers to those that were inactive in the previous period, which have a lower probability of being active. The second multiplier adjusts the probability that a searching final producer is active in the baseline dynamic model, $f^N$, for the composition of the pool of searching final producers. The two multipliers pushes the search externality in opposite directions, so their net effect is qualitatively ambiguous.

We now turn to characterizing the efficiency properties of the dynamic equilibrium with link destruction.
Following Theorem \ref{prop:eff_dynam_links}, the equilibrium allocation displays over-specialization if $\chi_1(f;N) (N\widetilde{\mathcal{M}}(N) -1) \mathbb{\hat{E}}[A(s^\star(z);z)]+x_0>0$, where $\widetilde{\mathcal{M}}(N) \equiv 1-\frac{\beta(1-\delta)}{1+\beta(1-\delta)}f^N \frac{\chi_2(f;N)}{\chi_1(f;N)}$. Substituting for the equilibrium reservation surplus $x_0$, the over-specialization condition reads:
\begin{align*}
    \chi_1(f;N)(N\widetilde{\mathcal{M}}(N)-1)\mathbb{\hat{E}}[A(s^\star(z);z)]-\left(\tilde{\alpha}(N)N-1\right)\mathbb{\hat{E}}[x^\star(z)]>0.
\end{align*}
We start by analyzing the first term, which represents the wedge between the network externality and the search externality. Since $\chi_1(f;N)\mathbb{\hat{E}}[A(s^\star(z);z)]>0$, the sign of this term is pinned down by the difference $N\widetilde{\mathcal{M}}(N)-1$. The following result allows us to sign the term.
\begin{applemma}[Network Externality Dominates Search Externality]
\label{lemma:ntwrk>search}
    If production is complex, the network externality always dominate the search externality, i.e.,
\begin{align*}
    N\widetilde{\mathcal{M}}(N)-1>0, \quad \forall N>1 \iff \widetilde{\mathcal{M}}(N) \in \left(\frac{1}{2},1\right].
\end{align*}
\end{applemma}

The derivation in equation (\ref{exp_surplus_links}) allows us to bound the expected surplus offered of active matches: $\mathbb{\hat{E}}[x^\star(z)] \in \left(0,\frac{f^2\mathbb{\hat{E}} \left[A(s^\star(\tilde{z});\tilde{z})\right]}{{f-(1-f)\ln{1-f}\left(\tilde{\alpha}(N)N-1\right)}}\right).$ Based on these bounds, we proceed by deriving sufficient conditions for over-specialization.

First, suppose that $\tilde{\alpha}(N)N-1>0$.
Upon substituting for the upper bound of $\mathbb{\hat{E}}[x^\star(z)]$, a sufficient condition for the equilibrium allocation to exhibit over-specialization is: 
\begin{align*}
    F(f;\delta,N,\beta) & \ \equiv \chi_1(f;N)(N\widetilde{\mathcal{M}}(N)-1)-\frac{f^2}{f \frac{1-\beta(\rho(N)-f^N)}{1-\beta \left[(\rho(N)-f^N)+(1-\beta)\left(f^N-\delta f \rho(N-1)\right) \right]}-(1-f)\ln{1-f}} \\
    & \ >0.
\end{align*}
This condition depends only on structural parameters ($\delta$, $N$, and $\beta$) and on the equilibrium input finding probability, which is bounded between $0$ and $1$.
 The following lemma allows us to characterize the behavior of the over-specialization condition as complexity varies.
\begin{applemma}[Asymptotic Behavior of Over-Specialization Condition]
\label{lemma:ismonotonic}
    Let $(f,\delta,\beta) \in (0,1)^3$. There exists a finite level of complexity $1<\underline N < \infty $ such that if the production process is sufficiently complex, i.e. $N>\underline N$, then $F(f;\delta,N,\beta) >0 \quad \forall N> \underline N$. 
\end{applemma}
\noindent
Lemma \ref{lemma:ismonotonic} allows us to claim that, if the production complex is sufficiently complex, the equilibrium allocation displays over-specialization and under-resilience
 for any combination of structural parameters \textit{and} equilibrium input finding probability. Figure \ref{fig:over-spec_links} shows that the dynamic equilibrium with link destruction exhibits over-specialization across nearly the entire parameter space -- especially so as complexity grows larger. 

\begin{figure}[ht!]
\begin{center}
\caption{Direction of equilibrium inefficiency}
\label{fig:over-spec_links}
\begin{subfigure}[b]{0.325\textwidth}

\label{fig:over-spec_links_N=2}

\caption{$N=2$}
\includegraphics[width=\textwidth,keepaspectratio]{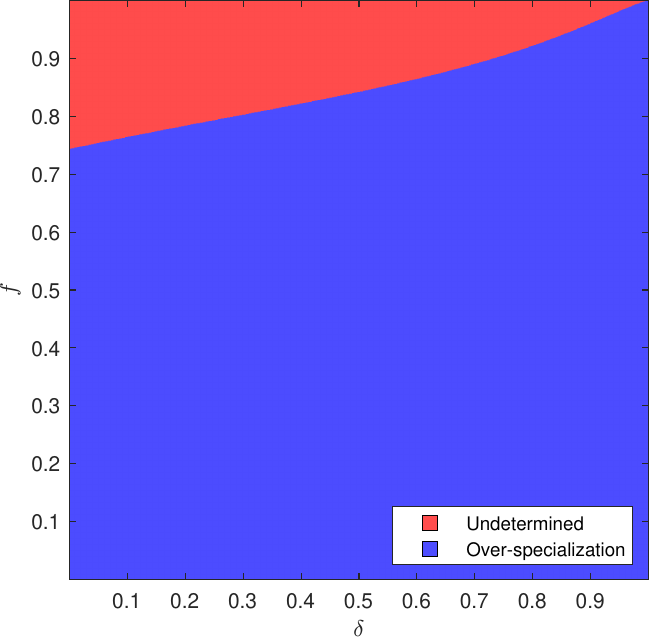}
\end{subfigure}
\begin{subfigure}[b]{0.325\textwidth}
\label{fig:over-spec_links_N=3}

\caption{$N=3$}
\includegraphics[width=\textwidth,keepaspectratio]{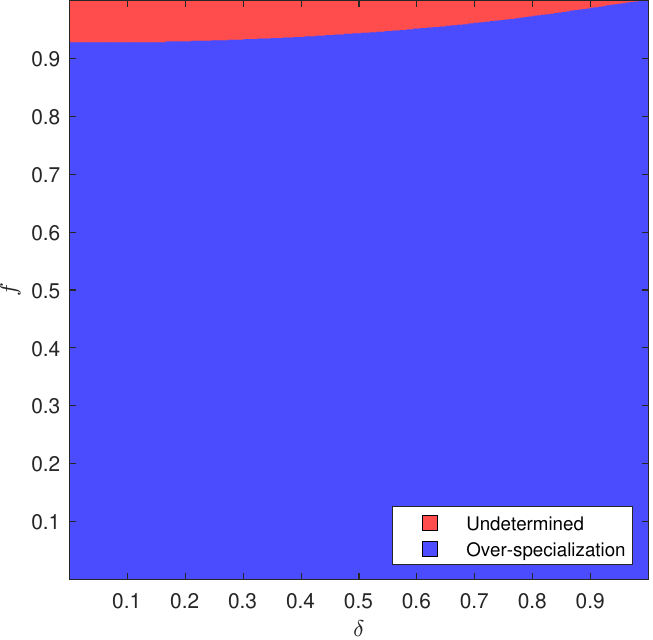}
\end{subfigure}
\begin{subfigure}[b]{0.325\textwidth}
\label{fig:over-spec_links_N=5}

\caption{$N=5$}
\includegraphics[width=\textwidth,keepaspectratio]{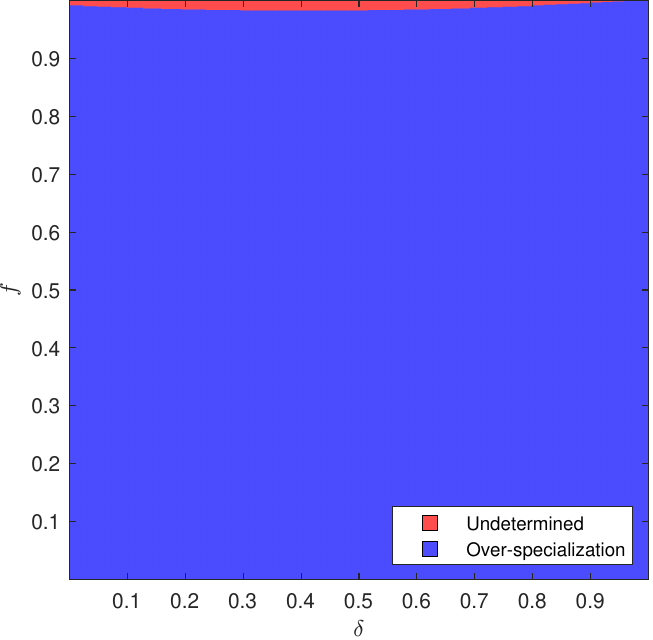}
\end{subfigure}
\begin{minipage}{0.9\textwidth} \scriptsize{} \textit{Note}:  The figure shows whether the function $F(f;\delta,N,\beta)$ is positive (blue) or negative (red) for any combination of $\delta, f \in \{0.001,0.002,0.003,\dots,0.999\}^2$, and $N \in \{2,3,5\}$. The discount factor $\beta$ is set to $0.995$.
\end{minipage} 
\end{center}
\end{figure}

Second, suppose that $\tilde{\alpha}(N)N-1<0$. 
Since $\mathbb{\hat{E}}[x^\star(z)]$ is bounded by zero from below, a sufficient condition for the equilibrium allocation to exhibit over-specialization is $ N\widetilde{\mathcal{M}}(N)-1>0, \quad \forall N>1$, which is always the case by Lemma \ref{lemma:ntwrk>search}.

Overall, our numerical checks allow us to conclude that, for any parametrization, the equilibrium of the dynamic model with link destruction displays over-specialization if the production process is sufficiently complex.

\subsubsection{Strategies to Mitigate the Impact of Supplier Breakdowns}
\label{app:inv_mult}
In this section we examine two key strategies that final producers can adopt to mitigate the impact
of supplier breakdowns: holding inventories and sourcing inputs from multiple intermediate
producers (multi-sourcing). Our goal is to determine whether these strategies can prevent firm-level disruptions caused by the inability of sourcing inputs. If they do, implementing the corresponding strategy would eliminate the network externality. Indeed, from the perspective of an intermediate producer, the probability that an attached final producer is active would be one regardless of the specialization of the other inputs.

\paragraph{Inventories.}
Suppose that active final producers can choose to purchase more units of each input than needed for current production, setting aside inventories as a buffer. If a supplier
relationship breaks down, these inventories allow production to continue until a new supplier is
found or until the stock is depleted. We further assume that inputs sourced can be used either in the current period or in the next -- otherwise they perish. This implies that final producers set aside either zero or one unit of input from an input provider. However, the analysis can be generalized to accommodate any finite number of periods after which inputs perish.

Active final producers choose optimal inventories by solving the following dynamic profits maximization problem:

\begin{align}
    \max_{\{I_j\}_{j=1}^N \in \{0,1\}^N} \underbrace{\mathcal{Q}_i-\sum_{j=1}^N(1+I_j)p_j}_{\pi_i\left(\{I_j^{(-1)}\}_{j=1}^N,\{I_j\}_{j=1}^N \right)} + \beta \mathbb{E}\left[\pi_i^\prime\left(\{I_j\}_{j=1}^N\right)\right],
\end{align}
where $\mathcal{Q}_i=\sum_{j=1}^N \max \{A_j^{(-1)} I_j^{\star(-1)},A_j\}$ is the consumption good quality and $I_j^{\star(-1)}$ denotes optimal inventories of input $j$ chosen in the previous period.\footnote{Since the inventory choice is made by active final producers, they always buy one unit of intermediate good by their input providers. Inventories are purchased on top of that.}

To highlight the basic economic forces shaping the inventory decision, we now specialize on the case where $N=2$. The optimal inventory policy function reads:
\begin{align} \label{eq:inventory}
    & I_j^{\star}\left(\{A_s\}_{s=1}^N,p_j,\{I^\star_s\}_{s \neq j}^N \right)=1 \iff \\ \nonumber
    p_j < & \ \beta \delta_j \Bigg( \left[(1-\delta_{-j})+\delta_{-j} \left(f_{-j}+(1-f_{-j})I^{\star}_{-j}\right)\right]f_j \text{Pr}(A^\prime_j \leq A_j) \mathbb{E}[A_j-A^\prime_j | A^\prime \leq A_j] \dots \\ \nonumber
    & \ + (1-\delta_{-j})(1-f_j)[A_{-j}+A_j-(1+I^{\star\prime}_{-j})p_{-j}] \dots \\ \nonumber
    & \ +\delta_{-j} f_{-j}(1-f_{j})[\max \{A_{-j} I_{-j}^{\star},A^\prime_{-j}\}+A_j-(1+I^{\star\prime}_{-j})p_{-j}] \dots \\ \nonumber
    & \ + \delta_{-j}(1-f_{-j})(1-f_{j})I^{\star}_{-j}[A_{-j} +A_j]\Bigg).
\end{align}
The marginal cost of an additional unit of inventories is the price of the input $j$ set by the current input provider. The marginal benefit comprises four terms -- all discounted by the probability that the supplier relationship with the current input provider breaks down, $\delta_j$. First, if the final producer is able to replace the supplier relationship in the next period, she would use the unit of inventories to produce were the new input supplier of lower value than the previous one (second row). Second, if the final producer is not able to replace the supplier relationship in the next period and the other supplier relationship does not break up, she would use the unit of inventories to continue production (third row). Third, if the final producer is not able to replace the supplier relationship in the next period and the other supplier relationship gets replaced upon breaking up, she would use the unit of inventories to continue production (third row). Finally, if also the other supplier relationship breaks down and the final producer is not able to replace either of them in the next period, she would use the unit of inventories to continue production provided that she has set aside inventories of the other input, as well  (fifth row).

Hence, inventories allows mitigating the impact of supplier breakdowns.
However, if inputs are perishable, inventories cannot completely hedge final producers from disruptions. The reason is that, with positive probability, the final producer will not be able to locate a compatible intermediate producer in $n<\infty$ periods.

If inputs were not perishable, a final producer may set aside a large number of units (approaching infinity) of inventories to completely hedge from disruptions. However, this strategy would not be profit-maximizing. As apparent from equation (\ref{eq:inventory}), the marginal benefit of holding inventories scales with the time-discounted probability of actually using such inventories. On the other hand, the marginal cost always equals the current price. It follows that there exists an $\tilde{n}<\infty$ such that, if $n>\tilde{n}$, the marginal benefit of the $n$-th unit of inventories approaches zero. Hence, holding infinite inventories cannot be a profit-maximizing strategy. 

Overall, inventories do not completely hedge final producers from disruptions in equilibrium. It follows that intermediate producers will always face a probability lower than one that an attached final producer is able to source all the other complementary inputs. Since this probability is decreasing in specialization, the network externality is not eliminated by endogenous inventories.

\paragraph{Multi-sourcing.}
Suppose final producers can establish links with multiple compatible
suppliers met during a supplier search. Specifically, assume that a final producer may choose to maintain an idle
backup link with the second-best compatible supplier at a fixed per-period cost $c_M$. If the link with
the best supplier breaks down first, the final producer can immediately switch to the backup
supplier and continue production. This implies that active final producers are in contact with either one or two intermediate producers of a certain input. However, the analysis can be generalized to accommodate any finite number of suppliers per input.

Let $D_j (D_j^M)$ be an indicator of link destruction with the main (backup) provider of input $j$.
Active final producers choose the optimal number of backup links by solving the following dynamic profits maximization problem:
\begin{align} \
    \max_{M_j \in \begin{cases}
         \scriptstyle{\{0,1\}} & \scriptstyle{\text{if } (D_j=1,N^c>1) \text{ or } (D_j^M=0,M_j^{\star (-1)}=1)} \\
         \scriptstyle{\{0\}} & \scriptstyle{\text{else}}
    \end{cases}
    , j=1,\dots,N}
    \pi_i-c_M\sum_{j=1}^N M_j + \beta \mathbb{E}\left[\pi_i^\prime\left(\{M_j\}_{j=1}^N\right)\right].
\end{align}
In words, final producers make a supplier search when the supplier relationship with their main input provider breaks down ($D_j=1$). Provided that they meet more than one compatible supplier ($N^c>1$), they can establish a backup link with the second-best compatible supplier. If they established a backup link in a past supplier search ($M_j^{\star (-1)}=1$), they can roll it over unless the link breaks for exogenous reasons ($D_j^M=0$).
To highlight the basic economic forces shaping the multi-sourcing decision, we now specialize on the case where $N=2$. The optimal multi-sourcing policy function reads:
\begin{align} \label{eq:multisource}
    & M_j^{\star}\left(\{x_s\}_{s=1}^N,\{x_s^M\}_{s=1}^N,c_M,\{M^\star_s\}_{s \neq j}^N \right)=1 \iff \\ \nonumber
    c_M < & \beta \delta_j(1-\delta_j^M) \Bigg( \left[(1-\delta_{-j})+\delta_{-j} \left(f_{-j}+(1-f_{-j})M^{\star}_{-j}\right)\right]f_j G(x^M_j) \mathbb{E}[x_j^M-x^\prime_j | x^\prime \leq x_j^M] \dots \\ \nonumber
    & + (1-\delta_{-j})(1-f_j)\left[\pi_i(x_j^M,x^M_{-j})-c_M\sum_j M_j^{\star \prime}\right] \dots \\ \nonumber
    & +\delta_{-j}f_{-j}(1-f_j)\left[\pi_i(x_j^M,\max\{x^M_{-j},x^\prime_{-j}\})-c_M\sum_j M_j^{\star \prime}\right]  \dots \\ \nonumber
    & +\delta_{-j}(1-f_{-j})(1-f_j) M^{\star}_{-j}\left[\pi_i(x_j^M,x^M_{-j})-c_M\sum_j M_j^{\star \prime}\right]\Bigg).
\end{align}
The optimality condition for the number of backup links shares many similarities with that of inventories (\ref{eq:inventory}). The marginal cost of an additional backup link is the per-period cost of maintaining an idle link alive, $c_M$. The marginal benefit comprises four terms -- all discounted by the probability that the supplier relationship with the main input provider breaks down and that with the backup supplier does not, $\delta_j(1-\delta_j^M)$. First, if the final producer is able to replace the supplier relationship in the next period, she would upgrade the backup link to main supplier were the new input supplier of lower value, thus gaining the differential surplus offered by the two (second row). Second, if the final producer is not able to replace the supplier relationship in the next period and the other supplier relationship does not break up, she would resort to the backup link to continue production (third row). Third, if the final producer is not able to replace the supplier relationship in the next period and the other supplier relationship gets replaced upon breaking up, she would resort to the backup link to continue production (third row). Finally, if also the other supplier relationship breaks down and the final producer is not able to replace either of them in the next period, she would resort to the backup link to continue production provided that she has a backup link for the other input, as well  (fifth row).

Hence, multi-sourcing allows mitigating the impact of supplier breakdowns.
However, if alternative input providers can be located only when performing a supplier search, multi-sourcing cannot completely hedge final producers from disruptions. The reason is that, with positive probability, the final producer will not be able to locate more than one compatible intermediate producer in a supplier search.

If final producers could search for alternative suppliers in any period (\textit{on-the-match} search), a final producer may build up a large number of backup links (approaching infinity) to completely hedge from disruptions. However, this strategy would not be profit-maximizing. As apparent from equation (\ref{eq:multisource}), the marginal benefit of an additional backup link scales with the time-discounted probability of actually using such links. On the other hand, the marginal cost always equals the link maintenance cost. It follows that there exists an $\tilde{n}<\infty$ such that, if $n>\tilde{n}$, the marginal benefit of the $n$-th backup link approaches zero. Hence, holding infinite backup links cannot be a profit-maximizing strategy. 

Overall, multi-sourcing does not completely hedge final producers from disruptions in equilibrium. It follows that intermediate producers will always face a probability lower than one that an attached final producer is able to source all the other complementary inputs. Since this probability is decreasing in specialization, the network externality is not eliminated by endogenous multi-sourcing.

\subsection{Normative Analysis}
\label{app:norm}
Suppose a social planner can use standards as the only policy instrument. We now study how the social planner would choose a standard in the static economy of Section \ref{static_sec}. 

The vector of optimal standards $\boldsymbol{\bar{s}} = (\bar{s}_1, \dots, \bar{s}_N)$  solves:
\begin{align*}
    \max_{\boldsymbol{\bar{s}}} & \sum_{n=1}^N \prod_{v \neq n} f(\boldsymbol{s}_v) \int A(s_n(\tilde{z});\tilde{z}) e^{-\lambda \hat{\phi}(\tilde{z},\bar{z})} \lambda \phi(s_n(\tilde{z})) \gamma(\tilde{z}) d\tilde{z} + \psi \log \left(1-m \sum_{n=1}^N q(s_n(\tilde{z}))\gamma(\tilde{z})d\tilde{z}\right) \\[.2cm]
    & \text{s.t.} \ s_j(z) = \min\{\bar{s}_j,s^\star_j(z)\} \ \forall j.
\end{align*}
Assuming that ${s^\star_j}^\prime(z)>0 \ \forall j, \forall z \in [\underline{z},\bar{z}]$, the problem underlying the choice of the standard for input $j$ can be expressed as follows:\footnote{As shown in Appendix \ref{app:spec_function}, the shape of the (exogenous) productivity distribution can be chosen to ensure that the equilibrium specialization function is monotonically increasing in productivity.}
\begin{align*}
  \max_{\bar{s}_j} & \prod_{v \neq j} f(\boldsymbol{s}_v) \left[\int_{\underline{z}}^{\hat{z}(\bar{s}_j)} A(s^\star_j(\tilde{z});\tilde{z}) e^{-\lambda \hat{\phi}_j(\tilde{z},\bar{z})} \lambda \phi(s^\star_j(\tilde{z})) \gamma(\tilde{z}) d\tilde{z} + \int_{\hat{z}(\bar{s}_j)}^{\bar{z}} A(\bar{s}_j;\tilde{z}) e^{-\lambda \hat{\phi}_j(\tilde{z},\bar{z})} \lambda \phi(\bar{s}_j) \gamma(\tilde{z}) d\tilde{z}\right] \\
  & +\sum_{n \neq j}^N \prod_{v \neq n} f(\boldsymbol{s}_v) \int A(s^\star_n(\tilde{z});\tilde{z}) e^{-\lambda \hat{\phi}_n(\tilde{z},\bar{z})} \lambda \phi(s^\star_n(\tilde{z})) \gamma(\tilde{z}) d\tilde{z} \\
  & + \psi \log \left(1-m \left[\int_{\underline{z}}^{\hat{z}(\bar{s}_j)} q(s^\star_j(z))\gamma(\tilde{z})d\tilde{z} +(1-\Gamma(\hat{z}(\bar{s}_j)))q(\bar{s}_j)\right] -m\sum_{n \neq j}^N q(s^\star_n(\tilde{z}))\gamma(\tilde{z})d\tilde{z}\right).
\end{align*}
From Proposition \ref{prop:stds} we know that the optimal standard is lower than the highest specialization observed in equilibrium, i.e., $\bar{s}_j \in [s^\star_j(\underline{z}),s^\star_j(\bar{z}))$. The optimal standard solves:
\begin{align*}
   & \sum_{n \neq j}^N \frac{\partial f(\boldsymbol{s}_n)}{\partial \bar{s}_j} \prod_{v \neq j,n} f(\boldsymbol{s}_v) \mathbb{E}_{\max \tilde{z}}[A(s_j(\tilde{z});\tilde{z})] +  \prod_{v \neq j} f(\boldsymbol{s}_v) \left[A(\bar{s}_j;\hat{z}) e^{-\lambda \phi(\bar{s}_j)(1-\Gamma(\hat{z}))} \lambda \phi(\bar{s}_j) \gamma(\hat{z}) \frac{\partial \hat{z}(\bar{s}_j)}{\partial \bar{s}_j} \right. \\
   & + \left. \int_{\underline{z}}^{\hat{z}(\bar{s}_j)} \frac{\partial}{\partial \bar{s}_j} \left(A(s^\star_j(z);z)e^{-\lambda \hat{\phi}_j(\tilde{z},\bar{z})} \lambda \phi(s^\star_j(\tilde{z}))\right)\gamma(\tilde{z})d\tilde{z}-  A(\bar{s}_j;\hat{z}) e^{-\lambda \phi(\bar{s}_j)(1-\Gamma(\hat{z}))} \lambda \phi(\bar{s}_j) \gamma(\hat{z}) \frac{\partial \hat{z}(\bar{s}_j)}{\partial \bar{s}_j} \right. \\
   & \left. + \int_{\hat{z}(\bar{s}_j)}^{\bar{z}}   \left(A^\prime(\bar{s}_j;\tilde{z}) \phi(\bar{s}_j)+A(\bar{s}_j;\tilde{z}) \phi^\prime(\bar{s}_j)\right) \lambda e^{-\lambda \phi(\bar{s}_j)(1-\Gamma(\tilde{z}))}\gamma(\tilde{z})d\tilde{z} \right. \\
   & \left. - \lambda \phi^\prime(\bar{s}_j)  \int_{\hat{z}(\bar{s}_j)}^{\bar{z}}   \vphantom{\frac{\partial \hat{z}}{\partial \bar{s}_j}} (1-\Gamma(\tilde{z})) A(\bar{s}_j;\tilde{z}) \phi(\bar{s}_j) \lambda   e^{-\lambda \phi(\bar{s}_j)(1-\Gamma(\tilde{z}))} \gamma(\tilde{z}) d\tilde{z} \right] \\
   & + \sum_{n \neq j}^N \sum_{m \neq n}^N \frac{\partial f(\boldsymbol{s}_m)}{\partial \bar{s}_j} \prod_{v \neq j,m} f(\boldsymbol{s}_v) \mathbb{E}_{\max \tilde{z}}[A(s_n^\star(\tilde{z});\tilde{z})]  + \sum_{n \neq j}^N \prod_{v \neq n} f(\boldsymbol{s}_n) \int_{\underline{z}}^{\bar z} \frac{\partial}{\partial \bar{s}_j} \left(A(s^\star_n(z);z)e^{-\lambda \hat{\phi}_n(\tilde{z},\bar{z})} \lambda \phi(s^\star_n(\tilde{z}))\right)  \\
   &  \gamma(\tilde{z})d\tilde{z} - \frac{\psi}{1-\ell}m \left[q(\bar{s}_j) \gamma(\hat{z}) \frac{\partial \hat{z}(\bar{s}_j)}{\partial \bar{s}_j}+ \int_{\underline{z}}^{\hat{z}(\bar{s}_j)} q^\prime(s^\star_j(\tilde{z})) \frac{\partial s^\star_j(\tilde{z})}{\partial \bar{s}_j} \gamma(\tilde{z}) d\tilde{z} -  q(\bar{s}_j) \gamma(\hat{z}) \frac{\partial \hat{z}(\bar{s}_j)}{\partial \bar{s}_j}+(1-\Gamma(\hat{z}))q^\prime(\bar{s}_j) \right.  \\
   & \left. + \sum_{n \neq j}^N \int q^\prime(s^\star_n(\tilde{z})) \frac{\partial s^\star_n(\tilde{z})}{\partial \bar{s}_j} \gamma(\tilde{z}) d\tilde{z} \right] := 0.
\end{align*}
The first term represents the marginal effect of setting a standard for input $j$ on the expected surplus from input $j$ induced by changing the input finding probability of the other inputs $n \neq j$. The second term in square brackets represents the marginal effect on the expected surplus from input $j$ for given finding probability of the other inputs. The third term represents the marginal effect of setting a standard for input $j$ on the expected surplus from complementary inputs $n \neq j$ induced by changing the input finding probability of input $j$. The fourth term represents the marginal effect on the expected surplus from complementary inputs $n \neq j$ for given finding probability of the other inputs. The last term in square brackets represents the marginal effect on the labor requirement for specialization.

This first-order condition can be significantly simplified by using an Envelope argument. Intuitively, small perturbations around the privately optimal specialization do not have first-order effects. Hence, we make use of equation (\ref{eq_s}) to observe that:
\begin{align} \nonumber
    & \ \prod_{v \neq j} f(\boldsymbol{s}_v)\int_{\underline{z}}^{\hat{z}(\bar{s}_j)} \frac{\partial}{\partial \bar{s}_j} \left(A(s^\star_j(z);z)e^{-\lambda \hat{\phi}_j(z,\bar{z})} \lambda \phi(s^\star_j(z))\right)\gamma(z)dz-\frac{\psi}{1-\ell}m  \int_{\underline{z}}^{\hat{z}(\bar{s}_j)} q^\prime(s^\star_j(z)) \frac{\partial s^\star_j(z)}{\partial \bar{s}_j} \gamma(z) dz \\ \nonumber
    = & \ \int_{\underline{z}}^{\hat{z}(\bar{s}_j)} \left[\prod_{v \neq j} f(\boldsymbol{s}_v) \frac{\partial }{\partial s^\star_j(z)} \int_{\underline{z}}^{\hat{z}(\bar{s}_j)}A(s^\star_j(\tilde{z});\tilde{z})e^{-\lambda \hat{\phi}_j(\tilde{z},\bar{z})} \lambda \phi(s^\star_j(\tilde{z}))\gamma(\tilde{z})d\tilde{z}-\frac{\psi}{1-\ell}m  q^\prime(s^\star_j(z)) \gamma(z)\right]\frac{\partial s^\star_j(z)}{\partial \bar{s}_j} dz \\ \nonumber
    & - \ \prod_{v \neq j} f(\boldsymbol{s}_v) \lambda \phi^\prime(\bar{s}_j) (1-\Gamma(\hat{z}))\int_{\underline{z}}^{\hat{z}(\bar{s}_j)} A(s^\star_j(\tilde{z});\tilde{z}) \lambda \phi(s^\star_j(\tilde{z})) e^{-\lambda \hat{\phi}(\tilde{z},\bar{z})} \gamma(\tilde{z}) d\tilde{z} \\
    \nonumber
    = & \ \int_{\underline{z}}^{\hat{z}(\bar{s}_j)} \Bigg(\frac{1}{m} \prod_{v \neq j} f(\boldsymbol{s}_v) \lambda \phi(s^\star_j(z)) e^{-\lambda \hat{\phi}_j(z,\bar{z})}\bigg[A^\prime(s^\star_j(z))+\frac{\phi^\prime(s^\star_j(z))}{\phi(s^\star_j(z))} \Big(A (s^\star_j(z);z) - f_j(z) \\ \nonumber
    & \ \mathbb{E}_{\max\{\tilde{z}\}|\tilde{z}\leq z}[A\left(s^\star_j(\tilde{z});\tilde{z}\right)]\Big)\bigg]  -\frac{\psi}{1-\ell}  q^\prime(s^\star_j(z)) \Bigg) \frac{\partial s^\star_j(z)}{\partial \bar{s}_j} m \gamma(z) dz\\ 
    \nonumber
    & - \ \prod_{v \neq j} f(\boldsymbol{s}_v) \lambda \phi^\prime(\bar{s}_j) (1-\Gamma(\hat{z}))\int_{\underline{z}}^{\hat{z}(\bar{s}_j)} A(s^\star_j(\tilde{z});\tilde{z}) \lambda \phi(s^\star_j(\tilde{z})) e^{-\lambda \hat{\phi}(\tilde{z},\bar{z})} \gamma(\tilde{z}) d\tilde{z} \\ \nonumber
    = & \ \int_{\underline{z}}^{\hat{z}(\bar{s}_j)} \prod_{v \neq j} f(\boldsymbol{s}_v) \lambda \phi^\prime(s^\star_j(z)) e^{-\lambda \hat{\phi}_j(z,\bar{z})}\left(1-f_j(z)\right) x_{0,j}  \frac{\partial s^\star_j(z)}{\partial \bar{s}_j} m \gamma(z) dz \\ \nonumber 
    & - \ \prod_{v \neq j} f(\boldsymbol{s}_v) \lambda \phi^\prime(\bar{s}_j) (1-\Gamma(\hat{z})) \int_{\underline{z}}^{\hat{z}(\bar{s}_j)} A(s^\star_j(\tilde{z});\tilde{z}) \lambda \phi(s^\star_j(\tilde{z})) e^{-\lambda \hat{\phi}(\tilde{z},\bar{z})} \gamma(\tilde{z}) d\tilde{z}.
\end{align}
where the last equality follows from substituting the first-order conditions for the privately optimal specialization (\ref{eq_s}) and surplus offered (\ref{eq_x_solved}). The same condition holds for the other inputs $n \neq j$ by letting $\hat{z}(\bar{s}_n)=\bar{z}$.

Upon implementing this simplification, as well as canceling out the derivatives of the extremes of integration, the optimal standard is pinned down implicitly by the following condition:
\begin{align} \nonumber
    & \prod_{v \neq j} f(\boldsymbol{s}_v) \ \lambda \phi(\bar{s}_j) \int_{\hat{z}(\bar{s}_j)}^{\bar{z}} \left[A^\prime(\bar{s}_j;\tilde{z})+\frac{\phi^\prime(\bar{s}_j)}{\phi(\bar{s}_j)}\left(A(\bar{s}_j;\tilde{z})-\tilde{x}(\bar{s}_j;\tilde{z})\right) \right]  e^{-\lambda \phi(\bar{s}_j) (1-\Gamma(\tilde{z}))} \gamma(\tilde{z}) d\tilde{z} = \\ 
    \nonumber
    & \ \frac{\psi}{1-\ell}m (1-\Gamma(\hat{z}(\bar{s}_j)))q^\prime(\bar{s}_j) -\sum_{n=1}^N \sum_{m \neq n}^N \frac{\partial f(\boldsymbol{s}_m)}{\partial \bar{s}_j} \prod_{v \neq n,m} f(\boldsymbol{s}_v) \mathbb{E}_{\max \tilde{z}} [A(s_n(\tilde{z});\tilde{z})] \\ \nonumber
    & - \prod_{v \neq j} f(\boldsymbol{s}_v)  \int_{\underline{z}}^{\hat{z}(\bar{s}_j)} \lambda \phi^\prime(s^\star_j(z)) e^{-\lambda \hat{\phi}_j(z,\bar{z})}\left(1-f_j(z)\right) x_{0,j}  \frac{\partial s^\star_j(z)}{\partial \bar{s}_j} m \gamma(z) dz \\ \nonumber
    & -\sum_{n \neq j}^N  \prod_{v \neq n} f(\boldsymbol{s}_v)  \int \lambda \phi^\prime(s^\star_n(z)) e^{-\lambda \hat{\phi}_n(z,\bar{z})}\left(1-f_n(z)\right) x_{0,n}  \frac{\partial s^\star_n(z)}{\partial \bar{s}_j} m \gamma(z) dz,
\end{align}
where $\tilde{x}(\bar{s}_j;z) \equiv \lambda \phi(\bar{s}_j)(1-\Gamma(z)) A(\bar{s}_j;z)+\int_{\underline{z}}^{\hat{z}(\bar{s}_j)} A(s_j(\tilde{z});\tilde{z}) \lambda \phi(s_j(\tilde{z})) e^{-\lambda \hat{\phi}(\tilde{z},\hat{z}(\bar{s}_j))} \gamma(\tilde{z}) d\tilde{z} =  \lambda \phi(\bar{s}_j) (1-\Gamma(z)) A(\bar{s}_j;z)+f(\hat{z}(\bar{s}_j))\mathbb{E}_{\max \tilde{z} | \tilde{z}<\hat{z}(\bar{s})}[A(s^\star(z);z)] \, \forall z \in [\hat{z}(\bar{s}_j),\bar{z}]$ acts as a \textit{shadow} surplus offered, i.e., the counterpart of the surplus offered in equation (\ref{eq_s}). Formally, it equals the sum of the foregone total surplus of intermediate producers with productivity $z>\hat z(\bar s)$ from capping specialization and the expected match surplus from the highest-productivity compatible intermediate producer contacted by a searching final producer with lower productivity than $\hat z(\bar s)$.

Upon imposing symmetry, the optimality condition can be rearranged as:
\begin{align} 
  & f^{N-1} \lambda \phi(\bar{s}) \int_{\hat{z}(\bar{s})}^{\bar{z}} \left[A^\prime(\bar{s};\tilde{z})+\frac{\phi^\prime(\bar{s})}{\phi(\bar{s})}\left(A(\bar{s};\tilde{z})-\tilde{x}(\bar{s};\tilde{z})\right) \right]  e^{-\lambda \phi(\bar{s}) (1-\Gamma(\tilde{z}))} \gamma(\tilde{z}) d\tilde{z} = \frac{\psi}{1-\ell}m (1-\Gamma(\hat{z}(\bar{s})))q^\prime(\bar{s}) \\ 
    \nonumber
    &  -(N-1) \left(\frac{\partial f_j}{\partial \bar{s}_j}+(N-1)\frac{\partial f_{-j}}{\partial \bar{s}_j}\right) f^{N-1} \mathbb{\hat{E}} [A(s(\tilde{z});\tilde{z})] - f^{N-1} m \bigg[\int_{\underline{z}}^{\hat{z}(\bar{s})} \lambda \phi^\prime(s(z))  e^{-\lambda \hat{\phi}(z,\bar{z})}\left(1-f(z)\right)\\
   \label{opt_std}
    &   \frac{\partial s^\star_j(z)}{\partial \bar{s}_j} \gamma(z)  dz  + (N-1)\int_{\underline{z}}^{\bar{z}} \lambda \phi^\prime(s(z)) e^{-\lambda \hat{\phi}(z,\bar{z})}\left(1-f(z)\right) \frac{\partial s^\star_{-j}(z)}{\partial \bar{s}_j}  \gamma(z) dz \bigg] x_0,
\end{align}
where the subscript $j$ is omitted whenever the respective variable is the same across input markets.
The optimality condition takes a familiar form. 
The left-hand side  provides the marginal benefit from higher specialization for constrained firms (recall that increasing $\bar s$ implies loosening the standard). All intermediate producers whose specialization is constrained increase their specialization to the new looser standard, thereby increasing the surplus ($A^\prime>0$) but reducing the compatibility probability ($\phi^\prime<0$) -- the latter being valued at the social surplus generated by the intermediate producer ($A-\tilde{x}$).

The right-hand side describes the marginal cost of loosening the standard. The marginal cost consists of three terms. First, a looser standard raises the marginal labor requirement for specialization. Second, changing the standard impacts the network externality on all input markets.\footnote{A looser standard on input $j$ increases the average specialization of the input itself, thus reducing the finding probability and worsening the network externality for complementary input producers, $\partial f_j/\partial \bar s_j<0$. Moreover, a looser standard affects the incentives to specialize in complementary inputs $-j$. On the one hand, both the higher marginal labor requirement for specialization and the lower trading probability in $j$ induce intermediate producers in other input lines to specialize less. On the other hand, a higher reservation surplus for their input (recall that the reservation surplus for a given input equals the expected surplus offered by complementary inputs) pushes them to specialize more. Hence, the impact on the network externality for complementary input producers is a priori ambiguous, $\partial f_{-j}/\partial \bar s_j \lesseqgtr 0$.}
Finally, changing the standard impacts the appropriability externality across markets induced by unconstrained intermediate producers on all input markets.  The sign of this term reflects the direction of the specialization response of unconstrained intermediate producers. According to the decomposition (\ref{dwds_expanded}), this term is weighted by the reservation surplus for the respective input.
Overall, the first line is the exact counterpart of the private optimal specialization condition (\ref{eq_s}). The last two lines capture the marginal change in the network externality and appropriability externality across markets.

We finally study how the effectiveness and bindingness of optimal standards vary with the structural parameters of our economy. To do so, we make use of a quantitative version of our dynamic model with steady-state behavior (see Section \ref{sec:dyn_ss}), and run comparative statics exercises around the parameter vector underlying Figure (\ref{fig:opt_stds}).\footnote{Focusing on the dynamic model with steady-state behavior allows us to gauge the welfare loss due to equilibrium over-specialization by directly comparing steady-state welfare in the equilibrium and efficient allocation. On the contrary, steady-state welfare is not the objective function of the dynamically-consistent Social Planer problems analyzed in Section \ref{dynamic_sec}.} Results are reported in Table \ref{tab:stds}. As search efficiency improves, firms choose higher optimal specialization. The planner also raises the standard $\bar s$, but does so more quickly than the market, so fewer firms end up constrained by the standard.
Conversely, when disruptions become more frequent, the planner lowers $\bar s$ more slowly than the market does, causing more firms to be constrained.
Likewise, as production complexity rises, firms reduce their privately optimal specialization and the planner also lowers the standard, but again at a slower pace, which makes the standard effectively tighter.
 \begin{table}[ht] 
 \tabcolsep=1.5pt\relax
  \centering
    \caption{Optimal standards - Comparative statics}
  \label{tab:stds}
\resizebox{0.4\textwidth}{!}{
  \begin{tabular}{c c c c}\hline\hline
  \Tstrut
 & \multicolumn{1}{c}{$ \quad \Wc(\Bar{s}^\star)/\Wc^{max} \quad $} & \multicolumn{1}{c}{$\quad \Bar{s}^\star/\mathbb{E}[s^\star(z)] \quad$} & \multicolumn{1}{c}{$\quad 1-\Gamma(z(\Bar{s}^\star)) \quad $} \\[.1cm]
 \hline \\
  \multicolumn{4}{c}{\textit{Search efficiency} } \\[.2cm]
 $\boldsymbol{\lambda=2}$ & $0.948$  & $0.662$  & $0.908$ \\[.2cm]
 $\boldsymbol{\lambda=3}$ & $ 0.991$  & $1.155$  & $0.382$ \\[.2cm]
 $\boldsymbol{\lambda=4}$ & $0.995$  & $1.595$  &  $0.004$ \\[.6cm]
  \multicolumn{4}{c}{\textit{Disruption probability}} \\[.2cm]
 $\boldsymbol{\delta=0.08}$ & $0.997$  & $1.271$ & $0.265$ \\[.2cm]
 $\boldsymbol{\delta=0.10}$ & $ 0.991$  & $1.155$  & $0.382$ \\[.2cm]
 $\boldsymbol{\delta=0.12}$ & $0.985$  & $1.097$ & $0.434$\\[.6cm]
  \multicolumn{4}{c}{\textit{Complexity}} \\[.2cm]
 $\boldsymbol{N=5}$ & $1.000$  & $1.430$  & $0.072$ \\[.2cm]
 $\boldsymbol{N=6}$ & $ 0.991$  & $1.155$  & $0.382$ \\[.2cm]
 $\boldsymbol{N=7}$ & $0.978$ & $1.041$ & $0.486$\\[.2cm]
 \hline \hline
 
 \end{tabular}
 }
 \vspace{.2cm}
 
 \begin{minipage}{0.7\textwidth} \scriptsize{} \textit{Note}:  The table reports (i) the ratio between the steady-state welfare by setting an optimal standard to all inputs in the stationary equilibrium of the dynamic model with steady-state behavior (see Section \ref{sec:dyn_ss}) and the steady-state welfare by implementing the steady-state welfare-maximizing specialization (\ref{eff_s_dyn_ss}), (ii) the ratio between optimal standard and average equilibrium specialization, and (iii) the share of constrained firms by the optimal standard for a combination of structural parameters. For each comparative statics exercise, all the other parameters are kept fixed at their value of Figure (\ref{fig:opt_stds}) (the intermediate value for the parameters considered).
\end{minipage} 
 \end{table}

\newpage
\section{Extensions}
\label{sec:extensions}
In this section, we extend our model along four dimensions.
First, we study the problem of final producers optimally choosing the level of complexity of their production process, that is, the length of their supply chain or the number of complementary inputs needed for final production. We show that, much like specialization choices, complexity is generally inefficient. 

Second, we allow final producers to invest resources to reduce the likelihood of own disruptions, e.g., through prudential investment in redundancies. We show that prudential investment is generally inefficient due to a well-known hold-up problem. 

Third, we consider a general quality function that allows for arbitrary substitutability or complementarity among complementary inputs.
This extension brings about a further externality induced by specialization choices associated with complementarity in quality.

Finally, we note that our baseline model features arbitrage opportunities. Specifically, we have not imposed any condition that makes the relative mass of intermediate and final producers adjust to arbitrage away any difference in the expected returns from running either class of firms. In our final extension, therefore, we allow entrepreneurs to enter the market either as intermediate or final producers. This extension brings about a further externality induced by specialization choices related to entry decisions.

For simplicity, in all extensions we assume steady-state behavior by both the private agents and the social planner -- as common in the equilibrium search literature \citep{burdett1998wage}. This means that intermediate producers maximize their steady-state profits, while the social planner maximizes the steady-state utility of the representative household. 
Hence, we start this section by developing the dynamic model with steady-state behavior to use as a benchmark.

\subsection{Dynamic Model with Steady-State Behavior}
\label{sec:dyn_ss}

\paragraph{Stationary equilibrium.}
Intermediate firms solve the following profits maximization problem:
\begin{align*}
    V(s,x;z) = \max_{s,x}
    & \ \mathcal{D}(s,x)\left(A(s;z)-x\right)-w q(s) \\
    \text{s.t.} & \ \mathcal{D}(s,x) =
    \frac{\theta \lambda}{\delta}\Pc(s,x).
\end{align*}
Differently from the baseline dynamic model, intermediate producers use stationary demand as demand constraint (instead of its law of motion).

The optimal surplus offered has the same structure as in the baseline static model, given by equation (\ref{eq_x_solved}). Due to steady-state behavior, the reservation surplus does not embed any option value of searching and boils down to its static value (\ref{eq:res_surplus}), as well.
Optimal specialization solves:
\begin{align} \label{eq_s_dyn_ss}
 \mathcal{D}(z;N)\bigg[A^\prime (s^\star(z))+\frac{\phi^\prime(s^\star(z))}{\phi(s^\star(z))}(A(s^\star(z);z)-x^\star(z))\bigg] =w q^\prime(s^\star(z)).
\end{align}
The only difference with respect to the optimal specialization condition of the baseline (dynamically-consistent) model (\ref{eq_s_dyn}) lies in the absence of any multiplier on the reduction in trading probability. 

\paragraph{Social planner problem.}
The social planner solves:
\begin{align*}
\mathcal{W} = \max_{ \substack{
s_j(z),\; j = 1,\dots,N \\
\quad z \in [\underline{z},\bar{z}]
}  } & \ m \sum_{n=1}^N \int D_j(z)A(s_j(z);z) \ \gamma(z)dz + \psi \log(1-\ell) \\[.4cm]
\text{s.t.} \ & D_j(z) = \frac{1}{\delta} \frac{\mu}{m} \lambda \phi(s_j(z)) e^{-\lambda \hat{\phi}(z,\bar{z})} \prod_{v\neq j} f_v, \\[.2cm]
& \mu = \frac{\delta}{\delta+(1-\delta)\prod_{n=1}^N f_{n}}, \\[.2cm]
& \ell = m \sum_{n=1}^N \int q(s_n(z)) \gamma(z) dz.
\end{align*}
Hence, the social planner uses stationary firm-level demand and the stationary share of searching final producers as constraints -- rather than their laws of motion.

Efficient specialization of intermediate producers with productivity $z$ is given by:
\begin{align}
    \nonumber
      & \mathcal{D}(z;N) \bigg[A^\prime(\Sc(z))+\frac{\phi^\prime(\Sc(z))}{\phi(\Sc(z))}\bigg(A(\Sc(z);z) - f(z)\mathbb{E}_{\max\{\tilde{z}\}|\tilde{z} \leq z}[A(\Sc(\tilde{z};\tilde{z}))]+(N-1)\left(1-f(z)\right)  \\ \label{eff_s_dyn_ss}
    & 
    \mathbb{\hat{E}}[A(\Sc(\tilde{z});\tilde{z})]-\frac{1-\delta}{\delta+(1-\delta)f^N}f^N N \left(1-f(z)\right) \mathbb{\hat{E}}[A(\Sc(\tilde{z});\tilde{z})]\bigg)\bigg] 
    = \frac{\psi}{1-N m \bar{q}} \  q^\prime(\Sc(z)). 
\end{align}
Comparing equations (\ref{eff_s_dyn_ss}) and (\ref{eff_s_dyn}), we observe that the former differs from the latter (i) by the absence of a multiplier on the reduction in trading probability, and (ii) in the multiplier on the search externality.

Rearranging (\ref{eff_s_dyn_ss}) allows us to establish the following proposition.
\begin{appprop}  [Efficiency of the Dynamic Economy with Steady-State Behavior] \label{prop:eff_dynam_ss}
   The economy is constrained-efficient if and only if $(N\mu(f;N)-1)\mathbb{\hat{E}}[A(s^\star(z);z)]+x_0=0$, where $\mu(f;N) = \frac{\delta}{\delta+(1-\delta)f^N}$. Let $\Xi(N) \equiv 1-(N-1)\frac{1-f}{f}\ln{1-f}$. If $N>\frac{\Xi(N)-1}{\Xi(N)\mu(f;N)-1}$ and $\Xi(N)\mu(f;N)-1>0$, then the economy features over-specialization.
\end{appprop}
Hence, assuming steady-state behavior, the variable $\mathcal{M}(N)$ defined in Theorem \ref{thm:eff_dynam} boils down the share of searching final producers $\mu(f;N)$. Overall, the dynamic model with steady-state behavior features the same qualitative sources of externality as the baseline dynamic model. 

\subsection{Endogenous Complexity}
\label{app:end_N}
In this extension, we endogenize the problem of firms choosing the complexity of their production process. In particular, final producers optimize over the number of key inputs in the production function. This approach is similar to the core idea of \cite{oberfield2018IO}, \cite{boehm2020misallocation}, \cite{kopytov2021endogenous} and \cite{kim2023supplychain}. 
For simplicity, in this extension, we treat $N$ as a measure rather than an integer.

\paragraph{Static model.}
We allow final producers to choose the complexity of the production process to maximize expected profits:
\begin{align*}
    N^\star = \argmax_{N\geq 1} f^{N} N \mathbb{\hat{E}}[x(z)],
\end{align*}
where $f$ is the input finding probability, and $\mathbb{\hat{E}}[x(z)]$ is the expected surplus accruing to the final producer on each input, conditional on being active. Intuitively increasing complexity induces two opposite effects. On the one hand, adding inputs increases the value of output directly. On the other hand, it reduces the probability that the consumption good is produced altogether.
The optimal level of complexity solves:\footnote{If $N$ were a positive integer, the optimal level of complexity would be $N^\star = \left\lfloor \frac{1}{1 - f} \right\rfloor$, where $\lfloor . \rfloor$ denotes the floor function.}
\begin{align}\label{eq_complex}
    N^\star = \frac{1}{\ln{1/f}}.
\end{align}
Hence, complexity is increasing in the input finding probability $f$.
Efficient complexity solves:
\begin{align*}
    \Nc = \argmax_{N \geq 1}   \ f^N N \mathbb{\hat{E}}[A(s(z);z)] + \psi \log \left(1-N m \bar{q} \right).
\end{align*}
Relative to the problem of individual firms, the social planner problem features two differences. First, additional complexity is evaluated at the expected social surplus, ${\mathbb{\hat{E}}[A(s(z);z)]}$, instead of the expected private surplus, ${\mathbb{\hat{E}}[x(z)]}$. Second, the planner accounts for the additional specialization costs associated with additional inputs, taking the form of higher disutility from labor. 
Hence, efficient complexity is given by:
\begin{align} \label{eff_complex}
    \Nc &= 
     \frac{1}{\ln{1/f}}\left(1- \frac{\tilde{w} \ell}{y}\right).
\end{align}
where $y$ is given by (\ref{output_static}), $\tilde{w} \equiv \frac{\psi}{1-Nm\bar{q}}$ is the shadow wage rate, and the labor share $\tilde{w}\ell/y$ represents the ratio between aggregate product design costs and value added.
Comparing equations (\ref{eq_complex}) and (\ref{eff_complex}) allows us to establish the following result.
\begin{appprop}[Efficiency of the Static Economy with Endogenous Complexity]\label{prop:efficiency_endog_N}
    The equilibrium is inefficient and features lower sourcing capacity than the constrained-efficient allocation, i.e., $f({\boldsymbol{s}}^\star)^{N^\star}<f(\boldsymbol{\Sc})^\Nc$. For given specialization, the equilibrium exhibits excess complexity, i.e., $N^\star>\Nc$. For given complexity,  the equilibrium exhibits over-specialization, i.e., $s^\star(z)>\Sc(z)$.
\end{appprop}
Intuitively, final producers do not internalize the marginal product design cost induced by their complexity choice, thereby generating a \textit{hold-up} externality. Choosing a higher level of complexity forces intermediate producers to invest more in product design. These costs are sunk at the time of price posting.  As a result, equilibrium sourcing capacity  is lower than efficient. 


\paragraph{Dynamic model.}
Equilibrium complexity is the same as in the static model and given by (\ref{eq_complex}).
Efficient complexity solves:
\begin{align*}
    \Nc = \argmax_{N \geq 1} \ &  \mu(f;N)\bigg(\frac{f^N}{\delta} N\mathbb{\hat{E}}[A(s(z);z)]\bigg) + \psi \log \left(1-Nm\bar{q}\right),
\end{align*}
where $\mu(f;N) = \frac{\delta}{\delta+(1-\delta)f^N}$ is the stationary share of searching final producers.
Efficient complexity is given by:
\begin{align} \label{eff_complex_dyn}
    \Nc = \frac{1}{\mu(f;\Nc)\ln{1/f}}\left(1- \frac{\tilde{w} \ell}{y}\right).
\end{align}
Comparing equations (\ref{eff_complex_dyn}) and (\ref{eff_complex}), we observe that a further inefficiency arises in a dynamic setting.
Following the same approach adopted for specialization, the marginal welfare effect of equilibrium complexity can be decomposed into two terms:  \begin{align}\label{dwdN}
    \pd{\Wc}{N}\bigg|_{N=N^\star}\propto \ \underbrace{- \ w\ell/y \ }_{\substack{\text{hold-up} \\ \text{externality}}} \ 
 \underbrace{\vphantom{-\mu}\ + (1-\mu(f;N^\star)) \ }_{\vphantom{\text{hold-up}}\text{search externality}}.
\end{align}
First, a hold-up externality pushes equilibrium complexity to be higher than efficient. The reason is that final producers fail to internalize the product design costs of each additional input, which are borne by intermediate producers before meeting. The surplus share of such product design costs is represented by the labor share. Second, final producers do not internalize the effect of their complexity choice on the equilibrium share of searching final producers. This search externality pushes the private marginal cost of complexity to exceed the social one. 
Rearranging (\ref{eff_complex_dyn}) allows us to establish the following result:
\begin{appprop}[Efficiency of the Dynamic Economy with Endogenous Complexity]\label{prop:dyn_endog_N}
The equilibrium is constrained-efficient if and only if:
\begin{align*}
    \begin{cases}
        (N \mu(f;N^\star)-1)\mathbb{\hat{E}}[A(s^\star(z);z)]+x_0=0,  \\
        LS^\star = 1-\mu(f;N^\star),
    \end{cases}
\end{align*}
where $\mu(f;N^\star) = \frac{\delta}{\delta+(1-\delta)f^{N^\star}}$ and $LS=\frac{w\ell}{y}$. 
The equilibrium features under-resilience if and only if ${LS}^\star>1-\mu(f;N^\star)$. 
\end{appprop}
\noindent
In general, the equilibrium features both inefficient complexity and specialization. 

\subsection{Endogenous Robustness}
\label{app:end_rob}
So far, we have treated the disruption probability of final producers as an exogenous parameter. In this section, we allow final producers to make prudential investments in robustness aimed at reducing their disruption probability, e.g., through redundancies. Investment in robustness comes at per-period cost $\kappa(r)$ per supply relationship, with $\kappa^\prime>0, \ \kappa^{\prime\prime}>0$. Specifically, we posit that $\delta=\delta(r)$, with $\delta^\prime<0, \ \delta^{\prime \prime}>0$. For analytical tractability, we assume that final producers
choose their robustness policy before matching with their input providers.\footnote{More specifically, we assume that final producers commit to a robustness policy $r$ at the search stage for the entire life of the firm until a disruption occurs. Implementing such a robustness policy requires paying a per-period cost $\kappa(r)$ per supplier relationship. If final producers could optimize their robustness policy once matched, intermediate producers would internalize the robustness response when choosing the surplus offered to final producers.} We think of robustness policies as the design and implementation of prudential strategies to reduce the exposure to exogenous shocks.

\paragraph{Stationary equilibrium.}
Final producers choose their robustness to maximize expected profits:
\begin{align}
    r^\star = \argmax_{r} f^N N\frac{ \mathbb{\hat{E}}[x(z)]-\kappa(r)}{\delta(r)},
\end{align}
where $\mathbb{\hat{E}}[x(z)]$ is the expected surplus accruing to the final producer on each input, conditional on being active.
Equilibrium robustness solves:
%
\begin{align} \label{eq_r}
    -\delta^\prime(r^\star)\frac{\mathbb{\hat{E}}[x(z)]-\kappa(r^\star)}{\delta(r^\star)} =\kappa^\prime(r^\star).
\end{align}
The privately optimal robustness trades off a lower disruption probability, $-\delta^\prime$, evaluated at the expected present discounted value of each supplier relationship, $\frac{ \mathbb{\hat{E}}[x]-\kappa}{\delta}$, against higher private marginal costs, $ \kappa^\prime$. Importantly, note that firms discount the future with their endogenous disruption probability, $\delta$. Under the assumption that final producers decide their robustness before matching, the expected surplus offered is constant across final producers, and so is their optimal robustness. Due to the presence of robustness costs, the reservation surplus of final producers, $x_0$, is pinned down by the following expected break-even condition:
\begin{align} \label{reservation_surplus_end_rob}
    x_0+(N-1)\mathbb{\hat{E}}[x(z)]=N\kappa(r^\star).
\end{align}
Final producers need to source $N$ key inputs. Whenever they meet with a potential supplier, they know that production will ultimately entail a $N \kappa(r)$ robustness cost. They are, therefore, willing to accept any surplus such that operating profits exceed that cost, where $(N-1)\mathbb{\hat{E}}[x(z)]$ represents the expected surplus they can obtain from the other $N-1$ suppliers.
%
The reservation surplus acts as a boundary condition for the differential equation governing the equilibrium surplus offered to final producers (\ref{eq_x}).
Hence, the equilibrium surplus offered to final producers reads:
\begin{align} \label{eq_x_solved_end_rob}
x^\star(z) &= \left(1-f(z)\right)\left(\kappa(r^\star)-(N-1)\mathbb{\hat{E}}[x^\star(z)-\kappa(r^\star)]\right)+f(z)\mathbb{E}_{\max\{\tilde{z}\}|\tilde{z}\leq z}[A(s^\star(\tilde{z});\tilde{z})] \\ \nonumber
&= \kappa(r^\star)- \left(1-f(z)\right)(N-1)\mathbb{\hat{E}}[x^\star(z)-\kappa(r^\star)]+f(z)\mathbb{E}_{\max\{\tilde{z}\}|\tilde{z}\leq z}[A(s^\star(\tilde{z});\tilde{z})-\kappa(r^\star)].
\end{align}
It follows that the expected surplus offered of active matches equals $\mathbb{\hat{E}}[x^\star(z)-\kappa(r^\star)]= \frac{1}{\Xi(N)}\mathbb{\hat{E}}[f(z)\mathbb{E}_{\max\{\tilde{z}\}|\tilde{z}<z}[A(s^\star(z);z)-\kappa(r^\star)]]$, where $\Xi(N) \equiv 1-(N-1)\frac{1-f}{f}\ln{1-f}>1$. 
Equations (\ref{reservation_surplus_end_rob}) and (\ref{eq_x_solved_end_rob}) jointly imply that the equilibrium reservation surplus offered to final producers is given by $x_0=  \kappa(r^\star)-\frac{N-1}{\Xi(N)}\mathbb{\hat{E}}[f(z)\mathbb{E}_{\max\{\tilde{z}\}|\tilde{z}<z}[A(s^\star(z);z)]]$. Note that $x_0<\kappa(r^\star)$, since final producers are willing to take on negative net surplus from a single supplier relationship to extract positive expected surplus from the other $N-1$ relationships.
Up to the different formulation of the equilibrium surplus offered to final producers, equilibrium specialization is still determined by  (\ref{eq_s_dyn_ss}):
\begin{align*} 
 \mathcal{D}(z;N,\delta(r^\star))\bigg[A^\prime (s^\star(z))+\frac{\phi^\prime(s^\star(z))}{\phi(s^\star(z))}(A(s^\star(z);z)-x^\star(z))\bigg] = w q^\prime(s^\star(z)).
\end{align*}
Interestingly, since $\delta(r)$ discounts the marginal benefit of specialization in (\ref{eq_s_dyn_ss}), equilibrium robustness and specialization are positively correlated.
This result aligns well with the empirical findings of \cite{khanna2022india}, which shows that supply chains involving more specialized inputs proved more robust to Covid-19-induced disruptions in India.

\paragraph{Social planner problem.}
 Next, we study the planner problem to evaluate the efficiency properties of equilibrium specialization and robustness decisions.
The efficient specialization condition has the same formulation as in the model with exogenous robustness, up to the definition of match surplus as net of robustness costs.

Efficient robustness solves the following social planner problem:
\begin{align*}
    \varrho = \argmax_{r} \ \mu(f,r;N)\frac{f^{N}}{\delta(r)} N \Big( \mathbb{\hat{E}}[A(s(z);z)]-\kappa(r)\Big).
\end{align*}
Notice that robustness costs are borne by active final producers, whose equilibrium mass is $\mu\frac{f^N}{\delta}=\frac{f^N}{\delta+(1-\delta) f^{N}}$.
Efficient robustness is implicitly defined by: 
\begin{align} \label{eff_r}
        -\delta^\prime(\varrho)\left(1-f^N\right)\mu(f,\varrho;N) \frac{ \mathbb{\hat{E}}[A(s(z);z)-\kappa(\varrho)]}{\delta(\varrho)} =\kappa^\prime(\varrho).
\end{align}
The socially optimal level of robustness trades off a lower disruption probability, $-\delta^\prime$, evaluated at the expected present discounted social value of each supplier relationship, $\left(1-f^N\right)\mu$ $\frac{\mathbb{\hat{E}}[A-\kappa]}{\delta}$, against the marginal robustness cost, $\kappa^\prime$.
Comparing equations (\ref{eq_r}) and (\ref{eff_r}), we observe that equilibrium robustness is generally inefficient. Let the aggregate surplus share accruing to final producers be  $\tilde{\zeta} \equiv \frac{\mathbb{\hat{E}}[x^\star(z)-\kappa(r^\star)]}{\mathbb{\hat{E}}[A(s^\star(z);z)-\kappa(r^\star)]}$. 
The marginal welfare effect of equilibrium robustness can be decomposed into two terms:
\begin{align}\label{dwdr}
    \pd{\Wc}{r}\bigg|_{r=r^\star}\propto \underbrace{\vphantom{-\mu}\ 1-\tilde{\zeta} \ }_{\substack{\text{hold-up} \\ \text{externality}}} \    \underbrace{- \ [1-(1-f^N)\mu(f,r^\star;N)] \ }_{\vphantom{\text{hold-up}}\text{search externality}} .
\end{align}
First, a standard hold-up externality induces equilibrium investment in robustness to be lower than efficient, as the flow private marginal benefit, $\mathbb{\hat{E}}[x-\kappa]$, falls short of the flow social marginal benefit, $\mathbb{\hat{E}}[A-\kappa]$. Second, firms do not internalize the effect on equilibrium tightness induced by their higher robustness. This search externality is evident when noting that the social planner discounts the marginal surplus at a higher rate, $\delta/[(1-f^N)\mu]$, than private firms do, $\delta$. It follows that the search externality pushes equilibrium investment in robustness to be higher than efficient.
The nature of the search externality is akin to that arising in models of endogenous match destruction, where firms ignore the effect of their reservation productivity choice on equilibrium tightness. As in those models, a specific surplus sharing rule allows decentralizing the efficient allocation \citep{pissarides2000}. 
Formally, efficient robustness obtains in equilibrium (for given aggregate specialization) if the aggregate surplus share accruing to final producers equals the equilibrium share of unattached final producers, i.e., $\tilde{\zeta}=(1-f^N)\mu$. Under this condition, the hold-up externality exactly offsets the search externality. 
Otherwise, the net effect of the hold-up and search externality is ambiguous.\footnote{Standard models of endogenous match destruction imply that either the equilibrium match destruction rate is efficient or lower than efficient. The reason why, in our model, match destruction can be excessive is the presence of the hold-up externality potentially countervailing the search externality.} 

We now turn to determining efficient specialization. The social planner problem reads:
\begin{align*}
\mathcal{W} = \max_{ \substack{
s_j(z),\; j = 1,\dots,N \\
\quad z \in [\underline{z},\bar{z}]
}  } & \ m \sum_{n=1}^N \int D_n(z)[A(s_n(z);z)-\kappa(r)] \ \gamma(z)dz + \psi \log(1-\ell) \\[.4cm]
\text{s.t.} \ & D_j(z) = \frac{1}{\delta(r)} \frac{\mu(\delta(r))}{m} \lambda \phi(s_j(z)) e^{-\lambda \hat{\phi}(z,\bar{z})} \prod_{v\neq j} f_v, \ \forall j,z, \\[.2cm]
& \mu(\delta(r)) = \frac{\delta(r)}{\delta(r)+(1-\delta(r))\prod_{n=1}^N f_{n}}, \\[.2cm]
& \ell = m \sum_{n=1}^N \int q(s_n(z)) \gamma(z) dz.
\end{align*}
%
Efficient specialization of intermediate producers with productivity $z$ is given by:
\begin{align}
   \nonumber
      & \mathcal{D}(z;N) \bigg[A^\prime(\Sc(z))+\frac{\phi^\prime(\Sc(z))}{\phi(\Sc(z))}\bigg(A(\Sc(z);z)-\kappa(\varrho) - f(z)\mathbb{E}_{\max\{\tilde{z}\}|\tilde{z} \leq z}[A(\Sc(\tilde{z};\tilde{z}))-\kappa(\varrho)] \\ \nonumber
    & 
    +(N-1)\left(1-f(z)\right)\mathbb{\hat{E}}[A(\Sc(\tilde{z});\tilde{z})-\kappa(\varrho)]-\frac{1-\delta}{\delta+(1-\delta)f^N}f^N N \left(1-f(z)\right) \mathbb{\hat{E}}[A(\Sc(\tilde{z});\tilde{z})-\kappa(\varrho)]\bigg)\bigg] \\  
    & 
    = \frac{\psi}{1-N m \bar{q}} \  q^\prime(\Sc(z)). 
\end{align}
%
%
Evaluating the marginal welfare effect of specialization at the equilibrium specialization yields:
\begin{align} \nonumber
    \pd{\Wc}{s(z)}\bigg|_{s(z)=s^\star(z)} \propto & \ \underbrace{ f(z)\mathbb{E}_{\max\{\tilde{z}\}|\tilde{z} \leq z}[A(s^\star(\tilde{z};\tilde{z}))-\kappa(r^\star)]}_{\text{business stealing externality}} \underbrace{- \vphantom{f(z)} \ [x^\star(z)-\kappa(r^\star)]}_{\substack{\text{appropriability} \\ \text{externality}}} \\ \nonumber
    & \ \underbrace{- \ (N-1)\left( 1-f(z) \right)\mathbb{\hat{E}}[A(s^\star(\tilde{z});\tilde{z})-\kappa(r^\star)]}_{\text{network externality}} \\ \label{dwds_end_rob}
    & \ \underbrace{\vphantom{\int_{\underline z}^z}+ \  \frac{1-\delta(r^\star)}{\delta(r^\star)+(1-\delta(r^\star))f^N}f^N \left( 1-f(z) \right) N\mathbb{\hat{E}}[A(s^\star(\tilde{z});\tilde{z})-\kappa(r^\star)]}_{\text{search externality}}.
\end{align}
The only difference with the baseline dynamic model lies in the definitions of match surplus and surplus offered \textit{net} of robustness costs.
We summarize our findings in the next proposition.
\begin{appprop} [Efficiency of the Dynamic Economy with Endogenous Robustness]\label{prop:robustness_efficiency}
The equilibrium is constrained-efficient if and only if:
\begin{align*}
    \begin{cases}
        (N\mu(f,r^\star;N)-1)\mathbb{\hat{E}}[A(s^\star(z);z)-\kappa(r^\star)]+[x_0-\kappa(r^\star)]=0,  \\
       \tilde{\zeta}=(1-f^N)\mu(f,r^\star;N),
    \end{cases}
\end{align*}
where $\mu(f,r^\star;N) = \frac{\delta(r^\star)}{\delta(r^\star)+(1-\delta(r^\star))f^N}$ and $\tilde{\zeta} \equiv \frac{\mathbb{\hat{E}}[x^\star(z)-\kappa(r^\star)]}{\mathbb{\hat{E}}[A(s^\star(z);z)-\kappa(r^\star)]}$. 
The equilibrium can feature both under- and over-specialization and both under- and over-robustness.
\end{appprop}
\noindent
In general, the equilibrium features inefficient investment in both robustness and specialization. 

\subsection{General Quality Function}
\label{app:gen_quality}
So far, we have looked at an economy in which final producers operate a specific technology. In particular, we assumed that intermediate goods are (i) perfect substitutes in quality (intensive margin) and (ii) complements in quantity (extensive margin). 
While both assumptions matter for the results derived so far, only (ii) is crucial for our new network externality to arise. 
The intuitive reason is that, although goods are complements, an intermediate producer only cares about losing its own share of the surplus if production does not take place. As a consequence, a final producer does not have an instrument to induce bilaterally efficient specialization choices (in a way akin to the inefficiencies induced by limited liabilities). Property (ii) is at the heart of this result. If goods are not complementary at the extensive margin, then no network externality arises, and the economy would be statically constrained-efficient. Note that the same extensive-margin complementarity is underlying some of the results in \cite{Elliott2022_Fragility} and \cite{acemoglu2024macroeconomics}.

We now show that our assumption of perfect substitutability in quality actually induces fewer inefficiencies than a more general quality function would. In particular, suppose that the quality function of a final producer $i$ takes the following generic form:
\begin{align*}
    \mathcal{Q}_i = \chi(A_1 ,\hdots,A_N),
\end{align*}
where $\chi$ is a constant-return-to-scale aggregator, which is increasing and concave in each argument. Importantly, we impose $\frac{\partial \chi}{\partial A_j}>0, \frac{\partial^2 \chi}{\partial A_j^2}<0  \ \forall j$. Hence, we allow for arbitrary substitutability or complementarity in quality (intensive margin), while retaining complementarity in quantity (extensive margin). 

\paragraph{Stationary equilibrium.}
We analyze the equilibrium of the dynamic economy.
Intermediate producers maximize expected operating profits net of product design costs:
\begin{align*}
    V(s,x;z) = \max_{s,x}
    \ \mathcal{D}(s,x)\left(\tilde{A}(s;\textbf{s}_{-i};z)-x\right)-w q(s).
\end{align*}
where $  \mathcal{D}(s,x;N) = \frac{\theta\lambda\Pc(s,x;N)}{\delta}$ and $\tilde{A}(s;\textbf{s}_{-i}) \equiv \frac{\partial \chi}{\partial A_i}A_i$.
Notice that the only difference with respect to the baseline dynamic model is that match surplus potentially depends on the specialization vector of all the other input providers $\textbf{s}_{-i}$. For a generic CES aggregator with substitution parameter $\sigma \in [0,1]$, it holds that $\tilde{A}_i = A_i^\sigma \chi^{1-\sigma}$, which nests our baseline model for $\sigma = 1$.
%
%
Up to the generalized definition of $\tilde{A}$ instead of $A$, equilibrium specialization of intermediate producers with productivity $z$ producing input $i$ is the same as in the baseline dynamic model and given by:
\begin{align*} 
  \mathcal{D}(z;N)\bigg[\tilde{A}^\prime (s^\star(z),\boldsymbol{s}_{-i})+\frac{\phi^\prime(s^\star(z))}{\phi(s^\star(z))}(\tilde{A}(s^\star(z),\boldsymbol{s}_{-i};z)-x^\star(z))\bigg] =w q^\prime(s^\star(z)).   
\end{align*}

\paragraph{Social planner problem.}
The social planner problem reads:
\begin{align*}
    \mathcal{W} = \max_{ \substack{
s_j(z),\; j = 1,\dots,N \\
\quad z \in [\underline{z},\bar{z}]
}  } & \ m \sum_{n=1}^N \int D_n(z) \tilde{A}(s_n(z),\boldsymbol{s}_{-n};z) \ \gamma(z)dz + \psi \log(1-\ell) \\[.4cm]
\text{s.t.} \ & D_j(z) = \frac{1}{\delta} \frac{\mu}{m} \lambda \phi(s_j(z)) e^{-\lambda \hat{\phi}(z,\bar{z})} \prod_{v\neq j} f_v, \ \forall j,z,\\[.2cm]
& \mu = \frac{\delta}{\delta+(1-\delta)\prod_{n=1}^N f_{n}}, \\[.2cm]
& \ell = m \sum_{n=1}^N \int q(s_n(z)) \gamma(z) dz.
\end{align*}
Efficient specialization of intermediate producers with productivity $z$ producing input $j$ solves:
\begin{align*}
  & \mathcal{D}(z;N) \bigg[\tilde{A}^\prime(\Sc(z),\boldsymbol{\Sc}_{-j})+\frac{\phi^\prime(\Sc(z))}{\phi(\Sc(z))}\bigg(\tilde{A}(\Sc(z),\boldsymbol{\Sc}_{-j};z) - f(z)\mathbb{E}_{\max\{\tilde{z}\}|\tilde{z} \leq z}[\tilde{A}(\Sc(\tilde{z}),\boldsymbol{\Sc}_{-j};\tilde{z})]+  \\
    & 
    (N-1)\left(1-f(z)\right)\mathbb{\hat{E}}[\tilde{A}(\Sc(\tilde{z}),\boldsymbol{\Sc}_{-j};\tilde{z})]-\frac{1-\delta}{\delta+(1-\delta)f^N}f^N N \left(1-f(z)\right) \mathbb{\hat{E}}[A(\Sc(\tilde{z}),\boldsymbol{\Sc}_{-j};\tilde{z})]\bigg)\bigg] \\
    &
    = \frac{\psi}{1-N m \bar{q}} \  q^\prime(\Sc(z))-\frac{\theta}{\delta} \sum_{i \neq j}^N \prod_{v \neq i} f_i \  \mathbb{E}_{\max \{\tilde{z}\}}\left[\frac{\partial \tilde{A}(\Sc_{i},\boldsymbol{\Sc}_{-i};\tilde{z})}{\partial  \Sc_j(\tilde{z})}\right]. 
\end{align*}
Efficient specialization with a general quality function coincides with that of the baseline dynamic model with steady-state behavior (\ref{eff_s_dyn_ss}) up to the presence of the final term, which nets out specialization costs by the marginal change in the value added from complementary inputs.
Hence, we observe that a further inefficiency arises if inputs are not perfect substitutes at the intensive margin. The reason is that intermediate producers do not internalize the effect of their specialization on the match surplus generated by other input providers:
\begin{align} \nonumber
    \pd{\Wc}{s(z)}\bigg|_{s(z)=s^\star(z)} \propto & \underbrace{f(z)\mathbb{E}_{\max\{\tilde{z}\}|\tilde{z}\leq z}\left[\tilde{A}(\boldsymbol{s}^\star(\tilde z);\tilde{z})\right]}_{\text{business-stealing externality}} \underbrace{\vphantom{\mathbb{E}_{\max\{\tilde{z}\}|\tilde{z}\leq z}\left[\tilde{A}(\boldsymbol{s}^\star(\tilde z);\tilde{z})\right]}- \ x^\star(z)}_{\substack{\text{appropriability} \\ \text{externality}}}  \underbrace{\vphantom{\mathbb{E}_{\max\{\tilde{z}\}|\tilde{z}\leq z}\left[\tilde{A}(\boldsymbol{s}^\star(\tilde z);\tilde{z})\right]}- \ (N-1)\left( 1-f(z) \right)\mathbb{\hat{E}}[\tilde{A}(\boldsymbol{s}^\star(\tilde{z});\tilde{z})]}_{\text{network externality}}\\ \nonumber
    &  \underbrace{\vphantom{\mathbb{E}\left[\frac{\partial \tilde{\boldsymbol{A}}(\boldsymbol{s}^\star_{-j};z)}{\partial  s_j(z)}\right]}+ \ N \frac{1-\delta}{\delta+(1-\delta)f^N}f^N \left( 1-f(z) \right)\mathbb{\hat{E}}[\tilde{A}(\boldsymbol{s}^\star(\tilde{z});\tilde{z})]}_{\text{search externality}}
    \\ \label{dwds_generalized}
    &  \underbrace{\vphantom{\int_{\underline z}^z}-\ \frac{(N-1)e^{\lambda \hat{\phi}(z,\bar{z})}}{\lambda \phi^\prime(s^\star(z))}\mathbb{E}_{\max\{\tilde{z}\}}\left[\frac{\partial \tilde{\boldsymbol{A}}(\boldsymbol{s}^\star_{-j};\tilde{z})}{\partial  s_j(\tilde{z})}\right] }_{\text{production externality}}.
\end{align}
The new production externality pushes towards equilibrium under-specialization. Hence, 
under the more general aggregator $\chi$, the economy features a further inefficiency: individual specialization influences not only the trading probability of other input providers but
also their match surplus. We characterize the efficiency properties of the economy with general quality function in the next proposition.
\begin{appprop}  [Efficiency of the Dynamic Economy with General Quality Function] \label{prop:eff_dynam_nas}
   The equilibrium is constrained-efficient if and only if:
   $$\left(\mu(f;N) N-1 \right)\mathbb{\hat{E}}[\tilde{A}(s^\star(z);z)]+x_0-\frac{(N-1)}{\lambda \phi^\prime(s^\star(z)) (1-f)} \mathbb{E}_{\max \{\tilde{z}\}}\left[\frac{\partial \tilde{\boldsymbol{A}}(s^\star_{-j};\tilde{z})}{\partial  s_j(\tilde{z})}\right]=0, \ \forall z \in [\underline{z},\bar{z}].$$
\end{appprop}
\noindent
To sum up, when inputs are complements at the extensive margin, equilibrium over-specialization is more likely the more the production process is complex (high $N$) and the more substitutable input quality is (the closer $\sigma$ is to $1$ with CES quality function). 

\subsection{Free Entry} \label{sec:free_entry}
Until now, we have assumed that the mass of intermediate producers $m$ is fixed, that is, a perfectly inelastic entry margin. We now relax this assumption by positing that $m$ is pinned down by a no-arbitrage condition between operating an intermediate or a final firm.
\paragraph{Stationary equilibrium.}
The equilibrium is the same as in the baseline static model, up to the addition of the no-arbitrage condition governing the equilibrium mass of intermediate producers.
The no-arbitrage condition reads:\footnote{In the formulation of equation (\ref{no_arbitrage_condition}) we take a stand on the organization of intermediate production being structured along single firms producing all inputs through distinct product lines.}
\begin{align} \label{no_arbitrage_condition}
\prod_{n=1}^N f_n \sum_{n=1}^N \mathbb{\hat{E}}[x_n(z)] = \theta \prod_{n=1}^N f_n \sum_{n=1}^N \mathbb{\hat{E}}[A_n(s_n(z);z)-x_n(z)]-w \sum_{n=1}^N \bar{q}_n,
\end{align}
where $\bar{q}_n \equiv \int q(s_n(z)) \gamma(z) dz$.
The no-arbitrage condition (\ref{no_arbitrage_condition}) and the labor market clearing condition jointly pin down the equilibrium mass of intermediate producers as the result of a second-order equation:
\begin{align} \label{mass_firms}
 m =& \frac{\left[\prod_{n=1}^N f_n \left(\sum_{n=1}^N \mathbb{\hat{E}}[x_n(z)]+ \sum_{n=1}^N \bar{q}_n \sum_{n=1}^N\mathbb{\hat{E}}[A(s_n(z);z)-x_n(z)]\right)+\psi \sum_{n=1}^N \bar{q}_n \right] \pm \sqrt{\Delta}}{2 \prod_{n=1}^N f_n \sum_{n=1}^N \mathbb{\hat{E}}[x_n(z)] \sum_{n=1}^N \bar{q}_n}
 \\ \nonumber
\Delta =& \left[\prod_{n=1}^N f_n \left(\sum_{n=1}^N \mathbb{\hat{E}}[x_n(z)]+ \sum_{n=1}^N \bar{q}_n \sum_{n=1}^N\mathbb{\hat{E}}[A(s_n(z);z)-x_n(z)]\right)+\psi \sum_{n=1}^N \bar{q}_n \right]^2-4\left(\prod_{n=1}^N f_n\right)^2 \\ \nonumber
& \sum_{n=1}^N \bar{q}_n \sum_{n=1}^N\mathbb{\hat{E}}[x_n(z)] \sum_{n=1}^N\mathbb{\hat{E}}[A(s(z);z)-x_n(z)].
\end{align}
Notice that aggregate output is independent of the mass of intermediate producers.
Since a higher mass of intermediate producers is associated with lower leisure, the two solutions can be Pareto-ranked according to social welfare. Therefore, we assume that the mass of firms equals the lowest solution to (\ref{mass_firms}).

\paragraph{Social planner problem.} We consider the problem of a social planner who is constrained by the free entry condition (\ref{no_arbitrage_condition}). Since the free entry condition pins down $m$ as a function of the surplus sharing rule, we formulate the social planner problem in terms of optimal choice of the surplus offered to final producers, $x(z)$, taking the equilibrium specialization function as given.    
By already assuming symmetry across inputs, the social planner problem reads:
\begin{align} \nonumber
\mathcal{W} = & \max_{ \substack{
x_j(z) \leq A(s(x_j(z));z), \\
j = 1,\dots,N, \ z \in [\underline{z},\bar{z}]
}  }\ \sum_{n=1}^N \prod_{v\neq j} f_v \int \lambda \phi(s(x_n(z))) e^{-\lambda \hat{\phi}(z,\bar{z})} A(s(x_n(z));z) \ \gamma(z)dz + \psi \log(1-\ell) \\[.4cm]
\nonumber
\text{s.t.} \ & \ell = m \sum_{n=1}^N \bar{q}_n, \\[.2cm] 
%
\label{m_constraint} 
& m = \frac{\left[\prod_{n=1}^N f_n\left(\sum_{n=1}^N \mathbb{\hat{E}}[x_n(z)]+ \sum_{n=1}^N \bar{q}_n \sum_{n=1}^N\mathbb{\hat{E}}[A(s(x_n(z));z)-x_n(z)]\right)+\psi \sum_{n=1}^N \bar{q}_n \right] - \sqrt{\Delta}}{2 \prod_{n=1}^N f_n \sum_{n=1}^N \mathbb{\hat{E}}[x_n(z)] \sum_{n=1}^N \bar{q}_n}, \\[.2cm]
\nonumber
  & \frac{1}{m} \lambda \Pc(s(x_j(z)),x_j(z))\bigg[A^\prime(s(x_j(z)))+\frac{\phi^\prime(s(x_j(z)))}{\phi(s(x_j(z)))}
   (A(s(x_j(z));z)-x_j(z))\bigg] = \dots \\
  \label{s_constraint}
   & \frac{\psi}{1-m \sum_{n=1}^N\bar{q}_n} q^\prime(s(x_j(z))) \ \forall j,z. 
\end{align}
The constrained-efficient specialization can be characterized in three steps. First, we compute the full derivative of specialization with respect to the surplus offered to final producers via the implicit function theorem from (\ref{s_constraint}). Second, we compute the derivative of the mass of intermediate producers with respect to $s_j(z)$ and $x_j(z)$ from (\ref{m_constraint}). 
Finally, we maximize social welfare with respect to $x_j(z)$. Importantly, social welfare depends directly on the mass of intermediate producers $m$ only through the hours worked $\ell$. Consequently, the structure of the efficient surplus offer and specialization choice is the same as in the baseline static model (\ref{eff_s}), except for an additional term that internalizes the marginal effect of changes in $m$ on hours worked:
%
%
\begin{align}
    \nonumber
      & \theta \lambda \Pc(z;N) \bigg[A^\prime(\Sc(z))+\frac{\phi^\prime(\Sc(z))}{\phi(\Sc(z))}\bigg(A(\Sc(z);z) - f(z)\mathbb{E}_{\max\{\tilde{z}\}|\tilde{z} \leq z}[A(\Sc(\tilde{z};\tilde{z}))]  \\ \nonumber
    & 
    + (N-1)\left(1-f(z)\right)\mathbb{\hat{E}}[A(\Sc(\tilde{z});\tilde{z})]\bigg)\bigg]
    = \frac{\psi}{1-N m \bar{q}} \  q^\prime(\Sc(z)) \\ \label{eff_x_fec}
    & + \underbrace{\frac{\psi}{1-N m \bar q} \left(\frac{\partial m(\Sc(z),x(z))}{\partial \Sc(z)}+\frac{\partial m(\Sc(z),x(z))/\partial x(z)}{\partial \Sc(z)/\partial x(z)}\right)\frac{N \bar{q}}{m \gamma(z)}}_{\text{entry externality}}.
\end{align}
The first two lines equal the efficient specialization condition of the baseline static model.\footnote{Since $\left(\frac{\partial s(x(z))}{\partial x(z)}\right)^{-1}=\frac{\partial x(s(z))}{\partial s(z)}$ by the inverse-function theorem, maximizing with respect to $x(z)$ by taking $s(x(z))$ as given is indeed isomorphic to maximizing with respect to $s(z)$ by taking $x(z)$ as given. 
} The third line highlights the new source of inefficiency related to the external effect of individual specialization on the equilibrium mass of intermediate producers.
Comparing equations (\ref{eff_x_fec}) and (\ref{eff_s}), we observe that the former nests the latter for fixed $m$. The additional inefficiency term brought about by free entry crucially depends on the sign of the total derivative of the mass of intermediate producers with respect to the surplus offered: if it is positive, the entry externality pushes towards equilibrium over-specialization; if negative, towards under-specialization. Either sign is possible in equilibrium. 
In any case, Proposition \ref{prop:eff_dynam_fec} shows that equilibrium specialization is inefficient whenever it affects entry decisions.
\begin{appprop} [Efficiency of the Static Economy: Free Entry] \label{prop:eff_dynam_fec}
   The static economy with free entry is always inefficient. 
\end{appprop}

\section{Proofs}
\label{app:proofs_resilience}

\begin{proof}[Proof of Lemma \ref{lem:res_surplus}]
Substituting the quality function (\ref{def:quality_fnct}) into the profits of a final producer $i$ (\ref{final_profits}) and using the equilibrium condition $P_i=\mathcal{Q}_i$ yields:
$$\pi_i = \sum_{j=1}^N \left(A_j-p_j\right) \mathbbm{1}\{Y_i=1\} = \sum_{j=1}^N x_j \mathbbm{1}\{Y_i=1\},$$
where the second equality follows from substituting for the identity  $p_j = A_j - x_j$.

 Therefore, the profits of an active final producer equal $\sum_{j=1}^N x_j$. The reservation surplus on input $j$ makes the final producer indifferent between sourcing the intermediate input in the market and not producing at all in expectation. Formally, $x_{0,j} + \sum_{n \neq j} \mathbb{\hat{E}}[x_n] = 0 \implies x_{0,j} = -  \sum_{n \neq j} \mathbb{\hat{E}}[x_n]$. Hence, in a symmetric stationary equilibrium, $x_{0} = -  (N-1) \mathbb{\hat{E}}[x^\star(z)]$.
\end{proof}

\begin{proof}[Proof of Lemma \ref{lem:offered_surplus}]
See the derivation in Appendix \ref{app:x_interpretation}.
\end{proof}

\begin{proof}[Proof of Proposition \ref{prop:externalities}]
The result follows directly from rearranging the marginal welfare effect of equilibrium specialization derived in Theorem \ref{thm:over-specialization}.
\end{proof}

\begin{proof}[Proof of Proposition \ref{prop:offsetting}]
The result follows directly from rearranging the marginal welfare effect of equilibrium specialization derived in Theorem \ref{thm:over-specialization}.
\end{proof}

\begin{proof}[Proof of Theorem \ref{thm:over-specialization}]
To establish the claim, we evaluate the marginal welfare effect of specialization at the equilibrium solution, as implicitly defined by equation (\ref{eq_s}). By definition, efficient specialization, as implicitly defined by equation (\ref{eff_s}), equalizes the marginal welfare effect of specialization to zero, i.e., $\pd{\Wc}{s(z)}\big|_{s(z)=\mathcal{S}(z)}=0$. 
\begin{align*}
   \pd{\Wc}{s(z)}\bigg|_{s(z)=s^\star(z)}= \  & \theta \lambda \Pc(z;N) \bigg[A^\prime(s^\star(z))+\frac{\phi^\prime(s^\star(z))}{\phi(s^\star(z))}\bigg(A(s^\star(z);z)  - f(z)\mathbb{E}_{\max\{\tilde{z}\}|\tilde{z} \leq z}[A(s^\star(\tilde{z});\tilde{z})]  \\ \nonumber
    & 
    + (N-1)\left(1-f(z)\right)\mathbb{\hat{E}}[A(s^\star(\tilde{z});\tilde{z})] \bigg)\bigg] 
    - \frac{\psi}{1-N m \bar{q}} \  q^\prime(s^\star(\tilde{z})) \\[.2cm]
   = \  &  \theta \lambda \Pc(z;N) \bigg[A^\prime(s^\star(z))+\frac{\phi^\prime(s^\star(z))}{\phi(s^\star(z))}\bigg(A(s^\star(z);z)  - f(z)\mathbb{E}_{\max\{\tilde{z}\}|\tilde{z} \leq z}[A(s^\star(\tilde{z});\tilde{z})]  \\ \nonumber
    & 
    + (N-1)\left(1-f(z)\right)\mathbb{\hat{E}}[A(s^\star(\tilde{z});\tilde{z})] \bigg)\bigg] 
    - w  q^\prime(s^\star(z)) \\[.2cm]
    = \  & \theta \lambda \Pc(z;N)\frac{\phi^\prime(s^\star(z))}{\phi(s^\star(z))} \bigg[x^\star(z) - f(z)\mathbb{E}_{\max\{\tilde{z}\}|\tilde{z} \leq z}[A(s^\star(\tilde{z});\tilde{z})] \\ \nonumber
    &  + (N-1)\left(1-f(z)\right)\mathbb{\hat{E}}[A(s^\star(\tilde{z});\tilde{z})] \bigg] \\[.2cm]
    = \  & \theta \lambda \Pc(z;N)\frac{\phi^\prime(s^\star(z))}{\phi(s^\star(z))} (N-1)\left(1-f(z)\right)\left(\mathbb{\hat{E}}[A(s^\star(\tilde{z});\tilde{z})]-\mathbb{\hat{E}}[x^\star(\tilde{z})]\right) \\[.2cm] \nonumber
    < \ & 0, \quad \forall N>1,
\end{align*}
where the second equality follows from the labor market clearing condition, $w=\frac{\psi}{1-Nm\bar q}$, the third equality follows from the equilibrium specialization condition (\ref{eq_s}) and the fourth equality from the equilibrium surplus offered to final producers (\ref{eq_x_solved}).

We establish the first statement by noting that the marginal welfare effect of equilibrium specialization is zero if and only if $N=1$.
Whenever $N>1$, the marginal welfare effect of equilibrium specialization is negative since $\phi^\prime(s)<0$.  
Hence, when the production process is complex, the planner would choose a lower level of specialization: the equilibrium features \textit{over-specialization}.

\end{proof}

\begin{proof}[Proof of Remark \ref{remark:bargaining}]
See Appendix \ref{app:bargaining}.
\end{proof}

\begin{proof}[Proof of Remark \ref{remark:MP_CES}]
The result follows as a corollary of Proposition \ref{prop:offsetting}.
\end{proof}

\begin{proof}[Proof of Remark \ref{rem:dir_search}]
See Appendix \ref{app:directed_search}.
\end{proof}

\begin{proof}[Proof of Remark \ref{remark:noncontingent}]
Following the same steps as in the proof of Theorem \ref{thm:over-specialization}, we obtain:
\begin{align*}
   \pd{\Wc}{s(z)}\bigg|_{s(z)=s^\star_{NC}(z)}= \  & \theta \lambda \Pc(z;N) \bigg[A^\prime(s^\star_{NC}(z))+\frac{\phi^\prime(s^\star_{NC}(z))}{\phi(s^\star_{NC}(z))}\bigg(A(s^\star_{NC}(z);z)  - f(z)\mathbb{E}_{\max\{\tilde{z}\}|\tilde{z} \leq z}[A(s^\star_{NC}(\tilde{z});\tilde{z})]  \\ \nonumber
    & 
    + (N-1)\left(1-f(z)\right)\mathbb{\hat{E}}[A(s^\star_{NC}(\tilde{z});\tilde{z})] \bigg)\bigg] 
    - \frac{\psi}{1-N m \bar{q}} \  q^\prime(s^\star_{NC}(\tilde{z})) \\[.2cm]
   = \  &  \theta \lambda \Pc(z;N) \bigg[A^\prime(s^\star_{NC}(z))+\frac{\phi^\prime(s^\star_{NC}(z))}{\phi(s^\star_{NC}(z))}\bigg(A(s^\star_{NC}(z);z)  - f(z)\mathbb{E}_{\max\{\tilde{z}\}|\tilde{z} \leq z}[A(s^\star_{NC}(\tilde{z});\tilde{z})]  \\ \nonumber
    & 
    + (N-1)\left(1-f(z)\right)\mathbb{\hat{E}}[A(s^\star_{NC}(\tilde{z});\tilde{z})] \bigg)\bigg] 
    - w  q^\prime(s^\star_{NC}(z)) \\[.2cm]
    = \  & \theta \lambda \Pc(z;N)\left(1-f^{-(N-1)}\right)\bigg[A^\prime(s^\star_{NC}(z))+\frac{\phi^\prime(s^\star_{NC}(z))}{\phi(s^\star_{NC}(z))}A(s^\star_{NC}(z);z)\bigg]
     \\
    & 
    + \theta \lambda \Pc(z;N)\frac{\phi^\prime(s^\star_{NC}(z))}{\phi(s^\star_{NC}(z))} \bigg[f^{-(N-1)} x^\star(z) - f(z)\mathbb{E}_{\max\{\tilde{z}\}|\tilde{z} \leq z}[A(s^\star_{NC}(\tilde{z});\tilde{z})]  \\
    & + (N-1)\left(1-f(z)\right)\mathbb{\hat{E}}[A(s^\star_{NC}(\tilde{z});\tilde{z})] \bigg] \\[.2cm]
    = \  & \theta \lambda \Pc(z;N)\left(1-f^{-(N-1)}\right)\bigg[A^\prime(s^\star_{NC}(z))+\frac{\phi^\prime(s^\star_{NC}(z))}{\phi(s^\star_{NC}(z))} \bigg(A(s^\star_{NC}(z);z)-f(z) \\
    & \mathbb{E}_{\max\{\tilde{z}\}|\tilde{z} \leq z}[A(s^\star_{NC}(\tilde{z});\tilde{z})] \bigg)\bigg]
    + \theta \lambda \Pc(z;N)\frac{\phi^\prime(s^\star_{NC}(z))}{\phi(s^\star_{NC}(z))} \left(1-f(z)\right)  \\
    & \bigg( (N-1)\mathbb{\hat{E}}[A(s^\star_{NC}(\tilde{z});\tilde{z})]-f^{-(N-1)}x_0\bigg) \\[.2cm]
    = \ & \theta \lambda \phi(s^\star_{NC}(z)) (1-f(z)) \Bigg[- \bigg[A^\prime(s^\star_{NC}(z))+\frac{\phi^\prime(s^\star_{NC}(z))}{\phi(s^\star_{NC}(z))} A(s^\star_{NC}(z);z) \bigg](1-f^{N-1}) \\
    & + \frac{\phi^\prime(s^\star_{NC}(z))}{\phi(s^\star_{NC}(z))} \bigg(f(z)(1-f^{N-1})\mathbb{E}_{\max\{\tilde{z}\}|\tilde{z} \leq z}[A(s^\star_{NC}(\tilde{z});\tilde{z})]-  \left(1-f(z)\right) \\
    & \bigg( f^{N-1}(N-1)\mathbb{\hat{E}}[A(s^\star_{NC}(\tilde{z});\tilde{z})]-x_0\bigg)\bigg)\Bigg] \\[.2cm] \nonumber
    < \ & 0, \quad \forall N >1.
\end{align*}
The marginal welfare effect of equilibrium specialization can be decomposed into two wedges. The first wedge arises from the non-contingency of contracts (\textit{contracting externality}). The second wedge arises from the gap between network externality and appropriability externality across markets. Both wedges push the equilibrium towards over-specialization.
\end{proof}

\begin{proof}[Proof of Remark \ref{remark:notrade_payments}]
See Appendix \ref{app:non-cont}.
\end{proof}

\begin{proof}[Proof of Remark \ref{remark:vertical_integration}]
The result follows as a corollary of Theorem \ref{thm:over-specialization}.
\end{proof}

\begin{proof}[Proof of Remark \ref{remark:in-house}]
See Appendix \ref{app:in-house}.
\end{proof}

\begin{proof}[Proof of Lemma \ref{lem:res_surplus_dyn}]
See the derivation in Appendix \ref{app:dyn}.
\end{proof}

\begin{proof}[Proof of Proposition \ref{prop:externalities_dyn}]
The result follows directly from rearranging the marginal welfare effect of equilibrium specialization derived in Theorem \ref{thm:eff_dynam}.
\end{proof}

\begin{proof}[Proof of Theorem \ref{thm:eff_dynam}]
To establish the claim, we follow the same strategy as in Theorem \ref{thm:over-specialization}. Evaluating the marginal welfare effect of specialization at the equilibrium solution yields:
\begin{align*}
   \pd{\Wc}{s(z)}\bigg|_{s(z)=s^\star(z)}= \  & \theta \lambda \mathcal{D}(z;N) \bigg[A^\prime(s^\star(z))+\delta[1+\beta(1-\delta)]\frac{\phi^\prime(s^\star(z))}{\phi(s^\star(z))}\bigg(A(s^\star(z);z)    \\ \nonumber
    & 
     - f(z)\mathbb{E}_{\max\{\tilde{z}\}|\tilde{z} \leq z}[A(s^\star(\tilde{z});\tilde{z})] +(N-1) \left(1-f(z)\right)\mathbb{\hat{E}}[A(s^\star(z);z)]  \\ \nonumber
     & 
    -\frac{\beta(1-\delta)}{1+\beta(1-\delta)}f^N N \left(1-f(z)\right)\mathbb{\hat{E}}[A(s^\star(z);z)]\bigg)\bigg] - \frac{\psi}{1-N m \bar{q}} \  q^\prime(s^\star(z)) \\[.2cm]
   = \  &  \theta \lambda \Dc(z;N) \bigg[A^\prime(s^\star(z))+\delta[1+\beta(1-\delta)]\frac{\phi^\prime(s^\star(z))}{\phi(s^\star(z))}\bigg(A(s^\star(z);z)    \\ \nonumber
    & 
     - f(z)\mathbb{E}_{\max\{\tilde{z}\}|\tilde{z} \leq z}[A(s^\star(\tilde{z});\tilde{z})] +(N-1) \left(1-f(z)\right)\mathbb{\hat{E}}[A(s^\star(z);z)]  \\ \nonumber
    & -\frac{\beta(1-\delta)}{1+\beta(1-\delta)}f^N N \left(1-f(z)\right)\mathbb{\hat{E}}[A(s^\star(z);z)] \bigg)\bigg] - w  q^\prime(s^\star(z)) \\[.2cm]
    = \  & \theta \lambda \Dc(z;N)\delta[1+\beta(1-\delta)]\frac{\phi^\prime(s^\star(z))}{\phi(s^\star(z))} \bigg[x^\star(z) - f(z)\mathbb{E}_{\max\{\tilde{z}\}|\tilde{z} \leq z}[A(s^\star(\tilde{z});\tilde{z})]   \\
    &  + (N-1)\left(1-f(z)\right)\mathbb{\hat{E}}[A(s^\star(z);z)] -\frac{\beta(1-\delta)}{1+\beta(1-\delta)}f^N N \left(1-f(z)\right)\mathbb{\hat{E}}[A(s^\star(z);z)]\bigg] \\[.2cm]
    = \  & \theta \lambda \Dc(z;N)\delta[1+\beta(1-\delta)]\frac{\phi^\prime(s^\star(z))}{\phi(s^\star(z))} \left(1-f(z)\right)\left[\left(\mathcal{M}(N)N-1\right)\mathbb{\hat{E}}[A(s^\star(z);z)]+x_0\right] \\[.2cm] \nonumber
    < \ & 0, \quad \forall N >1,
\end{align*}
where $\mathcal{M}(N) \equiv 1-\frac{\beta(1-\delta)}{1+\beta(1-\delta)}f^N$. The equilibrium is constrained-efficient if and only if $(N\mathcal{M}(N)-1)\mathbb{\hat{E}}[A(s^\star(z);z)]+x_0=0$. Substituting for the optimal reservation surplus (\ref{res_surplus_dyn}), the constrained efficiency condition reads:
\begin{align*}
    (N\mathcal{M}(N)-1)\mathbb{\hat{E}}[A(s^\star(z);z)]-\left(\alpha(N)N-1\right)\mathbb{\hat{E}}[x^\star(z)]=0.
\end{align*}
If the left-hand side is positive (negative), the equilibrium exhibits over-(under-)specialization.
Since $\mathbb{\hat{E}}[x^\star(z)] \in (0,\mathbb{\hat{E}}[A(s^\star(z);z)])$, we can characterize the sign of the inefficiency analytically.

First, suppose that $\alpha(N)N-1>0$. Then, the equilibrium exhibits over-specialization if: 
%
%
\begin{align*}
& (N\mathcal{M}(N)-1) - (\alpha(N) N - 1) > 0 \\[2mm]
%
%
\iff & \; \left[1 - \frac{\beta(1-\delta)}{1+\beta(1-\delta)} f^N \right]
- \left[ 1 - \frac{\beta(1-\delta) f^N}{1-\beta(1-\delta)(1-f^N)} \right] > 0 \\[1mm]
\iff & \; \beta (1-\delta) f^N \left[ \frac{1}{1-\beta(1-\delta)(1-f^N)} - \frac{1}{1+\beta(1-\delta)} \right] > 0 \\[1mm]
\iff & \; 1+\beta(1-\delta) > 1-\beta(1-\delta)(1-f^N) \\[1mm]
\iff & \; \beta(1-\delta)[1+(1-f^N)]>0,
\end{align*}
which is always the case for $\beta \in (0,1), \delta \in (0,1), f \in (0,1),$ and $N \geq 1$.

Second, suppose that $\alpha(N)N-1<0$. Then, the equilibrium exhibits over-specialization if: 
\begin{align*}
    N\mathcal{M}(N)-1>0 \iff N>\frac{1}{\mathcal{M}(N)} \impliedby N>1,
\end{align*}
Notice that $\frac{1}{\mathcal{M}(N)}=\frac{1+\beta(1-\delta)}{1+\beta(1-\delta)(1-f^N)}$. Since $\beta \in (0,1), \delta \in (0,1),$ and $f \in (0,1)$, it follows that $\frac{1}{\mathcal{M}(N)} \in (1,1+\beta(1-\delta))$, where $1+\beta(1-\delta)<2$.

Overall, if $N>1$, the dynamic equilibrium always displays over-specialization.

\end{proof}

\begin{proof}[Proof of Lemma \ref{lem:lem:res_surplus_link}]
See the derivation in Appendix \ref{app:links}.
\end{proof}

\begin{proof}[Proof of Theorem \ref{prop:eff_dynam_links}]
    To establish the claim, we follow the same strategy as in Theorem \ref{thm:over-specialization}. Evaluating the marginal welfare effect of specialization at the equilibrium solution yields:
\begin{align*}
   \pd{\Wc}{s(z)}\bigg|_{s(z)=s^\star(z)}= \  & \theta \lambda \Dc(z;N) \bigg[A^\prime(s^\star(z))+[1-(1-\delta)\rho(f;N-1)][1+\beta(1-\delta)\rho(f;N-1)]\frac{\phi^\prime(s^\star(z))}{\phi(s^\star(z))}  \\ \nonumber
    & 
     \bigg(A(s^\star(z);z)- f(z)\mathbb{E}_{\max\{\tilde{z}\}|\tilde{z} \leq z}[A(s^\star(\tilde{z});\tilde{z})]  + \chi_1(f;N) (N-1)\left(1-f(z)\right)\mathbb{\hat{E}}[A(s^\star(z);z)]   \\
    &
    -\chi_2(f;N)\frac{\beta(1-\delta)}{1+\beta(1-\delta)\rho(f;N-1)}f^N N\left(1-f(z)\right)\mathbb{\hat{E}}[A(s^\star(z);z)]\bigg)\bigg]  \\
    & - \frac{\psi}{1-N m \bar{q}} \  q^\prime(s^\star(z)) \\[.2cm]
   = \  &  \theta \lambda \Dc(z;N) \bigg[A^\prime(s^\star(z))+[1-(1-\delta)\rho(f;N-1)][1+\beta(1-\delta)\rho(f;N-1)]\frac{\phi^\prime(s^\star(z))}{\phi(s^\star(z))}  \\ \nonumber
    & 
     \bigg(A(s^\star(z);z)- f(z)\mathbb{E}_{\max\{\tilde{z}\}|\tilde{z} \leq z}[A(s^\star(\tilde{z});\tilde{z})] + \chi_1(f;N)(N-1)\left(1-f(z)\right)\mathbb{\hat{E}}[A(s^\star(z);z)]  \\
    &  -\chi_2(f;N)\frac{\beta(1-\delta)}{1+\beta(1-\delta)\rho(f;N-1)}f^N N  \left(1-f(z)\right)\mathbb{\hat{E}}[A(s^\star(z);z)] \bigg)\bigg] - w  q^\prime(s^\star(z)) \\[.2cm]
    = \  & \theta \lambda \Dc(z;N)[1-(1-\delta)\rho(f;N-1)][1+\beta(1-\delta)\rho(f;N-1)]\frac{\phi^\prime(s^\star(z))}{\phi(s^\star(z))} \bigg[x^\star(z)  \\
    &  - f(z)\mathbb{E}_{\max\{\tilde{z}\}|\tilde{z} \leq z}[A(s^\star(\tilde{z});\tilde{z})] +  \chi_1(f;N)(N-1)\left(1-f(z)\right)\mathbb{\hat{E}}[A(s^\star(z);z)] \\
    & -\chi_2(f;N)\frac{\beta(1-\delta)}{1+\beta(1-\delta)}f^N N  \left(1-f(z)\right)\mathbb{\hat{E}}[A(s^\star(z);z)]\bigg]\\[.2cm]
    = \  & \theta \lambda \Dc(z;N)[1-(1-\delta)\rho(f;N-1)][1+\beta(1-\delta)\rho(f;N-1)]\frac{\phi^\prime(s^\star(z))}{\phi(s^\star(z))} \left(1-f(z)\right) \\
    & \bigg[\chi_1(f;N)\left(\widetilde{\mathcal{M}}(N)N-1 \right) \mathbb{\hat{E}}[A(s^\star(z);z)] +x_0 \bigg] \\[.2cm] \nonumber
    < \  & 0, \quad \forall N > \underline{N},
\end{align*}
where $\widetilde{\mathcal{M}}(N) \equiv 1-\frac{\beta(1-\delta)}{1+\beta(1-\delta)}f^N \frac{\chi_2(f;N)}{\chi_1(f;N)}$. The equilibrium is constrained-efficient if and only if $\chi_1(f;N)(N\widetilde{\mathcal{M}}(N)-1)\mathbb{\hat{E}}[A(s^\star(z);z)]+x_0=0$. Substituting for the equilibrium reservation surplus (\ref{res_surplus_links}), the constrained efficiency condition reads:
\begin{align*}
    \chi_1(f;N)(N\widetilde{\mathcal{M}}(N)-1)\mathbb{\hat{E}}[A(s^\star(z);z)]-\left(\tilde{\alpha}(N)N-1\right)\mathbb{\hat{E}}[x^\star(z)]=0.
\end{align*}
If the left-hand side is positive (negative), the equilibrium exhibits over-(under-)specialization.
Notice that $0<\mathbb{\hat{E}}[x^\star(z)] = f\frac{\mathbb{\hat{E}} \left[f(z)\mathbb{E}_{\max{\tilde{z}}|\tilde{z}<z} [A(s^\star(\tilde{z});\tilde{z})]\right]}{f-(1-f)\ln{1-f}\left(\tilde{\alpha}(N)N-1\right)} < \frac{f^2\mathbb{\hat{E}} \left[A(s^\star(\tilde{z});\tilde{z})\right]}{{f-(1-f)\ln{1-f}\left(\tilde{\alpha}(N)N-1\right)}}.$ Hence, we can characterize the sign of the inefficiency analytically.

First, suppose that $N>\frac{1}{\tilde{\alpha}(N)}$.
Then, the equilibrium exhibits over-specialization if: 
\begin{align*}
    &\chi_1(f;N)(N\widetilde{\mathcal{M}}(N)-1)-f^2\frac{\tilde{\alpha}(N)N-1}{{f-(1-f)\ln{1-f}\left(\tilde{\alpha}(N)N-1\right)}}>0 \\
    \iff & \chi_1(f;N)(N\widetilde{\mathcal{M}}(N)-1)-\frac{f^2}{f \frac{1-\beta(\rho(N)-f^N)}{1-\beta \left[(\rho(N)-f^N)+(1-\beta)\left(f^N-\delta f \rho(N-1)\right) \right]}-(1-f)\ln{1-f}}>0.
\end{align*}
Let $F(f;\delta,N,\beta) \equiv \chi_1(f;N)(N\widetilde{\mathcal{M}}(N)-1)-\frac{f^2}{f \frac{1-\beta(\rho(N)-f^N)}{1-\beta \left[(\rho(N)-f^N)+(1-\beta)\left(f^N-\delta f \rho(N-1)\right) \right]}-(1-f)\ln{1-f}}$. 
By Lemma~\ref{lemma:ismonotonic}, $F(f;\delta,N,\beta)$ is always positive as $N$ grows large. 
Hence, there exists a finite $\underline N > 1$ such that
for all $N > \underline N$, the equilibrium exhibits over-specialization.

Second, suppose $N<\frac{1}{\tilde{\alpha}(N)}$. Since $\chi_1(f;N)>0$, the equilibrium exhibits over-specialization if: 
\begin{align*}
    N\widetilde{\mathcal{M}}(N)-1>0 \quad \forall N>1,
\end{align*}
which is always the case by Lemma \ref{lemma:ntwrk>search}.

 Overall, if $N$ is large enough, the dynamic equilibrium with link destruction displays over-specialization.
\end{proof}

\begin{proof}[Proof of Remark \ref{rem:inv_mult}]
See Appendix \ref{app:inv_mult}
\end{proof}

\begin{proof}[Proof of Remark \ref{rem:weakest}]
From equation (\ref{eq_s_links}), optimal specialization is decreasing in the disruption probability $\delta_j$. It follows that intermediate producers more prone to disruption will specialize less their products. Depending on the sensitivity of the average compatibility probability to specialization, the derivative $\frac{\partial}{\partial \delta_j} \left(\delta_j [1-f(\bar{\phi}_j)]\right)$ may be positive or negative. 
\end{proof}

\begin{proof}[Proof of Proposition \ref{prop:norm_static}]
Suppose the Social Planner pays a lump-sum subsidy $T$ to active final producers. Then, the reservation surplus (\ref{eq:res_surplus}) becomes:
\begin{align*}
 \pi(x_{0,j};\mathbb{\hat{E}}[\boldsymbol{x}_{-j}],T) = 0 
    & \iff x_{0,j} = - \left(\sum_{n \neq j} \mathbb{\hat{E}}[x_n]+T\right) = -\left[(N-1) \mathbb{\hat{E}}[x^\star(z)] +T\right] \quad \forall j.
\end{align*}
Substituting for the reservation surplus in the equilibrium surplus offered function (\ref{eq_x}) yields:
$$x^\star(z) = -\left(1-f(z)\right) \left[(N-1) \mathbb{\hat{E}}[x^\star(\tilde{z})]+T\right]+ f(z)\mathbb{E}_{\max\{\tilde{z}\}|\tilde{z}\leq z}[A\left(s^\star(\tilde{z});\tilde{z}\right)].$$
Comparing this expression with the constrained-efficient surplus offered (\ref{eff_surplus_static}), we conclude that the two coincide if and only if $T=T^\star \equiv (N-1) \mathbb{\hat{E}}[A(s^\star(\tilde{z});\tilde{z})-x^\star(\tilde{z})]$.
\end{proof}

\begin{proof}[Proof of Proposition \ref{prop:norm_dynamic}]
    Following the same steps of the proof of Proposition \ref{prop:norm_static}, the reservation surplus (\ref{res_surplus_dyn}) becomes:
    $$V(x_{0,j};\mathbb{\hat{E}}[\boldsymbol{x}_{-j}],T)=S \iff x_0 = -\left[(\alpha(N)N-1)\mathbb{\hat{E}}[x^\star(z)]+\alpha(N)T\right] \quad \forall j.$$
    Substituting for the reservation surplus in the equilibrium surplus offered function (\ref{eq_x}) yields:
$$x^\star(z) = -\left(1-f(z)\right)\left[(\alpha(N)N-1)\mathbb{\hat{E}}[x^\star(z)]+\alpha(N)T\right]+ f(z)\mathbb{E}_{\max\{\tilde{z}\}|\tilde{z}\leq z}[A\left(s^\star(\tilde{z});\tilde{z}\right)].$$
 Comparing this expression with the constrained-efficient surplus offered (\ref{eff_surplus_dynamic}), we conclude that the two coincide if and only if $T=T^\star_d \equiv \frac{(\mathcal{M}(N)N-1) \mathbb{\hat{E}}[A(s^\star(\tilde{z});\tilde{z})]-(\alpha(N)N-1)\mathbb{\hat{E}}[x^\star(\tilde{z})]}{\alpha(N)}$.
\end{proof}

\begin{proof}[Proof of Proposition \ref{prop:stds}]
From the proof of Theorem \ref{thm:eff_dynam}, we observe that the marginal welfare effect of specialization at the equilibrium solution is negative for any productivity level if the production process is complex. Formally, $N>1 \implies \frac{\partial \mathcal{W}}{\partial s(z)} \big|_{s(z)=s^\star(z)} <0 \, \forall z$. It follows that constraining the specialization of the highest-specialization firms is welfare-improving.
\end{proof}

\begin{proof}[Proof of Proposition \ref{prop:bargaining}]
Since the efficient specialization is independent of the surplus-sharing rule,
the result follows directly from applying Theorem \ref{thm:over-specialization} to the equilibrium specialization function of the bargaining model (\ref{eq_s_barg2}).
\end{proof}

\begin{proof}[Proof of Lemma \ref{lemma:ntwrk>search}]

    Let $N^\star(f,\delta,N,\beta) \equiv \frac{1}{\widetilde{\mathcal{M}}(N)}= \left(1-f^N\frac{\beta(1-\delta)}{1+\beta(1-\delta)}\frac{\chi_2(f;N)}{\chi_1(f;N)}\right)^{-1}$. If $N-N^\star(f,\delta,N,\beta)>0$, then $N\widetilde{\mathcal{M}}(N)-1>0, \quad \forall N>1$. We aim to bound $N^\star(f,\delta,N,\beta)$ from above for any combination of the structural parameters \textit{and} the endogenous input finding probability in its domain, i.e., $f \in (0,1)$. Spanning the entire domain of $f$ allows us to make a global statement -- independent of specific functional form assumptions. First, we reduce the parameter space by noting that $\frac{\beta(1-\delta)}{1+\beta(1-\delta)} \in \left(0,\frac{1}{2}\right)$. Then, we compute numerically $N-\bar{N}^\star(f,\delta,N) \equiv N-\left(1-f^N\frac{1}{2}\frac{\chi_2(f;N)}{\chi_1(f;N)}\right)^{-1}$ for every $(\delta,f,N)$ combination. Results are displayed in Figure (\ref{fig:over-spec_links2}). Since the minimum value of $N-\bar{N}^\star(f,\delta,N)$ detected is strictly positive, we conclude that $\widetilde{\mathcal{M}}(N) \in \left(\frac{1}{2},1\right]$.
\begin{figure}[ht!]
\begin{center}
\caption{Numerical check Lemma \ref{lemma:ntwrk>search}}
\label{fig:over-spec_links2}
\includegraphics[width=0.6\textwidth,keepaspectratio]{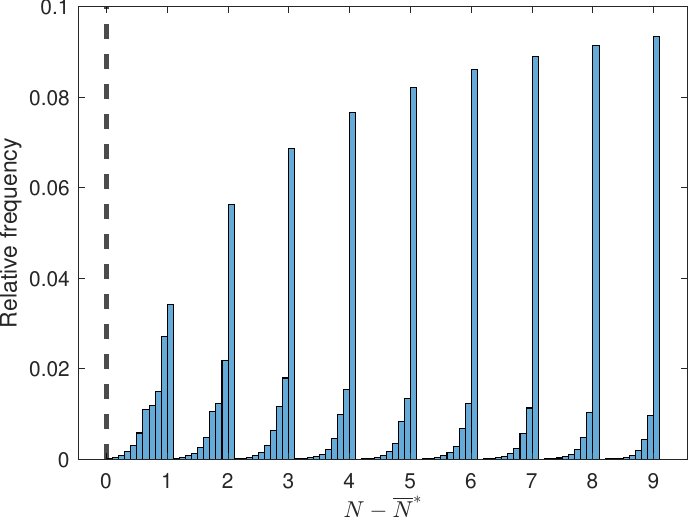}
\begin{minipage}{0.7\textwidth} \scriptsize{} \textit{Note}:  The histogram report the relative frequency of the values of $N-\bar{N}^\star(f,\delta,N)$ for any combination of $\delta, f \in \{0.001,0.002,0.003,\dots,0.999\}$, and $N \in \{2,3,4,\dots,10\}$.
\end{minipage} 
\end{center}
\end{figure}
\end{proof}

\begin{proof}[Proof of Lemma \ref{lemma:ismonotonic}]
We would like to characterize the asymptotic behavior of 
$F(f;\delta,N,\beta)$ as $N$ approaches infinity. 

First, we compute the limits of the multipliers $\chi_1(N), \chi_2(N),$ and $\tilde{\alpha}(N)$:
\begin{align*}
    \lim_{N \rightarrow +\infty} \chi_1(N) &= 1, \\
    \lim_{N \rightarrow +\infty} \chi_2(N) &= \frac{1+\beta(1-\delta)}{1-\delta}, \\
    \lim_{N \rightarrow +\infty} \tilde{\alpha}(N) &= 1.
\end{align*}
Then, we use them to compute the limit of the over-specialization function:
\begin{align*}
    & \lim_{N \rightarrow +\infty}  F(f;\delta,N,\beta) \\
    = & \lim_{N \rightarrow +\infty}   \underbrace{\chi_1(N) \left[\left(1-\frac{\beta(1-\delta)}{1+\beta(1-\delta)}f^N \frac{\chi_2(N)}{\chi_1(N)}\right)N-1 \right]}_{\rightarrow + \infty}-\underbrace{\frac{f^2}{ \frac{f}{\tilde{\alpha}(N)N-1}-(1-f)\ln{1-f}}}_{\rightarrow \frac{f^2}{-(1-f)\ln{1-f}}}  \\
    = & +\infty.
\end{align*}
Hence, $\exists \underline N>1$
such that $F(f;\delta,N,\beta) > 0 \quad \forall N > \underline N$.

\end{proof}

\begin{proof}[Proof of Proposition \ref{prop:eff_dynam_ss}]
Following the same steps as in the proof of Theorem \ref{thm:eff_dynam} yields:
\begin{align*}
   \pd{\Wc}{s(z)}\bigg|_{s(z)=s^\star(z)}= \  & \theta \lambda \Dc(z;N) \bigg[A^\prime(s^\star(z))+\frac{\phi^\prime(s^\star(z))}{\phi(s^\star(z))}\bigg(A(s^\star(z);z)  - f(z)\mathbb{E}_{\max\{\tilde{z}\}|\tilde{z} \leq z}[A(s^\star(\tilde{z});\tilde{z})]  \\ \nonumber
    & 
    + (N-1)\left(1-f(z)\right)\mathbb{\hat{E}}[A(s^\star(z);z)] -\frac{1-\delta}{\delta+(1-\delta)f^N}f^N N \left(1-f(z)\right)\mathbb{\hat{E}}[A(s^\star(z);z)]\bigg)\bigg] 
    \\ \nonumber
    & - \frac{\psi \  q^\prime(s^\star(z))}{1-N m \bar{q}}\\[.2cm]
    = \  & \theta \lambda \Dc(z;N)\frac{\phi^\prime(s^\star(z))}{\phi(s^\star(z))} \left(1-f(z)\right) \left(\left[N\left(1-\frac{1-\delta}{\delta+(1-\delta)f^N}f^N\right)-1\right]\mathbb{\hat{E}}[A(s^\star(z);z)]+x_0\right) \\[.2cm] \nonumber
    \lesseqgtr \ & 0.
\end{align*}
Notice that the multiplier $1-\frac{1-\delta}{\delta+(1-\delta)f^N}f^N$ equals the share of searching final producers $\mu(f;N)$. The equilibrium is constrained-efficient if and only if $(N\mu(f;N)-1)\mathbb{\hat{E}}[A(s^\star(z);z)]+x_0$. If $N>\frac{1}{\mu(f;N)}$, the equilibrium exhibits over-specialization.

Substituting for the equilibrium reservation surplus, the constrained efficiency condition reads:
\begin{align*}
    (N \mu(f;N)-1)\mathbb{\hat{E}}[A(s^\star(z);z)]-\left(N-1\right)\mathbb{\hat{E}}[x^\star(z)]=0.
\end{align*}
If the left-hand side is positive (negative), the equilibrium exhibits over-(under-)specialization.
Notice that $0<\mathbb{\hat{E}}[x^\star(z)] = \frac{\mathbb{\hat{E}} \left[f(z)\mathbb{E}_{\max{\tilde{z}}|\tilde{z}<z} [A(s^\star(\tilde{z});\tilde{z})]\right]}{\Xi(N)} < \frac{\mathbb{\hat{E}} \left[A(s^\star(\tilde{z});\tilde{z})\right]}{\Xi(N)}$, where $\Xi(N) \equiv 1-\frac{1-f}{f}\ln{1-f}\left(N-1\right)$ follows from the equilibrium surplus offered in symmetric equilibrium.

Suppose that $\Xi(N)\mu(N)>1$. Then, the equilibrium exhibits over-specialization if $N>\frac{\Xi(N)-1}{\Xi(N)\mu(N)-1}$, where $\frac{\Xi(N)-1}{\Xi(N)\mu(N)-1}>1$.
\end{proof}

\begin{proof}[Proof of Proposition \ref{prop:efficiency_endog_N}]
    Rearranging equations (\ref{eq_complex}) and (\ref{eff_complex}) yield $f(\boldsymbol{s}^\star)^{N^\star}=\exp\{-1\}$ and $f(\boldsymbol{\Sc})^{\Nc}=\exp\{-(1-\frac{\tilde{w} \ell}{y})\}$. 
    Since product design costs are strictly positive, the equilibrium sourcing capacity is lower than in the constrained-efficient allocation, i.e., $f(\boldsymbol{s}^\star)^{N^\star}<f(\boldsymbol{\Sc})^{\Nc}$. 
    
    Suppose the specialization function is given. Then, the input finding probability is fixed. Hence, the same equations imply that the equilibrium complexity is higher than the constrained-efficient complexity. Suppose that complexity is given. From Theorem \ref{thm:over-specialization}, equilibrium specialization is higher than the constrained-efficient specialization. 
\end{proof}

    \begin{proof}[Proof of Proposition \ref{prop:dyn_endog_N}]
    The marginal welfare effect of equilibrium complexity is:
\begin{align*}
   \pd{\Wc}{N}\bigg|_{N=N^\star}= \  & -\frac{\delta}{[\delta+(1-\delta)f^N]^2}(1-\delta)f^N\ln{f} \left(\frac{f^N}{\delta}N \mathbb{\hat{E}}[A(s(z);z)]\right)+\mu\frac{f^N}{\delta}N \mathbb{\hat{E}}[A(s(z);z)] \\ \nonumber
   & (N\ln{f}+1)-\frac{\psi}{1-Nm\bar{q}}m\bar{q} \\[.2cm]
    = \  &  \mu\frac{f^N}{\delta} \mathbb{\hat{E}}[A(s(z);z)] \left(-\frac{(1-\delta)f^N N\ln{f}}{\delta+(1-\delta)f^N}+N\ln{f}+1\right)-\frac{\psi}{1-Nm\bar{q}}m\bar{q}\\[.2cm] \nonumber
     = \  &  \mu\frac{f^N}{\delta} \mathbb{\hat{E}}[A(s(z);z)] \left(\mu N\ln{f}+1-\frac{\tilde{w}\ell}{y}\right)\\[.2cm]
     \nonumber
     = \  &  \mu(f;N^\star)\frac{f^N}{\delta} \mathbb{\hat{E}}[A(s^\star(z);z)] \left(1-\mu(f;N^\star)-LS^\star\right)\\[.2cm]
    \lesseqgtr \ & 0,
\end{align*}
%
where the last equality follows from substituting the equilibrium complexity condition (\ref{eq_complex}).
Hence, complexity is efficient for given specialization if and only if $\mu(f;N^\star)=1-LS^\star$. 
    From Proposition \ref{prop:eff_dynam_ss}, specialization is efficient for given complexity if and only if $(N\mu(f;N^\star)-1)\mathbb{\hat{E}}[A(s^\star(z))]+x_0$. If the two efficiency conditions hold simultaneously, the endogenous complexity equilibrium is therefore efficient. \\
    Since $f(\boldsymbol{s}^\star)^{N^\star}=\exp\{-1\}$ and $f(\boldsymbol{\Sc})^{\Nc}=\exp\left\{-\frac{1}{\mu}(1-{LS})\right\}$, the equilibrium features under-resilience if and only if ${LS}^\star>1-\mu^\star$.
\end{proof}

\begin{proof}[Proof of Proposition \ref{prop:robustness_efficiency}]
  The marginal welfare effect of equilibrium robustness is:
\begin{align*}
   \pd{\Wc}{r}\bigg|_{r=r^\star}= \  & -\frac{f^N}{[\delta+(1-\delta)f^N]^2}\frac{f^N}{\delta}N\left( \mathbb{\hat{E}}[A(s(z);z)]-\kappa(r)\right)\delta^\prime(r)-\mu\frac{f^N}{\delta^2}N \left( \mathbb{\hat{E}}[A(s(z);z)]-\kappa(r)\right)\delta^\prime(r) \\ \nonumber
   & -\mu\frac{f^N}{\delta}N \kappa^\prime(r) \\[.2cm]
    = \  &  \mu\frac{f^N}{\delta^2}N \left( \mathbb{\hat{E}}[A(s(z);z)]-\kappa(r)\right)\left(\frac{f^N}{\delta+(1-\delta)f^N}-1\right) \delta^\prime(r)-\mu\frac{f^N}{\delta}N \kappa^\prime(r)\\[.2cm] \nonumber
     = \  & -\mu\frac{f^N}{\delta}N \left[\mu(1-f^N)\frac{\delta^\prime(r)}{\delta(r)}\left( \mathbb{\hat{E}}[A(s(z);z)]-\kappa(r)\right)+\kappa^\prime(r)\right] \\[.2cm]
     \nonumber
     = \  & -\mu(f,r^\star;N)\frac{f^N}{\delta(r^\star)}N \frac{\delta^\prime(r^\star)}{\delta(r^\star)}\left[\mu(f,r^\star;N)(1-f^N)-\tilde{\zeta}\right] \\[.2cm]
    \lesseqgtr \ & 0,
\end{align*}
where $\tilde{\zeta} \equiv \frac{\mathbb{\hat{E}}[x^\star(z)-\kappa(r^\star)]}{\mathbb{\hat{E}}[A(s^\star(z);z)-\kappa(r^\star)]}$ is the aggregate surplus share accruing to final producers, the second equality follows from substituting for $\mu\frac{f^N}{\delta}=\frac{f^N}{\delta+(1-\delta) f^{N}}$, and the last equality follows from substituting the equilibrium robustness condition (\ref{eq_r}).

For given aggregate specialization, constrained-efficient robustness obtains if and only if: $\tilde{\zeta} = (1-f^N)\mu(f,r^\star;N)$. 
Equilibrium specialization is constrained-efficient if and only if: $(N\mu(f,r^\star;N)-1)\mathbb{\hat{E}}[A(s^\star(z);z)-\kappa(r^\star)]+[x_0-\kappa(r^\star)]=0$. The derivation follows from the proof of Proposition \ref{prop:eff_dynam_ss} upon defining match surplus and surplus offered as net of robustness costs.
\end{proof}

\begin{proof}[Proof of Proposition \ref{prop:eff_dynam_nas}]
The marginal welfare effect of equilibrium specialization reads:
\begin{align*}
   \pd{\Wc}{s(z)}\bigg|_{s(z)=s^\star(z)}= \  & \theta \lambda \Dc(z;N) \bigg[A^\prime(s^\star(z))+\frac{\phi^\prime(s^\star(z))}{\phi(s^\star(z))}\bigg(A(s^\star(z);z)  - f(z)\mathbb{E}_{\max\{\tilde{z}\}|\tilde{z} \leq z}[A(s^\star(\tilde{z});\tilde{z})] \\ \nonumber
    & 
      + (N-1)\left(1-f(z)\right)\mathbb{\hat{E}}[A(s^\star(z);z)]  -\frac{1-\delta}{\delta+(1-\delta)f^N}f^N N 
    \\ \nonumber
    & \left(1-f(z)\right)\mathbb{\hat{E}}[A(s^\star(z);z)]\bigg)\bigg]- \frac{\psi \  q^\prime(s^\star(z))}{1-N m \bar{q}}+(N-1)\frac{\theta}{\delta}f^{N-1}\mathbb{E}_{\max\{\tilde{z}\}}\left[\frac{\partial \tilde{\boldsymbol{A}}(\boldsymbol{s}^\star_{-j};\tilde{z})}{\partial  s_j(\tilde{z})}\right]\\[.2cm]
    = \  & \theta \lambda \Dc(z;N)\frac{\phi^\prime(s^\star(z))}{\phi(s^\star(z))} \left(1-f(z)\right) \bigg(\left(\mu(f;N) N-1\right)\mathbb{\hat{E}}[A(s^\star(z);z)]+x_0 \\
    & +\frac{(N-1)}{\lambda \phi^\prime(s^\star(z))(1-f)}\mathbb{E}_{\max\{\tilde{z}\}}\left[\frac{\partial \tilde{\boldsymbol{A}}(\boldsymbol{s}^\star_{-j};\tilde{z})}{\partial  s_j(\tilde{z})}\right]\bigg) \\[.2cm] \nonumber
    \lesseqgtr \ & 0,
\end{align*}
where the multiplier on the last term (production externality) follows from $e^{-\lambda\hat{\phi}(z,\bar{z})}(1-f(z))=1-f$.
\end{proof}

\begin{proof}[Proof of Proposition \ref{prop:eff_dynam_fec}]
By rearranging (\ref{eff_x_fec}), it follows that the equilibrium is constrained-efficient if and only if:
\begin{align*}
       \theta \lambda f^{N-1}\left(N-1\right)\mathbb{\hat{E}}[A(s(x^\star(z));z)-x^\star(z)]= & \ \frac{\psi}{1-N m \bar q}\frac{N \bar{q}}{f^\prime(s^\star(z))} 
       \frac{1}{m}\bigg(\frac{\partial m(s^\star(z),x^\star(z))}{\partial s(z)}+ \\[.1cm]
       & \ \frac{\partial m(s^\star(z),x^\star(z))/\partial x(z)}{\partial s^\star(z)/\partial x(z)}\bigg) \quad \forall z \in [\underline{z},\bar{z}].
   \end{align*}
Since the left-hand side is constant across firms and the right-hand side is firm-specific, the economy with free entry is always inefficient.
\end{proof}

\end{document}